\documentclass{article}

\input{setup.sty}

\usepackage[fontsize=10.5pt]{scrextend}

\title{\boldmath 
Entropy and Spectrum of Near-Extremal Black Holes: \\
semiclassical brane solutions to non-perturbative problems
}

\author{Sergio Hern\'andez-Cuenca}

\affiliation{Center for Theoretical Physics, Massachusetts Institute of Technology, Cambridge, MA 02139, USA}

\emailAdd{sergiohc@mit.edu}

\preprint{MIT-CTP/5741}

\abstract{\\~\\\indent
The black hole entropy has been observed to generically turn negative at exponentially low temperatures $T\sim e^{-S_0}$ in the extremal Bekenstein-Hawking entropy $S_0$, a seeming pathology often attributed to missing non-perturbative effects. In fact, we show that this negativity must happen for any effective theory of quantum gravity with an ensemble description.
To do so, we identify the usual gravitational entropy as an annealed entropy $S_a$, and prove that this quantity gives $S_0$ at extremality if and only if the ground-state energy is protected by supersymmetry, and diverges negatively otherwise. The actual thermodynamically-behaved quantity is the average or quenched entropy $S_q$, whose calculation is poorly understood in gravity: it involves replica wormholes in a regime where the topological expansion breaks down. Using matrix integrals we find new instanton saddles that dominate gravitational correlators at $T\sim e^{-S_0}$ and are dual to semiclassical wormholes involving dynamical branes. These brane solutions give the leading contribution to any black hole very near extremality, and a duality with matrix ensembles would not make sense without them. 
In the non-BPS case, they are required to make $S_q$ non-negative and also enhance the negativity of $S_a$, both effects consistent with matrix integrals evaluated exactly. Our instanton results are tested against the on-shell action of D3-branes dual to multiply wrapped Wilson loops in $\mathcal{N}=4$ super-YM, and a precise match is found. Our analysis of low-energy random matrix spectra also explains the origin of spectral gaps in supersymmetric theories, not only when there are BPS states at zero energy, but also for purely non-BPS supermultiplets. In the former, our quantitative prediction for the gap in terms of the degeneracy of BPS states agrees with the R-charge scaling in gapped multiplets of $\mathcal{N}=2$ super-JT gravity.
}

\begin{document}


\renewcommand{\baselinestretch}{1.1}\selectfont
\setlength{\parskip}{0.15\baselineskip}%

\maketitle

\renewcommand{\baselinestretch}{1.19}\selectfont
\setlength{\parskip}{0.35\baselineskip}%

\section{Introduction}
\label{sec:intro}

The Euclidean gravitational path integral is generally understood to compute a canonical partition function, an identification that provides a powerful window into black hole thermodynamics \cite{Hawking:1975vcx}.
Given a black hole at inverse temperature $\beta$, the gravitational partition function $\mathcal{Z}(\beta)$ can be used to obtain an entropy via the standard thermodynamic identity
\begin{equation}
\label{eq:slogZ}
	\mathcal{S}(\beta) = (1 - \beta \partial_\beta) \log \mathcal{Z}(\beta).
\end{equation}
The insights that this relation and generalizations thereof have yielded into the quantum gravitational physics of black holes cannot be overstated \cite{Ryu:2006bv,Hubeny:2007xt,Lewkowycz:2013nqa,Faulkner:2013ana,Engelhardt:2014gca,Penington:2019npb,Almheiri:2019psf,Penington:2019kki,Almheiri:2019qdq}.
However, that \cref{eq:slogZ} actually computes the thermal entropy of a standard quantum system is a statement that can only be established reliably via Euclidean methods in field theory, not in gravity. 
Hence that $\mathcal{S}$ always is and should behave as such an entropy is an assumption which may lead to drawing potentially incorrect physical conclusions. For instance, one may expect $\mathcal{S}(\beta)$ to be non-negative for all $\beta$, and argue that some contributions may be missing from $\mathcal{Z}(\beta)$ if $\mathcal{S}(\beta)$ happens to turn negative for some $\beta$. In fact, as first realized by \cite{Engelhardt:2020qpv}, quite generally $\mathcal{S}$ is not and should not be expected to behave like a thermal entropy in effective quantum gravity.

In this paper, we refer to the quantity $\mathcal{S}$ that \cref{eq:slogZ} computes as the gravitational entropy.
As will be explained, generically $\mathcal{S}(\beta)$ must behave drastically different from a thermal entropy at large $\beta$, a regime which gravitationally corresponds to near-extremal black holes. Crucially, this means that to reliably study the statistical physics of black holes near extremality one must revisit \cref{eq:slogZ} in pursuit of two main goals:
\begin{enumerate}
	\item\label{pt:whats} Identify precisely what kind of entropy $\mathcal{S}$ is and what its properties are near extremality.
	\item\label{pt:quens} Understand how to compute thermal entropies in gravity and, in particular, establish whether semiclassical gravity alone can succeed in describing black holes very near extremality.
\end{enumerate}
A brief preview of our findings in addressing these goals is as follows:
\begin{enumerate}
	\item\label{pt:ans1} This $\mathcal{S}$ is an annealed entropy, which provably must diverge negatively in the extremal limit unless supersymmetry protects the ground-state energy. This phenomenon is unrelated to the degeneracy of BPS states, which we quantitatively relate to the size of spectral gaps.
	\item\label{pt:ans2} Thermal entropies which behave non-negatively are computed by a replica trick, which is successfully implemented by novel near-extremal saddle points. Gravitationally, these are identified as semiclassical solutions for dynamical branes which dominate on-shell near extremality and provide a non-perturbative resummation of off-shell topological expansions.
\end{enumerate}

The introductory sections that follow contextualize these issues, explain the basic notions needed to tackle them, and summarize our main results. For the sake of conciseness, we keep the presentation here minimal, and relegate the discussion of additional details to appendices.

\subsection{Effective Gravity Ensembles}

It is expected that a quantum gravitational system, such as a black hole, should be describable from the outside as a standard, unitary quantum system with a discrete spectrum \cite{Strominger:1996sh,Banks:1996vh,Maldacena:1997re,Gubser:1998bc,Witten:1998qj,David:2002wn,Page:2004xp}.
This is well motivated in a complete theory of quantum gravity; however, an effective theory need not adhere to such an expectation.
Indeed, it has been observed that effective quantum gravity does not behave like any single quantum system, but like an ensemble thereof. This correspondence has been realized very explicitly in low-dimensional models of gravity 
\cite{Saad:2019lba,Stanford:2019vob,
Belin:2020hea,Afkhami-Jeddi:2020ezh,Maloney:2020nni,Cotler:2020ugk,Benjamin:2021wzr,Collier:2021rsn,Chandra:2022bqq,Collier:2023fwi,Collier:2024mgv}, and argued for not only in effective gravity \cite{Coleman:1988cy,Giddings:1988cx,Giddings:1988wv,Hawking:1991vs,Bousso:2019ykv,Marolf:2020xie,Bousso:2020kmy,Johnson:2022wsr}, but also in general effective theory \cite{Hernandez-Cuenca:2024pey}.

A simple class of observables which illustrates this phenomenon consists of gravitational correlation functions involving multiple asymptotic boundaries. 
Denote the gravitational path integral by $\mathcal{P}$ and consider imposing thermal boundary conditions $Z(\beta)$ at different $\beta$. Then for an $m$-point function one finds
\begin{equation}
\label{eq:Zmeval}
    \mathcal{P}(Z(\beta_1)Z(\beta_2)\cdots Z(\beta_m)) = \expval{Z(\beta_1)Z(\beta_2)\cdots Z(\beta_m)},
\end{equation}
where on the right-hand side $\expval{\,\cdot\,}$ is an average over an ensemble of non-gravitational quantum theories with canonical partition functions $Z(\beta)$ with the corresponding $\beta$ values.\footnote{\,\label{fn:trivem}No assumption is actually being made in writing \cref{eq:Zmeval}: this expression makes sense even if the ensemble trivially consists of a single representative, in which case the right-hand side would factorize into $\expval{Z(\beta_k)}$ terms.}
For instance, the object $\mathcal{Z}$ appearing in \cref{eq:slogZ} corresponds to
\begin{equation}
\label{eq:Zsinlge}
    \mathcal{Z}(\beta) \equiv \mathcal{P}(Z(\beta)) = \expval{Z(\beta)}.
\end{equation}
The emergence of an ensemble description may be attributed to the breaking of factorization on the gravity side.
This occurs because $\mathcal{P}$ allows for contributions from all topologies with the requisite boundary conditions, including multi-boundary wormholes connecting different boundaries.

Given that the gravitational path integral provides only an effective description of quantum gravity, there is also a simple heuristic reason why such an ensemble description ought to arise: if an effective theory admits multiple microscopic completions, its predictions should be consistent with but maximally ignorant of the microscopics of any one of them \cite{Johnson:2022wsr,deBoer:2023vsm,Hernandez-Cuenca:2024pey}. In other words, as soon as one begins to ignore details about a complete theory, all that an effective theory may have access to is expectation values capturing the statistics across the ensemble of all possible completions of the theory.\footnote{\,In path integral formulations, what causes these non-trivial statistics is strongly non-local correlations that are ubiquitous in effective theory, and which in gravity take the form of geometric wormholes \cite{Hernandez-Cuenca:2024pey}. The possibility of materializing these heuristics top-down for \cref{eq:Zmeval}, however, seems harder to justify than for \cref{eq:Zsinlge}.} Hereon, we thus take the viewpoint that effective quantum gravity behaves like an ensemble of quantum theories in the sense of \cref{eq:Zsinlge}.

Since a canonical partition function is given in terms of the Hamiltonian $H$ of a theory by
\begin{equation}
\label{eq:trbh}
    Z(\beta) \equiv \Tr  e^{-\beta H},
\end{equation}
one can always think of an ensemble of theories as specified by an ensemble of Hamiltonians. In turn, since $Z(\beta)$ actually only depends on the energy spectrum of $H$, one can generally focus on studying ensembles of spectra and think of $\expval{\,\cdot\,}$ as computing expectation values in them.\footnote{\,A natural example consists of having a space of theories parameterized by some set of coupling variables.
Then a probability distribution over these couplings defines an ensemble of theories, which can equivalently be expressed as a probability distribution over the corresponding Hamiltonians. By diagonalizing these into energy eigenvalues, what one ends up with is a probability distribution over the energies across the ensemble of all theories.} All in all, correlations across theories are seen to induce correlations across their energy levels, the statistics of which is what spectral ensembles capture. In these ensembles, the discrete energy levels of each theory combine into probabilistic distributions which are generically continuous. The drastic consequences this phenomenon has on entropic quantities cannot be overlooked, and one of the goals of this paper is to characterize these effects.

\subsection{Annealing and Quenching}
\label{ssec:annque}

It is important to carefully identify which quantities are actually of interest when faced with an ensemble of theories.
Consider an observable which for the partition function $Z$ of a single theory is computed by $F(Z)$.
In an ensemble, there are two natural classes of quantities one may compute:
\begin{itemize}
	\item \textit{Annealed}: These involve averaging $Z$ over theories, then evaluating $F$ on just $\expval{Z}$:
	\begin{equation}
        \label{eq:afZ}
		F_a(Z) \equiv F(\,\expval{Z}\,).
	\end{equation}
	\item \textit{Quenched}: These involve evaluating $F$ on every $Z$, then averaging $F(Z)$ over theories:
	\begin{equation}
        \label{eq:qfZ}
		F_q(Z) \equiv \expval{F(Z)}.
	\end{equation}
\end{itemize}

By definition, $F_q(Z)$ will clearly preserve every property of $F(Z)$ that the linearity of the expectation value respects.
In contrast, nothing guarantees that $F_a(Z)$ should preserve any such property of $F(Z)$, and thus in general it will not.
For instance, if $F(Z)\geq0$ always then clearly $F_q(Z)\geq0$ as well, but one should not be conflicted if $F_a(Z)$ happens to take negative values.

We wish to apply the distinction between annealing and quenching to the thermal entropy. In a single quantum theory with canonical partition function $Z(\beta)$, this is given by\footnote{\,Cf. \cref{eq:slogZ}, where calligraphy was used to highlight objects computed by the gravitational path integral.}
\begin{equation}
\label{eq:sbdef}
    S(\beta) \equiv (1 - \beta \partial_\beta) \log Z(\beta).
\end{equation}
Hence, in an ensemble of theories, \cref{eq:afZ} defines the annealed entropy as
\begin{equation}
\label{eq:annS}
    S_a(\beta) \equiv (1 - \beta \partial_\beta) \log \expval{Z(\beta)},
\end{equation}
whereas \cref{eq:qfZ} defines the quenched entropy as
\begin{equation}
\label{eq:quenS}
    S_q(\beta) \equiv (1 - \beta \partial_\beta) \expval{\log Z(\beta)}.
\end{equation}
Comparing \cref{eq:slogZ} with \cref{eq:annS}, it immediately follows by \cref{eq:Zsinlge} that the gravitational entropy is an annealed entropy \cite{Engelhardt:2020qpv}. We emphasize that the identification of $\mathcal{S}$ with $S_a$ also holds for a single theory, and thus makes no assumption about effective gravity generally having a non-trivial ensemble description (cf. \cref{fn:trivem}). In other words, the gravitational entropy can always be said to be an annealed entropy. This realization is the starting point for addressing goal (\ref{pt:whats}) above, and will allow us to make concrete predictions about how $\mathcal{S}(\beta)$ must behave as a function of $\beta$ in effective quantum gravity. For instance, we prove that it generically diverges negatively as the temperature goes to zero, implying that the gravitational entropy of generic black holes must become negative  near extremality. Regarding goal (\ref{pt:quens}), as the reader may suspect, the entropic quantity with the desired thermodynamic behavior is the quenched entropy.

The discrepancy between $S_a$ and $S_q$ clearly boils down to the difference between annealing and quenching the logarithm function. It will thus at times be useful to consider the free energy,
\begin{equation}
\label{eq:freedef}
    F(\beta) \equiv - \frac{1}{\beta} \log Z(\beta),
\end{equation}
a physical quantity that is linear in $\log Z(\beta)$ and in terms of which \cref{eq:sbdef} simply reads
\begin{equation}
\label{eq:sfree}
    S(\beta) = \beta^2 F'(\beta).
\end{equation}
Clearly $S$ is non-negative if and only if $F$ is monotonically non-decreasing, which means that negativity of the annealed entropy is equivalent to non-monotonicity of the annealed free energy.
In a single theory, the condition $F'(\beta)\geq0$ is equivalent to $\log Z(\beta) \geq \beta \partial_\beta \log Z(\beta)$, which gives the physical requirement $F(\beta)\leq\expval{E}_{\beta}$, where $\expval{E}_{\beta}$ is the statistical average energy of the thermal system. Our results thus imply that this property is violated when annealing.

\subsection{The No-Replica Trick}

Annealed quantities are usually easier to calculate, since their linearity in the partition function means that standard path integral techniques are applicable.
In contrast, quenching is generally non-trivial and has to be addressed case by case depending on the pertinent observable.
This is why annealing is sometimes used as a proxy for quenching. However, the success of such an approximation relies on whether the observable at hand happens to be self-averaging, a property which is \`a priori hard to ascertain. For instance, we will see that annealed entropies reliably approximate quenched ones at high temperatures, but radically disagree at low temperatures.

Given the possibility of discrepancy, in the ever surprising context of quantum gravity it would seem strictly more desirable to always consider quenched quantities.\footnote{\,Note that if the ensemble is trivial then annealing and quenching are obviously identical operations. Hence if uncertain whether one is dealing with an ensemble of theories or whether a quantity is self-averaging, it is always safer to do a quenched calculation, if possible.} 
Unfortunately though, except for the recent attempts of \cite{Engelhardt:2020qpv,Chandrasekaran:2022asa}, the tools for quenching quantities such as the entropy directly within quantum gravity remain completely mysterious.
The difficulties one faces when attempting such a computation are conceptual, technical, and one may worry potentially fundamental. A central corollary of this paper in regard to goal (\ref{pt:quens}) is that, in fact, no fundamental obstruction seems to exist: effective gravity does have the requisite ingredients for calculating quenched entropies.

To understand the difficulties, note that given an ensemble of theories with some measure, the calculation in \cref{eq:quenS} is conceptually clear even if technically involved. In gravity, however, an object like $\mathcal{P}(\log Z(\beta))$ cannot even be directly defined. This is because there is no separate notion of an ensemble of theories and its measure; the gravitational path integral defines them all at once.
This hurdle can be overcome by means of a well-known replica trick for the logarithm function,
\begin{equation}
\label{eq:reptrick}
    \expval{\log Z(\beta)} = \lim_{m\to0}  \frac{\expval{Z(\beta)^m}-1}{m} = \lim_{m\to0} \frac{d}{dm} \expval{Z(\beta)^m},
\end{equation}
the gravitational implementation of which was pioneered by \cite{Engelhardt:2020qpv}. The no-replica limit $m\to0$ in \cref{eq:reptrick} is very different from the limit $m\to1$ in the standard replica trick for the von Neumann entropy, which is already well-understood in gravity \cite{Lewkowycz:2013nqa}.
The two identities in \cref{eq:reptrick} can be used interchangeably, but consistency clearly requires both limits to exist and agree. The intermediate expression emphasizes that $\expval{\,\cdot\,}$ must be a unit-normalized expectation value.

In gravity, the computation of the moments $\expval{Z(\beta)^m}$ in \cref{eq:reptrick} can be performed by the gravitational path integral $\mathcal{P}(Z(\beta)^m)$ with identical boundary conditions $Z(\beta)$ on $m$ disjoint boundaries.
For integer $m\geq1$ these are well-defined $m$-point gravitational correlators like those in \cref{eq:Zmeval}. The replica trick in \cref{eq:reptrick}, however, requires an analytic continuation to real $m>0$. Ambiguities in this continuation posed technical difficulties in \cite{Engelhardt:2020qpv}, and a semiclassical resolution was proposed in \cite{Chandrasekaran:2022asa}. Nonetheless, the need to resort to numerics still rendered the gravitational implementation of this method inconclusive. There thus remained the troubling possibility that effective gravity would not even be able to compute quenched entropies at all.

In this work, we address this question by making use of the duality between theories of Jackiw-Teitelboim (JT) gravity \cite{Jackiw:1984je,Teitelboim:1983ux} and matrix integrals \cite{Saad:2019lba,Stanford:2019vob,Turiaci:2023jfa,Mertens:2022irh}, in which the ensembles of theories that emerge correspond to a random matrix Hamiltonian. In particular, this duality is established at the level of a topological expansion of the gravitational correlators on the left-hand side of \cref{eq:Zmeval} and a corresponding perturbative expansion of a double-scaled matrix integral for the right-hand side. While this correspondence is a precise match to all orders in the topological expansion, the random matrix formulation in general has access to non-perturbative details that may not be accessible from the gravity side.\footnote{\,Actually, even the topological expansion is only an asymptotic series whose convergence also depends on $\beta$. On the other hand, the matrix integral is defined exactly at any $\beta$. Hence the matrix dual not only provides a non-perturbative completion, but also a convergent resummation of the topological expansion of JT for all $\beta$.} Our goal in this paper is to study the replica trick in \cref{eq:reptrick} using the concrete framework of random matrix theory, but making sure to only rely on ingredients which are accessible to the gravitational path integral.\footnote{\,\label{fn:prevwk}\,Explorations of this replica trick using matrix integrals have been previously pursued, but using numerical or non-perturbative tools which do not help elucidate how to implement it in gravity \cite{Johnson:2020mwi,Okuyama:2020mhl,Okuyama:2021pkf,Janssen:2021mek,Johnson:2021rsh,Johnson:2021zuo,Johnson:2022wsr}.}
As a result, we are able to pinpoint the minimal ingredients necessary to succeed in this task, singling out the key role of a specific class of branes which effective gravity can indeed accommodate.

\subsection{Gravitational Entropy}
\label{ssec:gsint}

While the subject of quenched entropies in gravity is very recent and underexplored, extensive work has been devoted to calculating the annealed entropy $\mathcal{S}$ that \cref{eq:slogZ} gives.
It is of particular interest to study $\mathcal{S}(\beta)$ for black holes near extremality, corresponding to the parametric hierarchy $\beta\gg S_0 \gg 1$. This regime can be explored by looking at the large-$\beta$ behavior of $\mathcal{S}(\beta)$ on top of a large-$S_0$ perturbative expansion.
The computation of this near-extremal entropy including quantum gravitational effects is a formidable endeavor which has been pursued by \cite{Sen:2012kpz,Sen:2012cj,Sen:2012dw,Bhattacharyya:2012ye,Castro:2018hsc,Ghosh:2019rcj,Iliesiu:2020qvm,Heydeman:2020hhw,Boruch:2022tno,Iliesiu:2022kny,Iliesiu:2022onk,Banerjee:2023quv,H:2023qko,Turiaci:2023wrh,Rakic:2023vhv,Kapec:2023ruw}. We note that all of their results can be captured by the following general expression:\footnote{A more detailed discussion about this result and its subtleties when $\beta\sim O(e^{S_0})$ is given in \cref{sec:nebh}.}
\begin{equation}
\label{eq:genstt}
    \mathcal{S}(\beta) = S_0 - s\log\beta + \# \, \Delta e^{-\beta \Delta} + \# \beta^{-1} e^{-\beta \Delta} + O(\beta^{-2} e^{-\beta \Delta}),
\end{equation}
where we are using $\#$ to denote $O(1)$ numerical factors whose precise value will not play a role in our discussion. 
Here, $S_0$ is the quantum-corrected Bekenstein-Hawking black hole entropy at extremality,\footnote{\,More explicitly, $S_0 = S_{\smalltext{BH}} + c_{\smalltext{log}} \log S_\smalltext{BH}$, where $S_\smalltext{BH} \equiv \frac{A}{4G_{\smalltext{N}}}$ is the extremal Bekenstein-Hawking entropy given in terms of the horizon area $A$, and $c_{\smalltext{log}}$ is a constant which depends on the field content of the theory and quantifies their $1$-loop contributions to these
leading logarithmic corrections in $S_{\smalltext{BH}}$. See \cref{sec:nebh} for more details.} $s\in\mathbb{Q}$ is an important parameter determined by supersymmetry, and $\Delta\geq0$ is the size of a possible gap in the spectrum above the ground state. In particular, one observes that $s=0$ when the extremal black hole is BPS \cite{Heydeman:2020hhw} and strictly $s>0$ otherwise \cite{Iliesiu:2020qvm}; e.g. theories which are purely bosonic yield $s=3/2$.
As for $\Delta$, in general one finds $\Delta>0$ if and only if the extremal limit is BPS and has an extensive ground-state degeneracy $N_0 \sim O(e^{S_0})$; otherwise, $\Delta=0$ when $N_0 \sim O(1)$.
Summarizing, \cref{eq:genstt} is found to involve:
\begin{equation}
\label{eq:allcases}
    \begin{cases}
        \text{Non-near-BPS:} & s>0,\qquad\qquad\qquad\qquad~ \Delta = 0, \\
        \text{Near-BPS:} & s=0,~ 
        \begin{cases}
            N_0\sim O(1): & \Delta = 0, \\
            N_0\sim O(e^{S_0}): & \Delta > 0,
        \end{cases}
    \end{cases}
\end{equation}
where we are using the term near-BPS to refer to near-extremal black holes which preserve some supersymmetry in the extremal limit. Note that at this point \cref{eq:allcases} is a heuristic classification obtained by inspecting the form of \cref{eq:genstt} for various gravitational results. In this paper, we show that \cref{eq:allcases} holds in full generality, and explain the quantitative origin of both $s$ and $\Delta$.

Focus for now on the differences between the near-BPS ($s=0$) and non-near-BPS ($s>0$) cases. The behavior of $\mathcal{S}(\beta)$ in \cref{eq:genstt} is drastically different at large $\beta$ depending on supersymmetry due to the $-s\log\beta$ term. For near-BPS black holes, the $\beta\to\infty$ limit simply yields a finite extremal entropy given by $S_0$. Generically, this is a large number interpreted as quantifying the ground-state degeneracy by $e^{S_0}$, a satisfying result since one does expect a large number of degenerate BPS states based on symmetry arguments. That this extremal entropy indeed has a statistical interpretation has been verified by direct enumeration of BPS states in string theory \cite{Strominger:1996sh,David:2002wn}.
On the other hand, for black holes which are not near-BPS, we observe $\mathcal{S}$ becoming smaller with $\beta$ and eventually turning negative for $\beta \gtrsim O(e^{S_0})$. The monotonic decrease of $\mathcal{S}$ towards low temperatures has been interpreted positively, considering that in the absence of a symmetry arguments non-BPS black holes should not be highly degenerate \cite{Iliesiu:2020qvm,Iliesiu:2022onk}. As for the eventual negativity of $\mathcal{S}(\beta)$, it has been argued that for such large $\beta \gtrsim O(e^{S_0})$ some other quantum gravitational object must dominate $\mathcal{Z}(\beta)$, giving non-perturbative corrections to \cref{eq:genstt} which would take over and keep $\mathcal{S}(\beta)$ non-negative \cite{Turiaci:2023wrh,Rakic:2023vhv}. Statistically, this would effectively mean that non-near-BPS black holes at such low temperatures do not exist. Our findings in this paper suggest otherwise.

The qualitative features of \cref{eq:genstt} that follow from \cref{eq:allcases} will serve all along this paper as both an example and a consistency check for many results. Let us demonstrate this with a brief preview of the general picture that arises from our spectral analysis. What turns out to characterize near-BPS black holes is the fact that the ground-state energy is protected by supersymmetry across the ensemble, whereas for black holes which are not near-BPS this energy is allowed to vary. As we will show, this simple property of a spectral ensemble alone already implies that the annealed entropy should capture the degeneracy of the ground-state energy when protected, and diverge negatively when not. Hence the behavior \cref{eq:genstt} exhibits is perfectly consistent with $\mathcal{S}$ being an annealed entropy, including the fact that it becomes negatively divergent in the $\beta\to\infty$ limit. In other words, any potential non-perturbative correction to \cref{eq:genstt} should not make $\mathcal{S}$ non-negative, and nonetheless nothing compromises the existence of such non-near-BPS black holes.

Another important fact we will prove is that how degenerate the ground state is actually plays no role on the negativity issue from the spectral point of view; all that matters is whether its energy is protected. 
This is already suggested by \cref{eq:allcases}, where we see that the sufficient condition for non-negativity $s=0$ holds in the near-BPS case regardless of whether $N_0$ is extensive in $e^{S_0}$ or not. 
In particular, we will show that the annealed entropy asymptotes to $\log N_0$ at extremality so long as the ground-state energy is protected, no matter how small or large $N_0$ is. On the other hand, we prove that this entropy becomes negative at low temperatures if the ground-state energy is not protected, regardless of how degenerate it is.\footnote{ 
As a simple toy example, suppose that an ensemble of theories has a ground-state energy uniformly distributed over $[-\delta,0]$ with degeneracy $e^{S_0}$, and a density of states for all higher energies $E>0$ of the form $\rho(E)=e^{S_0} \sqrt{E}$. The average canonical partition function of such an ensemble is
$$
\expval{Z(\beta)} = (e^{\beta\delta}-1) \frac{e^{S_0}}{\beta\delta} + \frac{\sqrt{\pi}e^{S_0}}{2\beta^{3/2}}.
$$
One easily verifies that the strict $\delta\to0$ case, corresponding to a ground-state energy protected to $E=0$ across the ensemble, yields $S_a(\beta)\to S_0$ as $\beta\to\infty$. 
Note that this fact is independent of how large or small $S_0$ is, and that $S_a(\beta)\geq0$ in this case even though there is no gap in the spectrum.
In contrast, for any $\delta>0$ one finds $S_a(\beta)\to-\infty$ as $\beta\to\infty$, no matter how large $e^{S_0}$ is; in this toy model, the negativity of the annealed entropy occurs for $\beta\gtrsim e^{S_0}/\delta$.}
Black holes in supersymmetric theories which break all supersymmetry will nicely illustrate this fact.
Hence what makes the behavior of the gravitational entropy of near-BPS black holes so different is the fact that the ground-state energy is protected, not that it is highly-degenerate.
The characteristic feature that a ground state with high degeneracy does give rise to, as will be explicitly seen, is the appearance of a gap in the spectrum above its energy. In other words, the $\Delta$ gap in \cref{eq:genstt} does arise precisely because the ground state is highly degenerate, and would be zero otherwise (cf. \cref{eq:allcases}).

\subsection{Summary of Results}

As explained in \cref{ssec:annque}, the gravitational entropy $\mathcal{S}$ computed in effective quantum gravity through \cref{eq:slogZ} is an annealed entropy. Using this fact, our results for goal (\ref{pt:whats}) are as follows:
\begin{enumerate}[itemindent=2em,labelsep=.5em,leftmargin=0pt]
    \item The standard gravitational path integral gives a black hole entropy $\mathcal{S}$ which is non-negative at low temperatures if and only if the extremal limit is supersymmetric, and diverges negatively otherwise. This is shown to be an exact, non-perturbative statement for any effective theory of gravity with an ensemble description.
    This result turns the generic negativity of $\mathcal{S}$ from a pathological behavior that ought to be fixed by non-perturbative corrections into a required feature which should be robust against non-perturbative corrections.
    \item If the extremal limit is a BPS black hole, the usual gravitational entropy stays non-negative and asymptotes to the Bekenstein-Hawking entropy. Nonetheless, away from the strict extremal limit, this quantity continues to behave unlike the thermal entropy of any quantum system. We identify the supersymmetry protection of the ground-state energy as the sole mechanism behind the qualitative negativity difference between the BPS and non-BPS limits; whether or not the ground state is highly degenerate also has consequences which are interesting but unrelated.
\end{enumerate}

These results follow primarily from a general analysis of spectral ensembles in \cref{sec:specen}. This provides the framework needed to address the differences between annealed and quenched entropies given an ensemble of theories, and make very general predictions about the behavior of these entropies in effective quantum gravity. These predictions can be compared with explicit gravitational results for near-extremal black holes that we review in \cref{sec:bhspec} and \cref{sec:nebh}. The consequences of supersymmetry depending on whether the ground-state energy is protected across Hamiltonians in the ensemble are stated in \cref{thm:fact,cor:fact}. The discrepancy between annealed and quenched entropies even with supersymmetry, as well as the actual role of large BPS degeneracy, are more explicitly explored later in \cref{ssec:quni,ssec:gbps,sec:gaps}.

As for the quenched entropy, its calculation in gravity requires computing gravitational correlators of $Z(\beta)$ including wormhole contributions to the gravitational path integral. However, the physics of interest occur at $\beta\sim O(e^{S_0})$, a non-perturbative regime often regarded as inaccessible to effective gravity. Our findings using general ensembles of random matrix theory dual to effective gravity strongly suggest otherwise; in fact, already semiclassical gravity suffices to capture these effects. Using matrix integrals to study $\expval{Z(\beta)^m}$, our results for goal (\ref{pt:quens}) are as follows:
\begin{enumerate}[itemindent=2em,labelsep=.5em,leftmargin=0pt]
    \item For a completely general random matrix ensemble, we find that at $\beta\sim O(e^{S_0})$ a new saddle arises which exists and dominates at large $e^{S_0}$. This saddle is a single-eigenvalue instanton with a clear holographic dual: dynamical brane solutions. The saddle-point treatment of the matrix integral is thus understood to be dual to a semiclassical treatment of effective gravity where brane actions are evaluated on-shell. As a non-trivial test of this claim, we show that our instanton results for $\beta\sim O(e^{S_0})$ identically match the on-shell action of gravitational D3-branes holographically dual to $1/2$ BPS Wilson loops in $\mathcal{N}=4$ supersymmetric Yang-Mills with a random matrix description.
    \item These brane solutions provide a controlled resummation that includes contributions of arbitrary genus to the gravitational path integral in a $\beta \sim O(e^{S_0})$ regime where the usual topological expansion breaks down. In addition, our instantons continuously interpolate between small $\beta \ll O(e^{S_0})$ and large $\beta \sim O(e^{S_0})$, suggesting that semiclassical gravity alone should also be able to capture the transition to a brane regime. For near-extremal black holes, this transition suggests that the near-horizon throat gets cut off by a brane at a finite proper distance as soon as $\beta\sim O(e^{S_0})$ already at the semiclassical level.\footnote{\,As mentioned before, the fact that generically $\mathcal{Z}(\beta)\to0$ as $\beta\to\infty$ (i.e., the cause of $\mathcal{S}(\beta)\to-\infty$), has previously been interpreted as implying that near-extremal black holes that are not near-BPS do not exist \cite{Turiaci:2023wrh,Rakic:2023vhv}. Our results suggest that they do, but develop branes in their throat at $\beta\sim O(e^{S_0})$.}
    \item Our instantons are found to be necessary and sufficient to obtain a non-negative quenched entropy via the no-replica trick, motivating its gravitational implementation including brane replica wormholes. In addition, they capture non-perturbative effects which enhance the generic negative divergence of the annealed entropy at large $\beta$. Our analysis of the replica trick resolves a subtlety that has hindered previous attempts even within matrix integrals:\footnote{See e.g. \cite{Johnson:2020mwi}, where a large-$\beta$ expansion followed by a continuation in $m$ led to a problematic $m\to0$ limit.} the large-$\beta$ and $m\to0$ replica limits directly compete and thus do not commute. This dooms any approach that involves approximating integer $m\geq1$ moments at large $\beta$, and then analytically continuing to $m\to0$.
\end{enumerate}

The general framework of random matrix ensembles and its constructs relevant for the analysis of random Hamiltonians with and without supersymmetry is the subject of \cref{sec:rmf}. In particular, as explained in \cref{ssec:ensemmess,ssec:hamran} (see also \cref{sec:susyrmt}),
our matrix integral approach applies to all matrix ensembles relevant to quantum gravity: Wigner-Dyson (no supersymmetry), Altland-Zirnbauer (broken supersymmetry), and their Wishart generalizations (unbroken supersymmetry). Matrix scalings to spectral edges and double scaling limits are also accounted for (see \cref{ssec:mscal}).

These tools allow us to carry out a general exploration of the near-extremal limit in \cref{sec:canemsem}. Once the relevant matrix dynamics are identified to be eigenvalue instantons, an appropriate continuum formulation of the matrix integral for $\expval{Z(\beta)^m}$ is performed in \cref{ssec:split,ssec:insti,ssec:insact}. The consequences of these instantons for entropic quantities are analyzed in \cref{ssec:quni}. A crucial aspect of our treatment throughout \cref{sec:canemsem} is the restriction to employing solely random matrix theory tools with a clear semiclassical holographic dual.\footnote{\,In particular, we use the replica trick (for otherwise there is no sensible gravitational path integral definition of the quenched entropy), and rely on a saddle-point analysis of the matrix integral (for otherwise we would potentially be borrowing non-perturbative tools such as Fredholm determinants which are inaccessible to gravity, cf. \cref{fn:prevwk}).} This is what allows us to elucidate the key gravitational ingredients at play in \cref{ssec:lessons}. Furthermore, by working at a saddle-point level, the lessons we arrive at make sense also in higher-dimensional gravity, where an off-shell gravitational path integral is not even available. Explicit examples of our instanton results for various matrix ensembles are presented in \cref{sec:egs}, including the cases relevant to $1/2$ BPS Wilson loops in \cref{sec:gwd} and JT gravity in \cref{ssec:jtsec}.

It is worth emphasizing that many of our conclusions benefit from the robustness of universality results in random matrix theory, and thus are independent at a semiclassical level from any potentially non-unique choice of non-perturbative completion of the theory.\footnote{\,The semiclassical brane effects we describe are insensitive to non-perturbative details in the sense that e.g. to leading order any non-perturbative completion would give the same non-negative result for the quenched entropy, whereas not including branes would give a negative entropy near extremality. In terms of the spectrum, these branes simply capture the fact that there is an ensemble of spectra which are discrete; there continues to be no single discrete spectrum effective gravity can access, but an ensemble thereof.} More specifically, we find the dynamics of the relevant near-extremal instantons to be governed by edge statistics which are universal to leading order at large $e^{S_0}$. 
By analyzing universal aspects of the low energy spectrum of general matrix ensembles, we are able to make the following general observations and predictions about the near-extremal black hole spectrum in effective quantum gravity:
\begin{enumerate}[itemindent=2em,labelsep=.5em,leftmargin=0pt]
    \item The spectral statistics of Wigner-Dyson (non-supersymmetric) and Altland-Zirnbauer (supersymmetric) ensembles are often associated to soft Airy edges and hard Bessel edges, respectively. In fact, we find that in the more general Wishart construction for supersymmetric ensembles, the lower edge of the spectrum becomes soft and governed by Airy statistics.
    This is a consequence of the natural scaling of Wishart ensembles, which gives rise to a spectral gap above zero energy. This gap pushes the spectrum off the hard Bessel edge at zero, causing positive eigenvalues to have a soft Airy onset above the gap.
    \item By studying random Hamiltonians with at least $\mathcal{N}=2$ supersymmetry, we identify two distinct physical mechanisms explaining the origin of gaps in the spectrum. 
    For supermultiplets with an extensive degeneracy of BPS states at zero energy $N_0\sim\nu \, e^{S_0}$, the spectrum exhibits a gap $\Delta$ to the first non-BPS state which behaves as $\Delta\sim \nu^2$ for small gaps and as $\Delta\sim\nu$ for large gaps.
    This gap arises due to the Vandermonde repulsion exerted by the BPS states at zero energy.
    For purely non-BPS sectors, we observe that what causes the shift of the spectrum to higher energies for higher supermultiplets is a monotonic growth extensive in $e^{S_0}$ in the number of non-BPS states.
\end{enumerate}

The universal gap behavior we observe is consistent with gravitational results. In particular, the quadratic scaling of $\Delta$ with the ratio $N_0/e^{S_0}$ matches the behavior of the gap in supermultiplets of $\mathcal{N}=2$ super-JT, where $\Delta\sim q^2$ when $N_0\sim q \, e^{S_0}$ (see \cref{sec:bhspec,sec:gaps}). When $N_0>0$ but not extensive in $e^{S_0}$, the spectrum becomes gapless and develops a hard edge again even though a discrete ground state protected by supersymmetry still remains.
This is again consistent with gravity results for deformed models of $\mathcal{N}=1$ super-JT with a non-extensive number of BPS states, and a particularly anomalous sector of the theory with $\mathcal{N}=2$ supersymmetry. The argument for spectral gaps in purely non-BPS supermultiplets is first made in \cref{ssec:hamran}, and a random matrix construction is given in \cref{sec:susyrmt}.

\section{Spectral Ensembles}
\label{sec:specen}

We introduce here the general structure in an ensemble of theories that is relevant to the study of annealed and quenched entropies. This allows us to make completely general statements about how these quantities behave and differ from each other. Known gravitational entropy results are then presented and reviewed through the lens of spectral ensembles. 
The tools we employ here will also emerge and be explicitly realized in \cref{sec:rmf} when studying random matrix ensembles.

\subsection{Basics}
\label{ssec:basspecen}

Given a theory with Hamiltonian $H$, the canonical partition function $Z(\beta)$ is given by \cref{eq:trbh}. Let $H$ have a discrete spectrum of eigenvalues $\mathcal{E}\equiv \{E_k\in\mathbb{R}\}$. The spectral density of $H$ is
\begin{equation}
\label{eq:dcomb}
    \rho_\mathcal{E}(E) \equiv \sum_k \delta(E - E_k),
\end{equation}
and $Z(\beta)$ can more explicitly be written
\begin{equation}
\label{eq:zbspd}
    Z(\beta) = \sum_k e^{-\beta E_k} = \int\displaylimits_{\mathbb{R}} dE \, \rho_\mathcal{E}(E) \, e^{-\beta E}.
\end{equation}
As anticipated below \cref{eq:trbh}, given that $Z(\beta)$ only depends on the spectrum of the Hamiltonian, we can generally reduce our theory ensembles to spectral ensembles.
Specifically, an ensemble of theories given by a probability density function (PDF) over some space of Hamiltonians $H$ yields a unique joint PDF over their eigenvalues $\mathcal{E}$ upon diagonalization. This joint PDF defines the measure $p(\mathcal{E})$ of the spectral ensemble.
Equipped with it, the expectation value in the spectral ensemble of any function $f$ of $\mathcal{E}$ is defined by
\begin{equation}
\label{eq:evf}
    \langle f\rangle \equiv \int d\mathcal{E} \, p(\mathcal{E}) \, f(\mathcal{E}).
\end{equation}
For $N$ eigenvalues, one can generally write $p(\mathcal{E})$ in terms of a symmetric function $p_N$ as
\begin{equation}
\label{eq:jpdfsym}
    p(\mathcal{E}) = p_N(E_1,\dots,E_N).
\end{equation}
The joint PDF induced on any number $1\leq m \leq N$ of eigenvalues can be obtained by simply integrating $p_N$ over the other set of $N-m$ eigenvalues,
\begin{equation}
\label{eq:pnPN}
    p_n(E_1,\dots,E_n) \equiv \int\displaylimits_{\mathbb{R}^{N-m}} dE_{m+1}\cdots dE_{N} \; p_N(E_1,\dots,E_N).
\end{equation}
These can be understood as $m$-point eigenvalue correlation functions, and simply related to correlators of the discrete spectral density in \cref{eq:dcomb} by\footnote{\,
One may want to rename the dummy variables $\mathcal{E}$ being integrated over in \cref{eq:evf} to avoid confusion.}
\begin{equation}
\label{eq:specdensc}
    \langle\rho(E_1)\cdots\rho(E_n)\rangle = N^m p_n(E_1,\dots,E_n).
\end{equation}
For $m=1$, this defines the average spectral density in terms of the $1$-point function as
\begin{equation}
\label{eq:densn1}
    \langle \rho(E) \rangle = N p_1(E),
\end{equation}
When clear from context, the average $\langle \rho(E) \rangle$ will be referred to simply as the spectral density.

The definitions of annealed and quenched entropies in \cref{eq:annS,eq:quenS} can now be made more explicit.
For the logarithm function, annealing leads to
\begin{equation}
\label{eq:anlogexp}
    \log \, \langle Z(\beta) \rangle = \log\left(N \int\displaylimits_{\mathbb{R}} dE \, p_1(E) \, e^{-\beta E} \right),
\end{equation}
where we have simply used \cref{eq:zbspd,eq:densn1}.
Hence the annealed logarithm solely depends on the $1$-point function of the spectral ensemble, which is why it is simply expressible in terms of the spectral density. In contrast, when quenching we are stuck with
\begin{equation}
\label{eq:qlogexp}
    \langle \log Z(\beta)\rangle = \int\displaylimits_{\mathbb{R}^N} dE_1\cdots dE_N \, p_N(E_1,\dots,E_N) \, \log(\sum_{k=1}^N e^{-\beta E_k}),
\end{equation}
which makes manifest that the quenched logarithm depends on the full spectral $N$-point function. Clearly, the annealed entropy in \cref{eq:annS} is a much coarser quantity determined by only $1$-point statistics, whereas the quenched entropy in \cref{eq:quenS} is sensitive to all higher-point statistics. 
Needless to say, the quenched entropy is not computable from just knowledge of the average spectral density, unless the $m$-point functions are all trivial.\footnote{\,An $m$-point function is non-trivial if its connected $m$-point part is nonzero, for otherwise it would be determined by $k$-point functions with $k<m$. For instance, a trivial $2$-point function gives $\expval{\rho(E_1)\rho(E_2)} = \expval{\rho(E_1)}\expval{\rho(E_2)}$.} More explicitly, \cref{eq:anlogexp,eq:qlogexp} are equal only in the trivial case in which all $m$-point correlators in \cref{eq:specdensc} factorize as
\begin{equation}
    \langle\rho(E_1)\cdots\rho(E_n)\rangle = \langle\rho(E_1)\rangle\cdots\langle\rho(E_n)\rangle,
\end{equation}
i.e., an uninteresting ensemble of independently and identically distributed random energies.

\subsection{Quenched Expectations}
\label{ssec:qexp}

The logarithm expressions in \cref{eq:anlogexp,eq:qlogexp} imply that annealed and quenched entropies are generally very different quantities in non-trivial ensembles. The only room for agreement between them is if in some regime in $\beta$ the quenched logarithm becomes dominated by $1$-point statistics. Expanding $\log Z(\beta)$ about $\beta=0$, clearly a term of $O(\beta^k)$ can only possibly depend on up to $k$ distinct eigenvalues. Hence at small $\beta$,
\begin{equation}
\label{eq:hight}
    \expval{\log Z(\beta)} = \log N - \beta \expval{E} + O(\beta^2),
\end{equation}
which clearly matches $\log \expval{ Z(\beta)}$ up to $O(\beta)$. Contributions to \cref{eq:hight} from $2$-point functions begin at $O(\beta^2)$, and thus the annealed logarithm begins missing the effect of their connected parts already at this order.
Similarly, $k$-point functions begin contributing to $\expval{\log Z(\beta)}$ at $O(\beta^k)$, and their connected parts cannot possibly be captured by $\log \expval{Z(\beta)}$. 
For any non-trivial spectral ensemble, the annealed logarithm thus unavoidably begins to differ from the quenched one already at $O(\beta^2)$. As a result, we expect the annealed entropy to approach the quenched entropy at high temperatures, but become a worse approximation to it at lower temperatures if the connected parts of $k$-point functions contribute non-negligibly.

Since higher-point functions become important at higher orders in $\beta$, we learn that $\expval{\log Z(\beta)}$ becomes sensitive to the full statistics of the ensemble as $\beta$ increases. At large $\beta$ one obtains
\begin{equation}
\label{eq:lowt}
    \expval{\log Z(\beta)} = \expval{\log N_0} -\beta \expval{E_0} + \expval{O(e^{-\beta\Delta})},
\end{equation}
where $E_0\equiv \min\{E_1,\dots,E_N\}$ should be understood as the ground-state energy, $N_0\geq1$ is its degeneracy, and the exponentially suppressed corrections depend on the statistics of the separation $\Delta$ between $E_0$ and the next higher eigenvalue. Although $N_0$ may sometimes be a fixed constant, the expectation value of $E_0$ does generically depend non-trivially on the full spectral $N$-point function. In other words, the large-$\beta$ behavior of $\expval{\log Z(\beta)}$ can by no means be reproduced by $\log \expval{Z(\beta)}$. 

There is a simple exception to this which turns out to be of much interest.
In particular, it is possible to make annealed and quenched logarithms agree to leading order at large-$\beta$ as follows. Suppose some symmetry of the original ensemble of Hamiltonians makes their ground-state energy a protected quantity.\footnote{\,As one may anticipate, in theories with supersymmetry this is precisely what happens for near-BPS black holes.} In this case the expectation value of the ground-state energy in \cref{eq:lowt} just yields $\expval{E_0} = E_0$. On the other hand, one easily sees that the greatest lower bound on the support of the average spectral density is $\inf\supp\expval{\rho}=E_0$, and that
\begin{equation}
\label{eq:discpro}
    \expval{\rho(E)} \supset \expval{N_0} \delta(E-E_0).   
\end{equation}
Consequently, in this non-generic case the annealed logarithm in \cref{eq:anlogexp} would behave just like the quenched one to leading order at large $\beta$, up to a potential discrepancy between $\expval{\log N_0}$ and $\log{\expval{N_0}}$. In the concrete realizations of this setting that we will study using matrix ensembles, $N_0$ is in fact fixed and thus this constants agree. Even in this case though, discrepancies between annealing and quenching do still arise at large $\beta$ at orders higher than just discussed.

It is illustrative to be more explicit about how $\expval{\rho}$ preserves a discrete ground state as in \cref{eq:discpro} when its energy is protected. The average spectral density can be thought of as the result of smearing the discrete energy levels of each Hamiltonian in the ensemble as $p(\mathcal{E})$ dictates. Hence consider letting $\theta_k$ be the map on $\mathcal{E}$ which yields $1$ if say $E_1$ is the $k^{\text{th}}$ energy level, and zero otherwise. The sum over all these indicator functions $\theta_k$ adds up to $1$, and thus
\begin{equation}
	\rho_k(E) \equiv N \int d\mathcal{E} \, p(\mathcal{E}) \, \theta_k(\mathcal{E}) \, \rho_\mathcal{E}(E),
\end{equation}
gives the probability distribution of values that the $k^{\text{th}}$ energy level takes across the ensemble.
The average spectral density is then suggestively a superposition of these for every energy level,\footnote{\,This representation of the average spectral density was studied in the context of random matrices in \cite{Johnson:2021zuo,Johnson:2022wsr}, which beautifully spelled out its implications for the underlying discreteness of quantum gravitational spectra and the emergence of an ensemble description in effective gravity. Note that generically the support of some $\rho_k$ distributions, and possibly all of them, may overlap on open sets. For instance, some Hamiltonian in the ensemble may have a ground-state energy that is strictly higher than that of the first excited state of another, thus implying an overlap on an open set $\supp\rho_1\cap\supp\rho_2\neq\varnothing$. If instead the ground-state energy $E_0$ is protected across the ensemble, then excited states may at most reach $E_0$ from above, meaning that $\supp\rho_1=\{E_0\}$ and $\inf\supp\rho_k\geq E_0$ for all $k\geq2$.}
\begin{equation}
	\langle\rho(E)\rangle = \sum_{k=1}^N \rho_k(E).
\end{equation}
This expression grants the interpretation of $\langle\rho(E)\rangle$ as capturing the statistical uncertainty in the energy levels of the random Hamiltonian. The non-generic case described above where the ground-state energy $E_0$ is protected corresponds to having no uncertainty in the lowest eigenvalue, i.e.,
\begin{equation}
\label{eq:ann1exp}
    \rho_1(E) = \expval{N_0} \delta(E-E_0).
\end{equation}
By construction, in this case no $\rho_k$ with $k\geq2$ can possibly have a tail with support below $E=E_0$, and thus we indeed verify that $\inf\supp\expval{\rho}=E_0$. 

The observations here about $\log Z(\beta)$ lead us to expect large discrepancies between annealing and quenching when $\beta$ is large, and also point at a non-generic situation in which they may behave similarly. In trying to understand how the gravitational entropy compares to a thermal entropy, we are thus drawn to study how annealed entropies do actually behave at large $\beta$ in these situations. As we will show in \cref{sec:fact}, \cref{eq:ann1exp} turns out to be necessary and sufficient for the annealed entropy to stay non-negative at low temperatures. Before we turn to this though, let us conclude with two brief remarks.

As emphasized above, while $\log \expval{Z(\beta)}$ only captures $1$-point statistics, $\expval{\log Z(\beta)}$ depends on all-point statistics and is thus sensitive to much finer details of the spectral ensemble. The investigation of quantities probing higher-point functions has proven remarkably insightful in  quantum gravity, as most prominently demonstrated by calculations of the spectral form factor \cite{Cotler:2016fpe,Liu:2018hlr,Stanford:2020wkf,Saad:2018bqo,Okuyama:2018gfr,Saad:2019lba,Saad:2019pqd}. However, this object basically involves the second moment $\expval{Z(\beta)^2}$, which only depends on up to $2$-point functions.
The regime in which the spectral form factor becomes interesting corresponds precisely to when $\expval{\log Z(\beta)}$ begins to differ from $\log \expval{Z(\beta)}$.
Being sensitive to even higher-point statistics, one may thus expect quenched entropies to capture even more interesting physics.

Finally, note that our analysis above says that $\expval{\log Z(\beta)}$ becomes sensitive to $m$-point functions with increasingly higher $m$ as $\beta$ becomes large. On the other hand, the replica trick in \cref{eq:reptrick} says that to compute this quantity we should consider $m$-point functions in the $m\to0$ limit.
This is our first hint that $\beta\to\infty$ and $m\to0$ are competing limits whose order matters. Since the latter is strict in the replica trick, this already tells us that doing large-$\beta$ approximations at finite $m$ before taking the $m\to0$ replica limit may be a terrible idea.

\subsection{Negative Entropies}
\label{sec:fact}

Since this section focuses solely on annealed quantities, the notation $\expval{\,\cdot\,}$ will be temporarily dropped for the sake of clarity.
Annealed entropies just depend on the average spectral density, and thus can be studied in terms of the properties of a general such object. For the sake of generality, let this spectral density be given by an arbitrary map
\begin{equation}
    \widetilde{\rho}: \mathbb{R} \to [0,\infty],
\end{equation}
where we allow for distributions by including infinity in the codomain.
We demand $\widetilde{\rho}$ to be locally integrable, and assume its support is bounded from below to define a ground-state energy by
\begin{equation}
    E_0\equiv \inf \supp \widetilde{\rho}.
\end{equation}
The canonical partition function is then given by\footnote{\,As emphasized by the notation $\int_{0^-}^\infty \equiv \lim_{\epsilon\to0^+} \int_{-\epsilon}^\infty$, these integrals capture any point contribution at $0$.}
\begin{equation}
\label{eq:Zbeta}
    Z(\beta) =\int\displaylimits_{\mathbb{R}} dE \, \widetilde{\rho}(E) \, e^{-\beta E} = \int_{0^-}^\infty dE \, \rho(E) \, e^{-\beta E},
\end{equation}
where in the second equality we have introduced a canonically renormalized spectral density,
\begin{equation}
\label{eq:ltdef}
    \rho(E)\equiv e^{-\beta E_0} \widetilde{\rho}(E_0 + E).
\end{equation}
This is useful because now $Z$ takes the form of a Laplace transform.
In what follows, we refer to the canonically shifted $\rho$ whose ground-state energy is zero as the spectral density.

Allowing for $\rho$ to be discontinuous and distributional, the general analysis of its Laplace transform can get cumbersome.
Motivated by \cref{ssec:qexp}, we will assume that $\rho$ admits a decomposition into distributional and continuous parts of the general form
\begin{equation}
\label{eq:rdc}
    \rho(E) = \Delta(E) + g(E), \qquad \Delta(E) \equiv N_0\,\delta(E) + \sum_{k\in I} N_k \, \delta(E-E_k),
\end{equation}
for some set $I\subseteq \mathbb{N}$ (possibly empty or infinite), $N_0\geq0$, $N_k>0$, all $E_{k}>0$ distinct, and where $g$ is now an actual real function, assumed to be non-negative and piecewise continuous.
By the canonical shift in \cref{eq:ltdef}, it follows that either $\inf\supp g =0$, or $N_0>0$, or both. 
Generally $g$ may contain discontinuities and singularities, so when studying it we will generally be interested in keeping track of one-sided limits, and will not require two-sided limits to exist. Our convention will be to allow for limits to exist in the extended non-negative reals $[0,\infty]$.

The Laplace transform for $\rho$ of the general form in \cref{eq:rdc} reads
\begin{equation}
\label{eq:lpgen}
    Z(\beta) = \hat{Z}(\beta) + Z_g(\beta), \qquad    
    \begin{cases}
        \displaystyle \hat{Z}(\beta) \equiv N_0 + \sum_{k\in I} N_k \, e^{-\beta E_k} \\
        \displaystyle Z_g(\beta) \equiv \int_{0^+}^\infty dE \, g(E) \, e^{-\beta E}
    \end{cases}
\end{equation}
We will only be interested in spectral densities $\rho$ whose Laplace transform $Z$ exists, which requires $\hat{Z}$ and $Z_g$ to independently exist. For $\hat{Z}$, if its form in \cref{eq:lpgen} involves an infinite series, we assume it converges for any $\beta>c$ for some finite $c>0$.
As for $Z_g$, we assume $g$ is locally integrable and of exponential type,\footnote{\,This is, $\exists c,M,E^*>0$ such that $|g(E)|\leq M e^{c E}$ for all $E>E^*$.} which guarantees that $\lim_{\beta\to\infty} e^{-\beta E}g(E) = 0$ for any $\beta>c$ for finite $c>0$.
Note that the local integrability condition allows for integrable singularities where $g$ diverges. In particular $g(E)$ may diverge as $E\to0^+$ and, so long as it does so slower than $E^{-1}$, $Z_g$ will still exist.

With a well-defined $Z$, one can then proceed to study the annealed entropy $S_a$ from \cref{eq:annS}.
Note that this object is invariant under shifts of $\rho$ like \cref{eq:ltdef} or, at the level of $Z$, under $Z(\beta) \to e^{-\beta E_0} Z(\beta)$ for any $E_0\in\mathbb{R}$. This demonstrates that our canonical shift to always fix $\inf\supp\rho=0$ was without loss of generality. The main result of this section may now be stated:

\begin{restatable}{nthm}{thermofact}
\label{thm:fact}
    Let $S_a(\beta)$ be the annealed entropy defined in \cref{eq:annS}, with $\expval{Z(\beta)}$ given by the Laplace transform of an average spectral density $\rho(E)$ of the general form in \cref{eq:rdc}. Then:
    $$S_a(\beta) \ge 0 \quad\Longleftrightarrow\quad \rho(E) \supset N_0\, \delta(E) \quad\text{with}\quad N_0\ge1.$$
    If instead $N_0<1$, then $S_a(\beta)<0$, $\forall\beta>\beta^*$ for some $\beta^*>0$, and $\displaystyle \lim_{\beta\to\infty}S(\beta)=-\infty$ for $N_0=0$.\footnote{\,Restoring the $\expval{\,\cdot\,}$ notation, note that $0<\expval{N_0}<1$ is not disallowed: it can be realized by an ensemble where the ground-state energy is protected for only some Hamiltonians, with all others having strictly higher energies.}
\end{restatable}

The proof of this statement is given in \cref{sec:proof}. In words, \cref{thm:fact} says that the necessary and sufficient condition for the annealed entropy to be non-negative is that the average spectral density preserve at least one discrete ground state. This both confirms and explains the general behavior that the gravitational entropy is observed to exhibit in \cref{eq:genstt}. In particular, by \cref{thm:fact} the negativity of $\mathcal{S}(\beta)$ at large $\beta$ when no supersymmetry protects the energy at extremality is not an artifact of the perturbative treatment, but a genuine property that the gravitational entropy must exhibit as an annealed entropy. This negativity can thus be read as a remarkable consistency check of the ensemble description dual to the gravitational path integral.
The physical restatement of \cref{thm:fact} may be phrased as follows:

\begin{ncor}
\label{cor:fact}
	Let $\mathcal{S}(\beta)$ be the gravitational entropy defined in \cref{eq:slogZ} for a partition function $\mathcal{Z}(\beta)$ given by the Euclidean gravitational path integral. Then $\mathcal{S}(\beta)$ is non-negative for all $\beta>0$ if and only if the ground-state energy is protected. Otherwise, $\mathcal{S}(\beta)\to -\infty$ continuously as $\beta\to\infty$.
\end{ncor}

\subsection{Black Hole Spectra}
\label{sec:bhspec}

The spectral densities that realize the near-extremal gravitational entropies in \cref{eq:genstt} nicely illustrate these results. The $\beta$ dependence relevant to our discussion comes from strongly coupled gravitational dynamics that arise in the throat of black holes near extremality. These are universally governed by a quantum mechanical model known as the Schwarzian theory and its supersymmetric generalizations \cite{Jensen:2016pah,Maldacena:2016upp,Engelsoy:2016xyb,Ghosh:2019rcj,Iliesiu:2020qvm,Heydeman:2020hhw,Iliesiu:2022onk,Boruch:2022tno,Banerjee:2023quv,H:2023qko,Rakic:2023vhv,Kapec:2023ruw} (see \cref{sec:nebh} for more details).
The path integral of these Schwarzian theories has been shown to be $1$-loop exact and fully solved, thus providing a robust framework for the study of near-extremal black hole spectra \cite{Bagrets:2016cdf,Cotler:2016fpe,Stanford:2017thb,Belokurov:2017eit,Mertens:2017mtv,Kitaev:2018wpr,Yang:2018gdb,Heydeman:2020hhw}.

Here, we review the gravitational spectral densities that arise from these Schwarzian theories.
All cases, near-BPS or not, and gapped or not, can be concisely captured starting from an average spectral density which takes the simple form\footnote{\,Some $O(1)$ constants which are irrelevant to our discussion are for clarity being ignored against $S_0\gg1$.} \cite{Heydeman:2020hhw,Boruch:2022tno,Iliesiu:2022onk,Iliesiu:2022kny,Turiaci:2023wrh,Banerjee:2023quv,H:2023qko,Rakic:2023vhv,Kapec:2023ruw}
\begin{equation}
\label{eq:egdelta}
    \rho_c(E) = e^{S_0} \delta(E) \,\delta_{p\leq0} + 
    \frac{e^{S_0}}{\sqrt{2\pi}} \left(\frac{\sqrt{c E}}{2\pi c}\right)^{p-1} I_{p-1}\left(2\pi\sqrt{c E}\right),
\end{equation}
where $I_n$ is the modified Bessel function of the first kind. The parameter $p\in\mathbb{Q}$ matches $s=p$ for the non-near-BPS case of \cref{eq:genstt}, and takes values $p\leq0$ for near-BPS black holes. In consistency with \cref{thm:fact}, the annealed entropy stays non-negative precisely only when the discreteness of the ground-state energy is protected by supersymmetry. 

For non-near-BPS black holes, one may choose units such that $c=1$ and \cref{eq:egdelta} captures the complete spectral density of states. For near-BPS black holes, one can obtain the full spectral density by adding up the contribution of all possible saddles that supersymmetric localization gives \cite{Stanford:2017thb,Heydeman:2020hhw}. 
Each such saddle gives a spectral density of the form of \cref{eq:egdelta} with a corresponding value of $c$, which are to be summed over and appropriately weighted.
This series can be reorganized into a sum over supermultiplets with different R-charges. Supermultiplets are statistically independent, and thus their spectral densities are the ones of physical interest \cite{Turiaci:2023jfa}.

\begin{figure}
    \centering
    \includegraphics[width=0.99\textwidth]{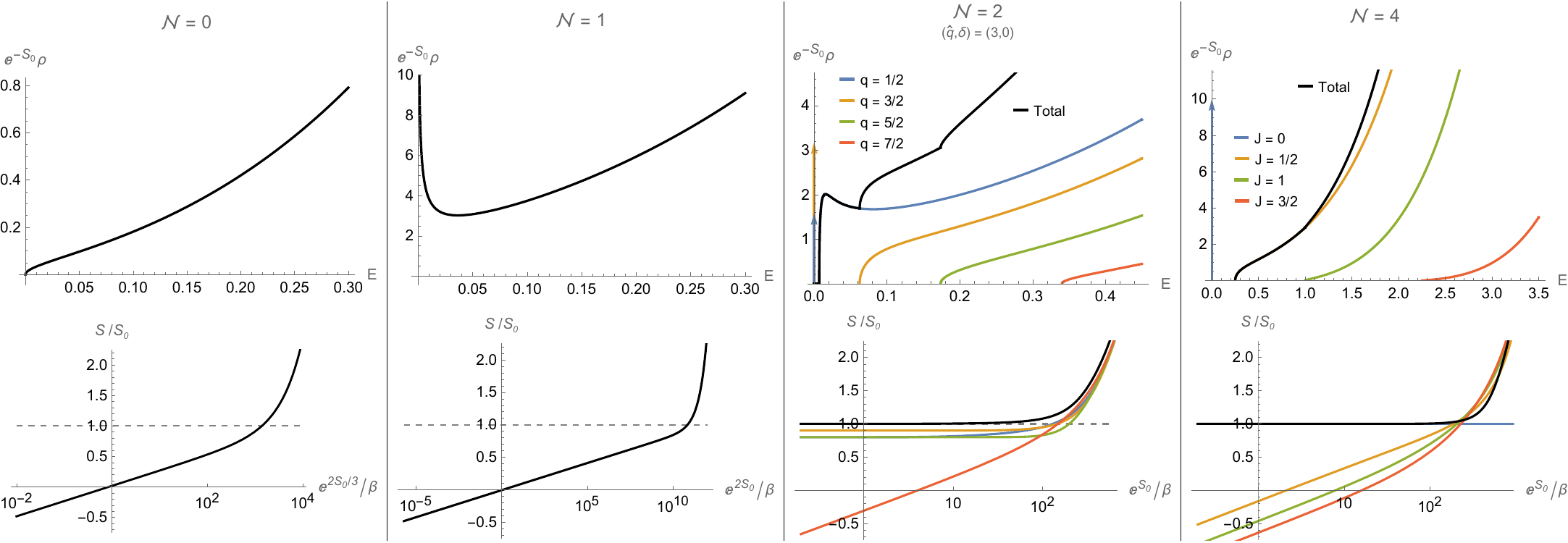}
    \caption{Black hole spectra (top) and gravitational entropies (bottom) from each $\mathcal{N}=0,1,2,4$ super-Schwarzian theory (left to right). Spectra are respectively given by \cref{eq:bos,eq:susynon,eq:susy2,eq:rq4}, and the corresponding entropies are obtained by Laplace transform and evaluation of \cref{eq:sbdef}. For the $\mathcal{N}=2$ theory we have chosen $\hat{q}=3$ and $\delta=0$. The plots for $\mathcal{N}=2,4$ show curves corresponding to different supermultiplets separately, and also for the total spectral density of the theory (black). Discrete ground-state contributions are symbolized by an arrow at zero of $10$ times the degeneracy value for ease of visualization. These BPS states are accounted for when computing the entropy for the total spectral density, which clearly is not just a sum over supermultiplet entropies. In particular, note that the entropies of supermultiplets with no BPS states ($q=7/2$ for $\mathcal{N}=2$ and all $J>0$ for $\mathcal{N}=4$) all diverge negatively at large $\beta$, just like the entropies for the $\mathcal{N}=0,1$ theories.
    Note the logarithmic scale of the horizontal axis.}
    \label{fig:bhspec}
\end{figure}

We now evaluate \cref{eq:egdelta} for the specific values of $p$ found in gravitational calculations, and illustrate the results in \cref{fig:bhspec}. In purely bosonic theories of gravity with no supersymmetry ($\mathcal{N}=0$), as mentioned below \cref{eq:genstt}, one finds $p=3/2$ and \cref{eq:egdelta} yields
\cite{Stanford:2017thb,Mertens:2017mtv,Stanford:2019vob,Iliesiu:2020qvm,Iliesiu:2022onk,Banerjee:2023quv,Rakic:2023vhv,Kapec:2023ruw}
\begin{equation}
\label{eq:bos}
    \rho^{\mathcal{N}=0}(E) = \frac{e^{S_0} \sinh (2 \pi \sqrt{E} )}{2\pi^2}, \qquad Z^{\mathcal{N}=0}(\beta) = \frac{e^{S_0+\pi^2/\beta }}{2\sqrt{\pi} \beta ^{3/2}}.
\end{equation}
Note that $\rho^{\mathcal{N}=0}(E) \sim \sqrt{E}$ near $E=0^+$.
Certain supersymmetric theories contain black holes which break all supersymmetry at extremality and whose near-extremal physics is governed by the $\mathcal{N}=1$ super-Schwarzian. In this case one finds $p=1/2$ and \cref{eq:egdelta} yields
\cite{Stanford:2017thb,Mertens:2017mtv,Stanford:2019vob,Johnson:2022wsr}
\begin{equation}
\label{eq:susynon}
    \rho^{\mathcal{N}=1}(E) = \frac{e^{S_0} \cosh(2\pi \sqrt{E})}{\pi \sqrt{E}}, \qquad Z^{\mathcal{N}=1}(\beta) = \frac{e^{S_0 + \pi^2/\beta}}{\sqrt{\pi}\beta^{1/2}}.
\end{equation}
Note that $\rho^{\mathcal{N}=1}(E) \sim 1/\sqrt{E}$ near $E=0^+$.
The near-BPS case requires the emergence of a theory with more supersymmetry near extremality.
For the $\mathcal{N}=2$ super-Schwarzian one finds $p=0$, and the saddle-point contributions labelled by $c$ in \cref{eq:egdelta} read
\cite{Stanford:2017thb,Mertens:2017mtv,Heydeman:2020hhw,Iliesiu:2022kny,Boruch:2022tno,Turiaci:2023jfa,Johnson:2023ofr}
\begin{equation}
\label{eq:susypro}
    \rho_c^{\mathcal{N}=2}(E) = e^{S_0} \delta(E) + e^{S_0} \sqrt{2\pi c} \, \frac{I_1\left( 2\pi \sqrt{cE}\right)}{\sqrt{E}}.
\end{equation}
These contributions can be reorganized into spectral densities for supermultiplets of the theory. Here we quote the results and defer a more detailed discussion to \cref{ssec:hamran}.
The general $\mathcal{N}=2$ theory can be specified by an odd integer $\hat{q}>0$ and a constant $\delta\in[0,1)$, and its supermultiplets labelled by a parameter $q\in\mathbb{Z}+\delta-1/2$.\footnote{\,Here $\pm\hat{q}+\delta$ are the R-charges of the supercharges, $\delta$ is a parameter accounting for a possible anomaly, and $q = k + \hat{q}/2$ is the average of the R-charges $(k,k+\hat{q})$ of the states that make up the supermultiplet \cite{Turiaci:2023jfa}.} The spectral density of the $q$ supermultiplet is \cite{Turiaci:2023jfa}
\begin{equation}
\label{eq:susy2}
    \rho_q^{\mathcal{N}=2}(E) = \frac{e^{S_0} \sin(\pi q/ \hat{q})}{2} \, \delta(E) \, \delta_{\abs{q}<\hat{q}}+\frac{e^{S_0} \sinh(2\pi \sqrt{E - \Delta_q})}{4\pi E} \theta(E - \Delta_q), \qquad \Delta_q \equiv \frac{q^2}{4\hat{q}^2},
\end{equation}
where $\theta$ is the Heaviside step function. 
Only supermultiplets with $0<\abs{q}<\hat{q}$ contain BPS states, which give rise to a discrete ground state with degeneracy proportional to $e^{S_0}$.
The continuous part of the spectrum is only supported on energies $E=\Delta_q\geq 0$ above the ground state. 
The gap size $\Delta_q$ depends on the supermultiplet and is generically nonzero.
We will show in \cref{sec:gaps} that the emergence of such a positive gap can be explained from the perspective of matrix integrals not only when there is a high degeneracy of BPS states at zero energy, but also for purely non-BPS supermultiplets (see also \cref{ssec:hamran,sec:susyrmt}). In addition, that the growth of $\Delta_q$ is quadratic in $q$ for supermultiplets with BPS states will be shown to be universally predicted by random matrix theory. The smallest such $\Delta_q$ determines the onset of the continuous density and thus the overall gap $\Delta\equiv \min_q \Delta_q$ for the full spectrum. 
Generically $\Delta>0$, except for the specific anomalous case with $\delta=1/2$, which is the only one that allows for a gapless $q=0$ supermultiplet \cite{Stanford:2017thb,Heydeman:2020hhw,Lin:2022zxd}.
With this anomaly, the ground state can become completely depopulated if $\hat{q}=1$; otherwise, there are always supermultiplets with $O(e^{S_0})$ BPS states.
Note that $\rho_q^{\mathcal{N}=2}(E-\Delta_q) \sim \sqrt{E-\Delta_q}$ near $E=\Delta_q^+$ for the generic $q\neq0$ case (cf. \cref{eq:bos}), whereas for the anomalous $q=0$ case $\rho_0^{\mathcal{N}=2}(E) \sim 1/\sqrt{E}$ near $E=0^+$ (cf. \cref{eq:susynon}).

Finally, for the $\mathcal{N}=4$ super-Schwarzian one finds $p=-1$ and the saddle-point contributions from \cref{eq:egdelta} read \cite{Heydeman:2020hhw,Iliesiu:2021are,Iliesiu:2022kny}
\begin{equation}
\label{eq:susypro2}
    \rho_c^{\mathcal{N}=4}(E) = e^{S_0} \delta(E) + e^{S_0} (2\pi)^{3/2} c \, \frac{I_2\left( 2\pi \sqrt{cE}\right)}{E}.
\end{equation}
The $\mathcal{N}=4$ theory is unique, and its supermultiplets labelled by a half-integer spin $J$ have \cite{Turiaci:2023jfa}
\begin{equation}
\label{eq:rq4}
    \rho_J^{\mathcal{N}=4}(E) = e^{S_0} \delta_{J,0} + \frac{e^{S_0} J \sinh(2\pi\sqrt{E-\Delta_J})}{2\pi^2 E^2} \theta(E - \Delta_J), \qquad \Delta_J \equiv J^2.
\end{equation}
In this case BPS states are all spin zero, and non-BPS states have $J\geq 1/2$.
In contrast with $\mathcal{N}=2$, the spectrum for $\mathcal{N}=4$ does always exhibit a positive gap $\Delta=1/4$, and note that $\rho_J^{\mathcal{N}=4}(E-\Delta_J^+) \sim \sqrt{E-\Delta_J}$ near $E=\Delta_J^+$ always in this case \cite{Turiaci:2023jfa}.\footnote{\,The canonical partition function of the theory would not even exist if $\Delta_J=0$ were allowed, since the Laplace transform of \cref{eq:rq4} would have a non-integrable singularity at zero in this case.}
The gapped spectral densities for both $\mathcal{N}=2,4$ give rise to gravitational entropies of the form in \cref{eq:genstt}. When gapless, this reduces to the behavior in \cref{eq:genstt} with $s=0$.

The intermediate case in \cref{eq:susynon} of a supersymmetric theory that arises for non-near-BPS black holes is  physically interesting. In particular, this spectral density diverges as $1/\sqrt{E}$ in the $E\to0^+$ limit, but the gravitational entropy nonetheless becomes negative at large $\beta$ in consistency with \cref{thm:fact}. The same happens in the anomalous $\mathcal{N}=2$ theory that allows for a gapless supermultiplet. This divergence at the ground-state energy may be understood as a supersymmetric transition from a non-degenerate to a highly-degenerate ground state. Intuitively, it captures how the degeneracy of the near-BPS cases gets lifted as soon as supersymmetry is broken, leaving ground-state energies highly populated but unprotected to values above zero.
In addition, $\mathcal{N}=1$ also exhibits a dip between its divergent ground state and its exponential growth at large $E$, anticipating the formation of a spectral gap with supersymmetry.

Besides the obvious fact that \cref{eq:bos,eq:susynon,eq:susypro,eq:susy2,eq:susypro2,eq:rq4} have continuous parts, we have thus far said nothing about how these gravitational spectra are associated to spectral ensembles.
The realization of these as average spectral densities is more explicitly established by dimensional reduction to a universal AdS$_2$ sector of the geometry that emerges in the throat of near-extremal black holes. This reduction allows for the definition of $2$-dimensional theories of JT gravity where the corresponding Schwarzian governs large diffeomorphisms on the AdS$_2$ geometry. Formulating JT as a gravitational path integral with a sum over topologies leads to a non-perturbative topological expansion in $e^{-S_0}$ for observables like those in \cref{eq:Zmeval}. Through topological recursion relations, these have been shown to precisely match double-scaled matrix integrals with expansion parameter $e^{-S_0}$ to all orders \cite{Saad:2019lba,Stanford:2019vob,Turiaci:2023jfa,Mertens:2022irh}. This result is what leads to the identification of \cref{eq:bos,eq:susynon,eq:susypro,eq:susy2,eq:susypro2,eq:rq4} as capturing the statistics of spectral ensembles in the sense of \cref{ssec:basspecen}.

In the context of JT gravity and random matrices, the above results are all exact perturbatively at large $S_0$, but only hold to leading non-perturbative order at large $e^{S_0}$. 
The non-perturbative effects that the genuinely quantum gravitational sum over topologies gives rise to appear suppressed by factors of $e^{-S_0}$. These non-perturbative corrections are what makes the duality between JT gravity and matrix ensembles very non-trivial. 
In this paper, however, we mostly work to leading order at large $e^{S_0}$, where the emergence of an ensemble description is not as obvious. 
Namely, oftentimes one interprets higher-order wormhole contributions to gravitational correlators and their consequent lack of factorization as the root cause for such an ensemble description to arise.
It is thus very remarkable that to leading order at large $e^{S_0}$, Schwarzian theories alone without wormhole contributions give the spectral densities in \cref{eq:bos,eq:susynon,eq:susypro,eq:susy2,eq:susypro2,eq:rq4} all in perfect consistency with universal expectations from random matrix ensembles. Let us briefly comment on these universal features in anticipation of the results in \cref{sec:rmf}.

With regard to \cref{cor:fact}, it is worth emphasizing that in random matrix integrals the generic negativity of the annealed entropy is an exact statement, i.e., exact to all perturbative and non-perturbative orders in $S_0$. As understood in the previous sections, this negativity is caused by the ensemble generically having a non-trivial distribution of ground-state energies, unlike the non-generic case in \cref{eq:ann1exp}. Hence the fact that the Schwarzian theories with $\mathcal{N}=0,1$ lead to negative gravitational entropies at low temperatures while the $\mathcal{N}=2,4$ super-Schwarzians do not may be understood as a remarkable consistency check of the ensemble description of effective gravity already to leading non-perturbative order. In particular, in this sense the Schwarzian already knows about the ensemble nature of the gravitational path integral, even though to leading order at large $e^{S_0}$ we are not even including wormhole contributions.
This point of view should be contrasted with the discussion in \cref{sec:nebh}, where we review the traditional interpretation of these negativities as a pathology caused by the neglect of non-perturbative corrections that arise at $\beta\sim O(e^{S_0})$. 
While we agree that such corrections indeed become important at non-perturbatively low temperatures, the random matrix framework says that the qualitative negativity of the gravitational entropy must remain at non-perturbative orders.\footnote{\,This random matrix statement holds for the non-perturbative effects coming from $2$-dimensional higher topologies of the JT theory of quantum gravity obtained upon dimensional reduction in the throat. All we can say about non-perturbative effects coming from topological fluctuations of the higher-dimensional black hole spacetime is that, if the gravity theory continues to be effective and dual to an ensemble, then the same conclusion holds.}

The consistency between \cref{eq:bos,eq:susynon,eq:susypro,eq:susy2,eq:susypro2,eq:rq4} and other well-known features of the spectral statistics of matrix ensembles is also noteworthy. In all cases, the continuous part of the spectrum behaves in one of two characteristic ways near its lower edge. The spectral densities for $\mathcal{N}=0$, $\mathcal{N}=2$ with $\Delta_q>0$, and $\mathcal{N}=4$, all go to zero as $\rho \sim \sqrt{E}$, whereas those for $\mathcal{N}=1$, and $\mathcal{N}=2$ with $\Delta_q=0$, diverge as $\rho \sim 1/\sqrt{E}$. These are precisely the two universal behaviors that arise in the analysis of edge statistics in random matrix theory, as we review in \cref{sec:rmf}. Another general fact about matrix integrals we will address is that a gap in the spectrum arises if and only if the degeneracy of the zero eigenvalue scales with the size of the matrix, and that the size of the gap is quadratic in their ratio. The perfect consistency between this phenomenon and what is observed in \cref{eq:susy2,eq:rq4} is rather remarkable. In both cases, supermultiplets with a positive gap $\Delta_q>0$ always involve a discrete ground state with a degeneracy that is extensive in $e^{S_0}$, while those $\mathcal{N}=2$ supermultiplets with $\Delta_q=0$ do not have such a ground state \cite{Heydeman:2020hhw,Turiaci:2023jfa}.\footnote{\,Deformations of JT gravity with a non-extensive ground-state degeneracy continue to behave consistently \cite{Iliesiu:2021are}.}

\section{Random Matrices}
\label{sec:rmf}

Our main motivation for studying random matrix ensembles is to understand whether the replica trick in \cref{eq:reptrick} can be implemented in effective gravity and, if so, to what perturbative or non-perturbative order in $S_0$ one has to work to succeed.\footnote{\,Previous attempts in gravity have been inconclusive \cite{Engelhardt:2020qpv,Chandrasekaran:2022asa}, and explorations using the random matrix framework have made use of doubly non-perturbative tools of $O(e^{-e^{S_0}})$ which are inaccessible to gravity \cite{Johnson:2020mwi,Johnson:2021rsh}. The former cannot answer if gravity is sufficient, while the latter cannot answer if $O(e^{-e^{S_0}})$ physics is necessary.}
To pursue this goal in broad generality, we consider integrals over arbitrary matrix ensembles, thus making our predictions apply to any quantum gravity theory with such a dual description.
This section reviews the basics of the framework necessary for our discussion, including how Hamiltonians are constructed out of the random matrix in different ensembles, the manifestation of supersymmetry, the matrix scalings relevant to gravity, and the relation between spectral gaps and zero-eigenvalue degeneracy.\footnote{\,See e.g. \cite{Ginsparg:1993is,Eynard:2015aea,Anninos:2020ccj} for more details on random matrix theory and its applications to physics relevant to this paper.}
The application of this technology to the replica trick from \cref{eq:reptrick} is the subject of \cref{sec:canemsem}.

\subsection{Probability Densities}

Let $\mathfrak{M}\subseteq \mathbb{C}^{N\times N}$ be some space of square matrices with real eigenvalues.
An ensemble of matrices drawn from $\mathfrak{M}$ can be defined by a PDF on this space, $P : \mathfrak{M} \to \mathbb{R}_{\geq0}$, normalized via
\begin{equation}
\label{eq:normp}
    \int\displaylimits_{\mathfrak{M}} dM \, P(M) = 1,
\end{equation}
where $dM$ is a Lebesgue measure on the algebraically independent entries of $M$.
A standard way of expressing $P$, up to a normalizing factor that ensures \cref{eq:normp}, is as follows:
\begin{equation}
\label{eq:matrixZ}
    P(M) \propto 
    e^{- N \Tr V(M)},
\end{equation}
where $V$ is known as the matrix potential,\footnote{\,Quite generally $V$ is taken to be an analytic function, thus expressible as a power series $V(x) = \sum_{k=1}^\infty v_k x^k$. When lifted to matrices, powers are defined by matrix multiplication.} and the prefactor of $N$ is a conventional scaling (see \cref{ssec:mscal}). Writing the diagonal matrix of eigenvalues of $M$ as $\Lambda\equiv\diag\{\lambda_i\}_{i=1}^N$, we have
\begin{equation}
\label{eq:trsum}
    \Tr V(M) = \Tr V(\Lambda) = \sum_{i=1}^N V(\lambda_i).
\end{equation}
One may understand $P(M)$ as a joint PDF for all entries of $M$ under the measure $dM$. However, we are only interested in the spectral statistics $M$. To obtain the induced distribution on the space of eigenvalues, one simply has to integrate out all other degrees of freedom. This amounts to computing a Jacobian for whichever transformation diagonalizes $M$. The upshot is $p(\Lambda)$, a joint PDF for the $N$ eigenvalues of $M$. Denoting the Jacobian matrix of the diagonalization of $M$ by $J(\Lambda)$, this joint PDF against the Lebesgue measure on the eigenvalues may be written
\begin{equation}
\label{eq:canonicalrm}
    p(\Lambda) \equiv \frac{1}{\mathcal{Z}} e^{-N \Tr V(\Lambda)} \abs{\det J(\Lambda)}, \qquad \mathcal{Z} \equiv \int\displaylimits_{\mathbb{R}^N} d\Lambda \, e^{-N \Tr V(\Lambda)} \abs{\det J(\Lambda)},
\end{equation}
where $d\Lambda\equiv\prod_{i=1}^N d\lambda_i$, and $\mathcal{Z}$ is the partition function of the matrix ensemble, which ensures $p$ is normalized to unity.
In general, $p(\Lambda)$ is a full permutation-symmetric function of all eigenvalues.

This is how random matrices realize the general structure of spectral ensembles in \cref{ssec:basspecen}. In terms of $p(\Lambda)$, the definition of expectation values $\expval{\,\cdot\,}$, $n$-point functions $p_n$, and in particular the average spectral density $\expval{\rho}$, just parallel \cref{eq:evf,eq:pnPN,eq:densn1}, respectively.
For now, however, $p(\Lambda)$ here captures the spectral statistics of an arbitrary random matrix $M$, not necessarily the Hamiltonian. To construct physically meaningful Hamiltonians out of $M$, as we do in \cref{ssec:hamran}, we first review in \cref{ssec:ensemmess} the spaces of matrices $\mathfrak{M}$ of interest to physics.

\subsection{Ensemble Measures}
\label{ssec:ensemmess}

Some important ensembles in random matrix theory are those where the space $\mathfrak{M}$ is diagonalizable under the action of some symmetry group, which also leaves the distribution $P$ invariant. These are referred to as invariant ensembles, and those exhausting the ten discrete symmetry classes of topological invariants are of particular physical relevance to us \cite{Altland:1997zz,Zirnbauer:1996zz,Kitaev:2009mg}. Three of them are the well-known Wigner-Dyson (WD) ensembles \cite{Wigner1955,Dyson:1962es,Dyson:1962eeu}, and the other seven will be referred to here as the Altland-Zirnbauer (AZ) ensembles \cite{Altland:1997zz,Zirnbauer:1996zz,zirnbauer2010symmetry}. In contrast with more general matrix constructs, for invariant ensembles symmetry allows one to work out Jacobian measures explicitly.

The three WD ensembles \cite{Wigner1955,Dyson:1962es,Dyson:1962eeu} are characterized by invariance under the adjoint action of one of the classical Lie groups. They are known as the orthogonal ensemble (OE), unitary ensemble (UE), and symplectic ensemble (SE). Respectively, these describe real symmetric matrices with a distribution invariant under $O(N)$, complex Hermitian matrices with a distribution invariant under $U(N)$, and quaternionic anti-Hermitian matrices with a distribution invariant under $Sp(N)$.\footnote{\,In the symplectic case $N$ has to be even and matrix eigenvalues come in pairs. It is thus convenient to instead consider $Sp(2N)$ such that $N$ corresponds to the number of independent eigenvalues.
This is just an example of a plethora of subtleties stemming from the parity of $N$ in the study of invariant matrix ensembles. These will not affect our general discussion and we refer the reader to \cite{Stanford:2019vob} for a comprehensive account.}
For these three WD ensembles, the Jacobian determinant gives the eigenvalue measure
\begin{equation}
\label{eq:WDmeasure}
    \abs{\det J(\Lambda)} \reprel{WD}{=}  \abs{\Delta(\Lambda)}^\upbeta,
\end{equation}
where $\Delta(\Lambda)$ is the Vandermonde determinant
\begin{equation}
\label{eq:vandermonde}    
    \Delta(\Lambda) \equiv \det \{\lambda_j^{i-1}\}_{i,j=1}^N = \prod_{i<j} (\lambda_j-\lambda_i),
\end{equation}
and the parameter $\upbeta=1,2,4$ respectively for the OE, UE, and SE. 
The other seven AZ ensembles \cite{Zirnbauer:1996zz,Altland:1997zz,zirnbauer2010symmetry} are characterized by additional discrete symmetries: charge conjugation $C$, spatial reflections $R$, and time reversal $T$.\footnote{\,The WD ensembles correspond to matrices with either no $T$ symmetry (U), $T^2=1$ (O), or $T^2=-1$ (S). Respectively, these give the symmetry classes A, AI, and AII that make up Dyson's Threefold Way \cite{Dyson:1962eeu,Chiu_2016}.} The eigenvalue measure these give takes the form
\begin{equation}
\label{eq:AZmeasure}
    \abs{\det J(\Lambda)} \reprel{AZ}{=} \prod_{j<i} \abs{\lambda_i^2-\lambda_j^2}^\upbeta \prod_k \abs{\lambda_k}^\upalpha.
\end{equation}
The specific values of the pair $(\upalpha,\upbeta)$ for which \cref{eq:AZmeasure} descends from a random matrix integral with specific discrete symmetries can be found in \cite{Ivanov2001RandomMatrixEI}.\footnote{\,The measures in \cref{eq:WDmeasure,eq:AZmeasure2} can also be studied for general values of their parameters irrespective of any underlying random matrix construct \cite{Desrosiers_2006}. In the literature, these are often referred to as generalized Hermite or Laguerre $\upbeta$-ensembles \cite{Dumitriu_2002}, respectively, for reasons that will become apparent in \cref{sssec:univ} (see also \cref{ssec:ortho}). Although physically our interest is in matrix ensembles, our results apply to these generalizations just as well.}
Since \cref{eq:AZmeasure} is invariant under sign changes, the canonical integral over eigenvalues can be taken to be non-negative and \cref{eq:AZmeasure} turned into a measure over squared eigenvalues $\omega_i \equiv \lambda_i^2$ for $\Omega \equiv \Lambda^2$,
\begin{equation}
\label{eq:AZmeasure2}
    \abs{\det J(\Omega)} \reprel{AZ}{=} 2^{-N} |\Delta(\Omega)|^\upbeta \prod_k \abs{\omega_k}^{\frac{\upalpha-1}{2}}.
\end{equation}
Although for AZ ensembles with $\upalpha=1$ this reduces to \cref{eq:WDmeasure}, one should not be misguided into thinking that such AZ ensembles will give equivalent results to WD ensembles. In the latter we are interested in the eigenvalues of the random matrix, which lie unconstrained on $\mathbb{R}$, whereas for the former \cref{eq:AZmeasure2} only arises for squared eigenvalues, which are non-negative (cf. \cref{sec:linazens}).

Some of these AZ ensembles admit invariant generalizations of much interest. In four of them, $M$ is a rank-$2$ tensor of some type, with the continuous symmetry $G(N)$ being one of the classical Lie groups. The other three AZ ensembles describe a random $N\times N$ matrix $M$ which transforms as a bifundamental of the direct product $G(N)\times G(N)$, i.e., each index of $M$ is separately acted upon by one copy of $G(N)$ in the fundamental representation.

The latter three can very naturally be generalized to ensembles of rectangular matrices by considering products of symmetry groups of different rank. In particular, letting the symmetry group be $G(N')\times G(N)$ allows one to describe a rectangular random matrix $M$ of size $N'\times N$.
As will be seen shortly in \cref{ssec:hamran}, this construction precisely reproduces the requisite structure to describe the supermultiplet spectra from \cref{eq:susy2,eq:rq4}, with the number of BPS states given by $\bar{\upnu} \equiv \abs{N'-N}$.\footnote{\,Note that in some supermultiplets $\bar{\upnu}$ actually acquires a different meaning; see \cref{ssec:hamran,sec:susyrmt}.} For this reason, the generalization of standard AZ ensembles to rectangular matrices with $\bar{\upnu}>0$ will be referred to here as BPS ensembles.

The spectrum of a rectangular matrix $M$ is given by its singular values, which by definition are real and non-negative for any complex matrix. For the construction above, singular value decomposition gives a measure of the AZ form in \cref{eq:AZmeasure}, but with the replacement \cite{Turiaci:2023jfa}
\begin{equation}
\label{eq:alphanu}
    \upalpha \to \upalpha + \upbeta \bar{\upnu} = \upbeta (\bar{\upnu} + 1) - 1,
\end{equation}
and $\upbeta$ specified as usual by the corresponding simple group. On the right, we used that $\upalpha=\upbeta-1$ for the AZ ensembles with bifundamental symmetry that generalize to BPS ensembles.

Without loss of generality, let the rectangular matrix $M$ in our BPS ensembles be $(N+\bar{\upnu})\times N$ with $\bar{\upnu}>0$.
As a random matrix, $M$ has rank $N$ with unit probability, so consider the generic case in which the $N$ singular values of $M$ are all positive.
Using $M$, one can construct positive semidefinite square matrices $MM^\dagger$ and $M^\dagger M$ of sizes $N+\bar{\upnu}$ and $N$, respectively.
Here $M^\dagger$ denotes the Hermitian conjugate, transpose, or quaternionic conjugate of $M$, depending respectively on whether one is interested in unitary, orthogonal, or symplectic ensembles. 
In either case, the rank of both of the resulting square matrices is just $N$ as inherited from $M$, implying that only $M^\dagger M$ is generically full-rank, whereas $MM^\dagger$ is always singular and deterministically has $\bar{\upnu}$ zero eigenvalues.
Despite this qualitative difference, the positive eigenvalues of both $MM^\dagger$ and $M^\dagger M$ are simply the squares of the singular values of $M$, and thus identically distributed. 
In particular, their measure is the same and given by \cref{eq:AZmeasure2} with the appropriate value of $\upbeta$ and $\upalpha$ as specified by \cref{eq:alphanu}.
When addressing black hole spectra, this fact will help explain why the onset of the continuous part of the spectrum in \cref{eq:susy2,eq:rq4} is separated from zero by a positive gap, not only for supermultiplets with BPS states, but also for purely non-BPS ones.

Because the measures of all matrix ensembles of interest end up looking like \cref{eq:WDmeasure,eq:AZmeasure2}, we will henceforth take the Jacobian determinant to have the following general form:\footnote{\,Overall constants can always be reabsorbed into a redefinition of $\mathcal{Z}$ and will thus be ignored (cf. \cref{eq:AZmeasure2}).}
\begin{equation}
\label{eq:genjac}
    \abs{\det J(\Lambda)} = \abs{\Delta(\Lambda)}^\upbeta \, \prod_k |\lambda_k|^\upnu.
\end{equation}
This reproduces all invariant ensemble measures of interest as follows:
\begin{equation}
\label{eq:params}
    \begin{cases}
        \text{WD:} \quad & \upnu = 0,\\
        \text{AZ:} \quad & \upnu = \frac{\upalpha-1}{2},\\
        \text{BPS:} \quad & \upnu = \frac{\upbeta}{2} (\bar{\upnu}+1)-1,\\
    \end{cases}
\end{equation}
where $(\upalpha,\upbeta)$ are the standard parameters determined respectively by the discrete and continuous symmetry groups of the corresponding WD or AZ ensemble, and $\bar{\upnu}$ is the parameter specifying the difference in rank between groups in the bifundamental of the corresponding BPS ensemble. Despite all WD, AZ, and BPS ensembles being accounted for by \cref{eq:genjac}, their spectral statistics end up being qualitatively very different and rich in structure.

Let us conclude this section by noting that the construction of square matrices above in terms of a random rectangular matrix $M$ is a well-studied one in random matrix theory.
In particular, the full-rank $M^\dagger M$ defines what is traditionally known as a Wishart matrix \cite{wishart1928generalised,pastur2011eigenvalue}, and $M M^\dagger$ is often referred to as a singular Wishart matrix \cite{singwishart}. For this reason, BPS ensembles may be understood as generalizing certain AZ ensembles by a Wishart construction. Note, however, that the literature on Wishart ensembles usually defines $M$ by making its entries follow Gaussian distributions. In contrast, our BPS ensembles allow for arbitrary non-Gaussian matrix potentials.

\subsection{Random Hamiltonians}
\label{ssec:hamran}

The physical motivation for considering matrix ensembles is to study the spectrum of random Hamiltonians. 
Non-supersymmetric theories can be described by WD ensembles where the random matrix $M$ itself is identified as the Hamiltonian,
\begin{equation}
\label{eq:WDH}
    H \equiv M.
\end{equation}
Correspondingly, the eigenvalue measure for such Hamiltonians is given by \cref{eq:genjac} with $\upnu=0$.

Supersymmetric theories require the richer structure of AZ and BPS ensembles. 
To make the discussion of supersymmetric Hamiltonians as general as possible, it is thus convenient to allow for the random matrix $M$ to be $(N+\bar{\upnu})\times N$ for any integer $\bar{\upnu}\geq0$. This way, both the square matrices of AZ ensembles $(\bar{\upnu}=0)$ and the rectangular matrices of BPS ensembles $(\bar{\upnu}>0)$ are accounted for.
The algebraic structure of supersymmetry makes supercharges the natural objects to be actually identified with $M$. 
As usual in supersymmetric theories, the Hamiltonian is then constructed from anti-commutators of conjugate supercharges which here involve $M$ and $M^\dagger$ matrices.
The resulting random Hamiltonian naturally consists of the Wishart combinations $M M^\dagger$ and $M^\dagger M$ described in the previous section. More specifically, the Hamiltonian generally decomposes into statistically independent supermultiplets, each described by its own ensemble. Hence the overall spectrum of the theory can be obtained by combining Hamiltonians acting on supermultiplets of the form
\begin{equation}
\label{eq:hcool}
    H =
    \begin{pmatrix}
         MM^\dagger & 0 \\
         0 & M^\dagger M
    \end{pmatrix}.
\end{equation}
A detailed discussion of how this structure arises in supersymmetric random matrix theory is provided in \cref{sec:susyrmt}. For both of the positive semidefinite sub-matrices in \cref{eq:hcool}, the spectrum of positive eigenvalues is governed by the measure in \cref{eq:genjac} with the corresponding AZ and BPS parameters given by \cref{eq:params}. The difference between the two types of ensembles is that only the latter are able to actually make $MM^\dagger$ singular with a zero eigenvalue of arbitrarily high degeneracy $\bar{\upnu}$. Interestingly, this degenerate zero eigenvalue does not always capture zero-energy BPS states. As we show in \cref{sec:susyrmt} and explain below, the parameter $\bar{\upnu}$ acquires two possible meanings in supersymmetric random matrix theory.

With $\mathcal{N}=1$ supersymmetry, \cref{eq:hcool} captures the full Hamiltonian of the theory and thus determine its spectrum. As a result, in this case the $\bar{\upnu}$ zero eigenvalues that $M M^\dagger$ has genuinely correspond to BPS ground states. This is the random matrix mechanism behind the non-generic spectral ensembles alluded to in \cref{ssec:qexp} and giving rise to a discrete ground state in \cref{eq:ann1exp}. Indeed, this rectangular matrix construction results in a singularly distributed ground-state energy which in the spectral density becomes $\bar{\upnu} \, \delta(E)$ (cf. the near-BPS spectra in \cref{sec:bhspec}). The other $N$ positive eigenvalues of $M M^\dagger$ are equal to those of $M^\dagger M$, thus giving a $2$-fold degenerate spectrum of non-BPS states. This degeneracy is the result of pairs of non-BPS states of the theory forming irreducible doublets. In other words, the $2N$ positive eigenvalues correspond to a (reducible) supermultiplet of $N$ doublets consisting of pairs of non-BPS states with the same energy (see \cref{sec:susyrmt} for more details). One may think of the full Hilbert space of non-BPS states of the $\mathcal{N}=1$ as forming a single supermultiplet.

With $\mathcal{N}=2$ supersymmetry and higher, the situation is different. In particular, the supersymmetry algebra now has a non-trivial R-symmetry which allows one to organize the Hilbert space into different representations. 
Let us for simplicity here assume that states have R-charge $k\in\mathbb{Z}$ and the supercharges $Q$ and $Q^\dagger$ have $\pm1$ R-charge.\footnote{\,The general case is addressed in \cref{sec:susyrmt}; our simplified discussion here corresponds to $\hat{q}=1$ and $\delta=0$.}
Then the Hilbert space reads $\mathcal{H}=\bigoplus_k \mathcal{H}_k$, where $\mathcal{H}_k$ is the subspace of states with R-charge $k$. The restriction of the supercharges to these subspaces gives $Q_k : \mathcal{H}_k \to \mathcal{H}_{k+1}$ and $Q_k^\dagger : \mathcal{H}_{k+1} \to \mathcal{H}_{k}$. This way, the supercharges decompose into $Q=\sum_k Q_k$ and $Q^\dagger=\sum_k Q_k^\dagger$, and the Hamiltonian $H=\{Q,Q^\dagger\}$ becomes $H = \sum_k H_k$ with $Q_kQ_k^\dagger + Q_k^\dagger Q_k$. As it turns out, every $Q_k$ describes the spectrum of a statistically independent supermultiplet involving states with R-charge values $(k,k+1)$. What \cref{eq:hcool} captures in this setting is the spectrum of a single supermultiplet, not of the full Hamiltonian. In other words, every $H_k$ term in $H$ is described by a different ensemble realization of \cref{eq:hcool} with the supercharge operator $Q_k$ identified as $Q_k\sim M$ in terms of the random matrix $M$.

The important consequence this has on the meaning of $\bar{\upnu}$ can be easily described. Consider some state $\psi_{k+1}\in\mathcal{H}_{k+1}$.\footnote{\,With hindsight, we are just choosing $\psi_{k+1}$ rather than $\psi_k$ in order to keep the convention that $M$ has $\bar{\upnu}\geq0$.} If $\psi_{k+1}$ is a BPS state, then it is annihilated by both supercharges and, in particular, $Q_k^\dagger \psi_{k+1} = 0$. In terms of the random matrix $M\sim Q_k$, this means that $\psi_k$ is in the kernel of $M M^\dagger$ and thus contributes to a positive value $\bar{\upnu}>0$. This is the familiar case we already encountered where $\bar{\upnu}$ is naturally associated to the number of BPS states. However, suppose now that $\psi_{k+1}$ is a non-BPS state, and recall that supercharges are exact operators satisfying $Q^2={Q^\dagger}^2=0$. This supersymmetry condition implies that $\psi_{k+1}$ must be annihilated by at least one supercharge,\footnote{\,If $Q^\dagger\psi_{k+1} = \lambda \psi_k$ with $\lambda\neq0$ then $\lambda Q\psi_k = \psi_{k+1}$, so $Q\psi_{k+1} = 0$; similarly, $Q\psi_{k+1} \neq 0$ would imply $Q^\dagger\psi_{k+1} =0$.} while the assumption that it is non-BPS means that $\psi_k$ must be annihilated by only one supercharge.
Out of the two cases, consider first $Q\psi_{k+1}=0$, which in particular means $Q_k^\dagger\psi_{k+1} \neq 0$. For $M\sim Q_k$ this implies $\psi_{k+1}$ is not in the kernel of $M M^\dagger$, and thus we learn nothing about $\bar{\upnu}$ and its meaning. Hence consider now the case in which $Q^\dagger\psi_{k+1}=0$, such that we have $Q_k^\dagger\psi_{k+1} = 0$. Then $M\sim Q_k$ actually does make the kernel of $M M^\dagger$ non-trivial and $\bar{\upnu}>0$. In other words, here we see the non-BPS state $\psi_{k+1}$ contributing to $\bar{\upnu}$ simply because it is annihilated by $Q_k$, even though it is not BPS.\footnote{\,There is actually a subtle interplay between actual random matrix degrees of freedom and the supersymmetry constraints which make not every such non-BPS state contribute to $\bar{\upnu}$ \cite{Turiaci:2023jfa}. As we explain in \cref{sec:susyrmt}, the total contribution of non-BPS states to $\bar{\upnu}$ for $Q_k\sim M$ ends up depending on the number of non-BPS states with R-charge $k$ that are in $(k-1,k)$ doublets, as well as on the number of non-BPS states with R-charge $k+1$ that are in $(k+1,k+2)$ doublets. Namely, in both cases, these are non-BPS states which could in principle be in $(k,k+1)$, but are not. The minimum number out of the two classes is what determines $\bar{\upnu}$ (see \cref{sec:susyrmt} for more details).} Indeed, $\psi_{k+1}$ gives a zero eigenvalue for the supermultiplet Hamiltonian $H_k$, but it would give a non-zero eigenvalue for $H_{k+1}$. The matrix ensemble description of $H_k$ hence feels a zero eigenvalue from $\psi_{k+1}$, even though there is no zero-energy BPS states associated to it in the theory.

In summary, we see that in supersymmetric random matrix theory \cref{eq:hcool} generally captures the spectrum of single supermultiplets, and that its $\bar{\upnu}$ zero eigenvalues may or may not be associated to BPS states. In the thorougher discussion from \cref{sec:susyrmt}, we find that $\bar{\upnu}$ is always associated to either one or the other. In particular, when there are BPS states associated to the supermultiplet R-charges, $\bar{\upnu}$ is precisely the number of such BPS states. Otherwise, $\bar{\upnu}$ can generally be associated to the varying number of states in non-BPS supermultiplets. More explicitly, if $N_k^+$ is the number of non-BPS states in $(k,k+1)$ doublets and say $N_{k-1}^+ < N_k^+ < N_{k+1}^+$, then the random matrix description for $Q_k\sim M$ turns out to involve  $\bar{\upnu}=N_{k-1}^+$.

Regarding the black hole spectra from \cref{sec:bhspec}, the relation to random matrices is already rather manifest.
The random Hamiltonians of WD ensembles describe non-supersymmetric theories like the bosonic Schwarzian theory.
The seven standard AZ ensembles can describe supersymmetric Hamiltonians of theories with no unbroken supersymmetry, and thus are particularly suited for the $\mathcal{N}=1$ super-Schwarzian theory, or the specially anomalous $q=0$ gapless supermultiplet of the $\mathcal{N}=2$ super-Schwarzian theory with $\delta=1/2$ and $\hat{q}=1$. Finally, the richer Wishart matrix structure of BPS ensembles allows to describe the more general cases of the $\mathcal{N}=2,4$ super-Schwarzian theories, both for supermultiplets with and without BPS states. In particular, the $\bar{\upnu}$ parameter of the latter will turn out to be crucial to explain not only the spectral gap when there are BPS states, but also for purely non-BPS sectors where the onset of the continuous spectrum continues to rise to higher energies for higher supermultiplets (see \cref{sec:gaps}).

\subsection{Equilibrium Conditions}

It is often convenient to raise the Jacobian determinant in \cref{eq:canonicalrm} into the exponential so as to write the matrix partition function as an action integral,
\begin{equation}
\label{eq:Iint}
    \mathcal{Z} = \int\displaylimits_{\mathbb{R}^N} d\Lambda \, e^{-N I(\Lambda)}, \qquad I(\Lambda) \equiv \Tr \left( V(\Lambda) - \frac{1}{N} \log\abs{ J(\Lambda) } \right).
\end{equation}
where we have used the matrix identity $\log\det M = \Tr\log M$ to rewrite the  $J(\Lambda)$ contribution. Correspondingly, the joint eigenvalue PDF in terms of this action reads
\begin{equation}
\label{eq:pIZ}
    p(\Lambda) = \frac{e^{-N I(\Lambda)}}{\mathcal{Z}}.
\end{equation}
For the general Jacobian in \cref{eq:genjac}, the matrix action can be written out to be
\begin{equation}
\label{eq:Iout}
    I(\Lambda) = \sum_{i=1}^N V_\upnu(\lambda_i) - \frac{\upbeta}{N} \sum_{\underset{j>i}{i,j=1}}^N \log \, \abs{\lambda_j - \lambda_i},
\end{equation}
where we have absorbed the $\upnu$-dependent term into a redefinition of the matrix potential to
\begin{equation}
\label{eq:redef}
    V_\upnu(\lambda) \equiv V(\lambda) - \frac{\upnu}{N} \log|\lambda|.
\end{equation}
Minima of $I(\Lambda)$ are known as Fekete points and obey the equilibrium conditions
\begin{equation}
\label{eq:fekete}
    \evalat{\frac{\partial I(\Lambda)}{\partial \lambda_k} }{\Lambda=\Lambda_*} = 0, \qquad \forall k = 1,\dots,N.
\end{equation}
These points, which dominate the matrix integral at large $N$, motivate the notion of an effective potential characterizing the interaction of each eigenvalue with all others. Such an effective potential $\widehat{V}$ for the eigenvalue $\lambda_k$ can be easily constructed by simply keeping those terms in $I(\Lambda)$ which depend on $\lambda_k$. Applying this to \cref{eq:Iout},
\begin{equation}
\label{eq:fpotwi}
    \widehat{V}_\Lambda(\lambda_k) \equiv V_\upnu(\lambda_k) - \frac{\upbeta}{N} \sum_{k \neq i = 1}^N \log \, \abs{\lambda_k - \lambda_i},
\end{equation}
where the $\Lambda$ subscript emphasizes the dependence of $\widehat{V}_\Lambda$ on all eigenvalues. Upon differentiation,
\begin{equation}
\label{eq:effprime}
    \widehat{V}_\Lambda'(\lambda_k) = V_\upnu'(\lambda_k) - \frac{\upbeta}{N} \sum_{k \neq i = 1}^N \frac{1}{\lambda_k - \lambda_i}.
\end{equation}
By construction, this gives the same equilibrium conditions as \cref{eq:fekete} and thus
\begin{equation}
\label{eq:feketeV}
    \evalat{\widehat{V}_{\Lambda}'(\lambda_k)}{\Lambda=\Lambda_*} = 0, \qquad \forall k =1,\dots,N,
\end{equation}
at Fekete points.
A crucial property of the effective potential is that it exists even when the action does not, as happens in \cref{ssec:doubscal} for the double scaling limit relevant to gravity.\footnote{\,\label{fn:effV}Note that $I(\Lambda) \neq \sum_{k=1}^N \widehat{V}_\Lambda(\lambda_k)$, so the action is not just a sum over the effective potential for every eigenvalue. Basically $\widehat{V}_\Lambda(\lambda_k)$ is only specified by the Fekete condition up to an arbitrary function of all eigenvalues but $\lambda_k$.}

Of course, the system of $N$ saddle-point equations that \cref{eq:feketeV} defines for a general matrix ensemble is highly non-trivial. A standard procedure in random matrix theory to evaluate such large-$N$ matrix integrals, whether it is at a saddle-point level or not, is to first take a continuum limit. In the continuum, the Fekete conditions in \cref{eq:feketeV} for effective potentials like \cref{eq:fpotwi} become well-studied differential equations which can be solved exactly.

\subsection{Continuum Formalism}
\label{ssec:contl}

The continuum limit can be easily obtained by first turning summations into integrals via
\begin{equation}
\label{eq:discon}
    \sum_{k=1}^N (\,\cdot\,) = N \int\displaylimits_{\mathbb{R}} dx \, \hat{\rho}_\Lambda(x)\, (\,\cdot\,), \qquad \hat{\rho}_\Lambda(x) \equiv \frac{1}{N} \sum_{k=1}^N \delta(x - \lambda_k).
\end{equation}
In the continuum, the discrete measure $\hat{\rho}_\Lambda$ is replaced by a continuous measure $\hat{\rho}$, in terms of which the discrete $\lambda_k$ eigenvalues are now captured by a continuous $x$ variable.
Doing so, $\hat{\rho}$ takes on the dynamical role of the variables $\Lambda$, thereby turning the matrix integral from an eigenvalue integral to a functional integral over unit-normalized measures,
\begin{equation}
\label{eq:gencon}
    \int\displaylimits_{\mathbb{R}^N} d\Lambda \quad \longrightarrow \quad  N \int \mathcal{D}\hat{\rho}.
\end{equation}
Applying this to \cref{eq:Iint} leads to the continuum limit
\begin{equation}
\label{eq:contZ}
    \mathcal{Z} =  N \int \mathcal{D} \hat{\rho} \, e^{- N^2 \widehat{I}[\hat{\rho}]},
\end{equation}
where the action $I(\Lambda)$ in \cref{eq:Iout} has been replaced by a functional $N \widehat{I}[\hat{\rho}]$ given by
\begin{equation}
\label{eq:contIm}
    \widehat{I}[\hat{\rho}] \equiv \int\displaylimits_{\mathbb{R}} dx \, \hat{\rho}(x) \, V_\upnu(x) - \frac{\upbeta}{2} \int\displaylimits_{\mathbb{R}} dx \, dy \, \hat{\rho}(x) \, \hat{\rho}(y) \, \log|x-y|.
\end{equation}
Similarly, in the continuum the effective potential in \cref{eq:fpotwi} becomes
\begin{equation}
\label{eq:effwd}
    \widehat{V}[\hat{\rho};x] = V_\upnu(x) - \upbeta \int\displaylimits_{\mathbb{R}} dy \, \hat{\rho}(y) \log|x-y|,
\end{equation}
where the notation emphasizes that the effective potential $\widehat{V}$ itself depends not only on the eigenvalue $x$, but also on all the other eigenvalues through the $\hat{\rho}$ measure (cf. $\widehat{V}_\Lambda$ in the discrete case).
Differentiating \cref{eq:effwd} we also obtain the continuum version of \cref{eq:effprime},
\begin{equation}
\label{eq:dveffr}
    \partial_x \widehat{V}[\hat{\rho};x] = V_\upnu'(x) - \upbeta \fint\displaylimits_{\mathbb{R}} dy \, \frac{\hat{\rho}(y)}{x-y},
\end{equation}
where $\fint$ denotes a principal value integral. 
More abstractly, in general $\widehat{V}$ is defined in the continuum by functional variation as
\begin{equation}
\label{eq:effV}
    \widehat{V}[\hat{\rho};x] \equiv \frac{\delta \widehat{I}[\hat{\rho}]}{\delta \hat{\rho}(x)}.
\end{equation}
The extremization of the action for a saddle-point approximation of the integral in \cref{eq:contZ} thus clearly involves the effective potential. Since such an extremization is over measures satisfying the normalization condition $\int_{\mathbb{R}} dx\,\hat{\rho}(x) = 1$, the variational problem is better solved by introducing a Lagrange multiplier for this constraint. The upshot is that extremality of $\widehat{I}[\hat{\rho}]$ against functional variations requires $\widehat{V}$ to be constant along the support of $\hat{\rho}$ or, equivalently,
\begin{equation}
\label{eq:deffV0}
    \partial_x \widehat{V}[\hat{\rho}_*;x] = 0, \qquad \forall x\in\supp\hat{\rho}_*,
\end{equation}
for $\hat{\rho}=\hat{\rho}_*$ a unit-normalized saddle point of $\widehat{I}[\hat{\rho}]$.
This is the continuum form of the Fekete conditions in \cref{eq:feketeV}. The equilibrium measure $\hat{\rho}_*$ that \cref{eq:deffV0} defines is recognized as the leading average spectral density of the matrix ensemble at large $N$, 
\begin{equation}
\label{eq:rhoslim}
    \hat{\rho}_*(x) \equiv \lim_{N\to\infty} \frac{\langle \rho (x) \rangle}{N}.
\end{equation}
The equilibrium condition in \cref{eq:deffV0} can be made more explicit using \cref{eq:dveffr},
\begin{equation}
\label{eq:expcond}
    \fint\displaylimits_{\mathbb{R}} dy \, \frac{\hat{\rho}_*(y)}{x-y} = \frac{V_\upnu'(x)}{\upbeta}, \qquad \forall x \in \supp\hat{\rho}_*.
\end{equation}
Obtaining the leading spectral density thus amounts to finding a function $\hat{\rho}_*$ that satisfies \cref{eq:expcond}. This is a type of singular integral equation whose solutions remarkably admit a closed-form expression in terms of a general matrix potential.
The derivation of such a general solution for $\hat{\rho}_*$ is reproduced in \cref{sec:gensol}. The final result is
\begin{equation}
\label{eq:rho0main}
    \hat{\rho}_*(x) = \frac{\sqrt[+]{\eta(x)}}{\upbeta \, \pi^2} \fint\displaylimits_\Sigma \frac{dy}{y-x} \frac{V_\upnu'(y)}{\sqrt[+]{\eta(y)}}, \qquad x\in \Sigma,
\end{equation}
where $\Sigma=\supp\hat{\rho}_*$ is specified by \cref{eq:support} and the conditions in \cref{eq:endpcons,eq:endp2}, and the function $\eta$ is given by \cref{eq:sigmawd,eq:sigmaaz} for WD and AZ ensembles, respectively. The notation $\sqrt[+]{\cdot}$ instructs one to approach the square-root branch cut from above as in \cref{eq:bcutapp}.

Having $\hat{\rho}_*$, it is useful to define the leading effective potential that appears in \cref{eq:deffV0} by
\begin{equation}
\label{eq:leadeff}
    \widehat{V}_*(x) \equiv \widehat{V}[\hat{\rho}_*;x].
\end{equation}
Besides vanishing along $\supp\hat{\rho}_*$ by \cref{eq:deffV0}, the gradient of $\widehat{V}_*$ also defines an important object upon analytic continuation to $\mathbb{C}$ off $\supp\hat{\rho}_*$ known as the spectral curve,
\begin{equation}
\label{eq:speyc}
    y(z) = -\frac{\widehat{V}_*'(z)}{\upbeta}, \qquad z\in\mathbb{C}\smallsetminus\supp\hat{\rho}_*.
\end{equation}
As we show in \cref{ssec:speccur}, $y$ has important analytic properties and is intimately related to $\hat{\rho}_*$ itself. In particular, denoting the analytic continuation of $\hat{\rho}_*$ off its support to $\mathbb{C}$ by $\hat{\rho}_*^c$,
\begin{equation}
\label{eq:rhotoy}
    y(z) = - i\pi \hat{\rho}_*^c(z), \qquad z \in \mathbb{C} \smallsetminus \supp\hat{\rho}_*.
\end{equation}
Furthermore, approaching the branch cut of $y$ along $\supp\hat{\rho}_*$ from opposite sides leads to
\begin{equation}
\label{eq:ytorho}
    \hat{\rho}_*(x) = \pm \frac{i}{\pi} \lim_{\epsilon\to0^{+}} y(x \pm i\epsilon), \qquad x\in\supp\hat{\rho}_*.
\end{equation}
In words, \cref{eq:rhotoy} says that $y$ is fully determined by $\hat{\rho}_*$, while \cref{eq:ytorho} says that $\hat{\rho}_*$ is fully determined by $y$.
The fact that both objects contain equivalent information about the spectral ensemble will turn out to be very useful for us.

\subsection{Matrix Scalings}
\label{ssec:mscal}

Matrix scalings refer to large-$N$ limits which amplify certain features of the spectral statistics of a matrix ensemble. These scalings may be single $N\to\infty$ limits that zoom into the bulk \cite{erdos2009bulk} or the edge \cite{BOWICK199121} of the spectrum, or double limits in which the parameters of the matrix integral are simultaneously tuned towards critical values as $N\to\infty$ \cite{Gross:1989vs,Brezin:1990rb,Douglas:1989ve}.

As part of these procedures, one generally introduces a zooming variable $\bar{x}$ that replaces the continuum matrix eigenvalue variable $x$. This is done by setting $x \equiv f_N(\bar{x})$ and keeping $\bar{x}\sim O(1)$ as $N\to\infty$, where the choice of scaling function $f_N$ determines which spectral features the zooming variable captures in the limit. The resulting construct describes a spectral ensemble in its own right, capturing the statistics of the effective eigenvalue variable $\bar{x}$ in the original matrix ensemble at large $N$. Certain matrix ensembles are actually more naturally defined in their scaled version and suitable parameters therein, with no allusion to finite matrices.

When constructing Hamiltonians out of random matrices as in \cref{ssec:hamran}, we are faced with a choice to either interpret $x$ or $\bar{x}$ as the random variable associated to eigenvalues of the Hamiltonian. If the Hamiltonian is naturally identified with the random matrix before any scaling, then $x$ is the desired variable. For instance, if the Hamiltonian is the random matrix itself, a Boltzmann factor would read $e^{-\beta x}$.
However, if we are interested in the statistics of a specific zooming variable, then applying the scaling after would lead to $e^{-\beta f_N(\bar{x})}$. This would defeat the purpose, since such a Boltzmann factor would remain sensitive to the statistics of $x$ rather than of $\bar{x}$ in the limit.
Instead, we will be addressing situations where the Hamiltonian spectrum is naturally associated to whichever zooming variables are relevant to the problem, and thus particularly sensitive to the spectral statistics that these highlight. In other words, we take scaled matrix ensembles as the starting point when constructing spectral ensembles, and take $\bar{x}$ as the only eigenvalue variable that matters. This way, if e.g. the Hamiltonian is the random matrix itself, a Boltzmann factor would read $e^{-\beta \bar{x}}$ and thus successfully capture the statistics of $\bar{x}$.

The above implies that in general there will be no simple mapping between quantities involving Hamiltonians for matrix ensembles before and after scaling, or between different scalings. Namely, Hamiltonian spectra defined in terms of the inequivalent zooming variables $x$ and $\bar{x}$ will give rise to different large-$N$ physics which we will not be interested in relating. Since matrix integrals in different scaling limits will generally be treated as different constructs, it will be simpler to think of matrix scalings as replacements of the form $x\to f_N(x)$ before taking the $N\to\infty$ limit. This way it will not be necessary to deal with the notation $\bar{x}$ for zooming variables, and we can unambiguously use $x$ to refer to whichever eigenvalue variable remains $O(1)$ in the large-$N$ limit. 
Let us now describe in more detail the matrix scalings relevant to this paper.

\subsubsection{Single Limits}
\label{sssec:zoom}

We begin by addressing single limits in which all parameters of the matrix integral are kept fixed as $N\to\infty$.
Already the factor of $N$ in \cref{eq:matrixZ} was introduced with hindsight as a form of matrix scaling. Namely, this factor guarantees that for potentials with a quadratic term the typical smallest and largest eigenvalues are of $O(1)$ size in $x$. In other words, had we not introduced this factor of $N$, the size of these eigenvalues would be extensive and $O(N^{1/2})$.
Keeping $x\sim O(1)$ at large $N$ guarantees that the leading spectral density given by \cref{eq:rhoslim} converges to a bounded support of finite size in $x$ in the limit. For instance, this large-$N$ limit for Gaussian WD matrices leads to the well-known Wigner semicircle law,
\begin{equation}
\label{eq:wdscal}
    \hat{\rho}_*(x) \reprel{GWD}{=} \frac{1}{\pi\upbeta} \sqrt{2\upbeta - x^2}, \qquad -2\upbeta < x < 2\upbeta,
\end{equation}
which is illustrated in \cref{fig:glue}.
Given that there are $N$ eigenvalues in total, a consequence of such a scaling is that in the $x$ variable the typical spacing between eigenvalues is $O(N^{-1})$. Since this goes to zero at large $N$, the variable $x$ is clearly not suitable for capturing eigenvalue fluctuations. A better zooming variable would be obtained by replacing $x\to x/N$, in which case the typical separation between eigenvalues would be expected to remain $O(1)$ at large $N$. This heuristic argument gives the provably correct scaling one should perform in order to study the statistics of eigenvalues in the bulk of the spectrum.

When studying random Hamiltonians, such a bulk scaling would focus on the statistics of the energies of generic excited states. However, by the observations in \cref{sec:specen}, we are actually more interested in understanding the spectral statistics of energies near the ground state. 
This corresponds to studying the tail of eigenvalues near the lower edge $\inf\supp\hat{\rho}_*$ of the leading spectrum.
Interestingly, as it turns out the eigenvalue fluctuations near the edge are often larger than $O(N^{-1})$. For e.g. the WD ensembles the typical fluctuations around the $x=-\sqrt{2\upbeta}$ edge are $O(N^{-2/3})$, and would thus blow up in the large-$N$ limit that focuses on bulk statistics. As a result, studying edge statistics generally requires both a precise shift to the desired edge and a subtler scaling in $N$ (see \cref{ssec:ortho} for explicit examples). 

Edge scalings generally result in non-normalizable spectral functions in the $N\to\infty$ limit. When zooming on spectral edges it is thus useful to introduce an eigenvalue density parameter $e^{S_0}$ that effectively replaces the diverging total eigenvalue number $N$. This can be done by simply using $N e^{-S_0}$ instead of just $N$ in the appropriate zooming variable.
For the GUE this procedure defines the Airy matrix model, whose exact spectral density in $e^{S_0}$ is \cite{Tracy:1992kc}\footnote{\,\label{fn:ncorr}In the random matrix literature, this is a large-$N$ result for the GUE, which of course receives corrections away from the strict limit that are subleading in $N$. In certain contexts in physics, however, this leading result at large $N$ is interpreted as an exact spectral density in the new parameter $e^{S_0}$, which can then be broken down into perturbative and non-perturbative corrections at large $e^{S_0}$. See \cref{ssec:doubscal} and \cref{ssec:ortho} for more details.}
\begin{equation}
\label{eq:airyf}
    \hat{\rho}^{\smalltext{Airy}}(x) = e^{-S_0/3} \left( \text{Ai}'\left( - e^{2S_0/3}x \right)^2 + e^{2S_0/3}x \, \text{Ai}\left( - e^{2S_0/3}x \right)^2 \right), \qquad x\in\mathbb{R}.
\end{equation}
At large $e^{S_0}$ one obtains the leading spectral density of the Airy model,
\begin{equation}
\label{eq:airedge}
    \hat{\rho}^{\smalltext{Airy}}_*(x) \equiv \lim_{S_0\to\infty} \hat{\rho}^{\smalltext{Airy}}(x) = \frac{\sqrt{x}}{\pi}, \qquad x>0,
\end{equation}
which is illustrated against \cref{eq:airyf} in \cref{fig:glue}.
This precisely matches the behavior of the leading spectral density of the GUE, given by \cref{eq:wdscal} for $\upbeta=2$, near the $x=-2$ edge of the spectrum.

\begin{figure}
    \centering
    \includegraphics[width=0.45\textwidth]{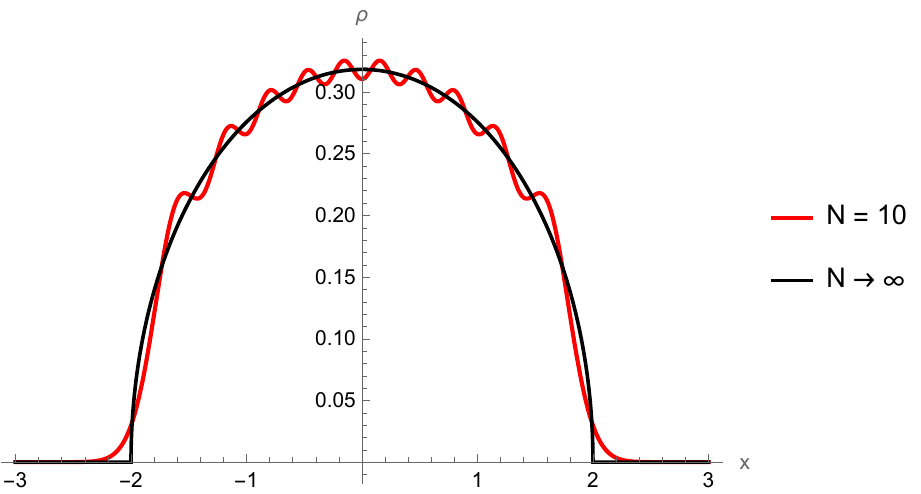}
    ~~
    \includegraphics[width=0.45\textwidth]{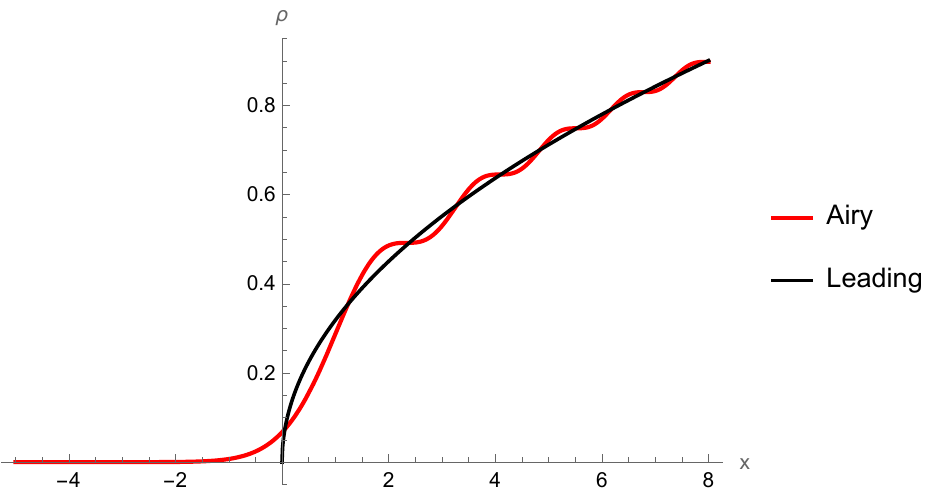}
    \caption{Illustration of spectral densities for the Gaussian WD ensemble with $\upbeta=2$, i.e., the GUE. Left: scaling to a compact distribution, with the $N=10$ result (red) obtained via the orthogonal polynomials techniques from \cref{ssec:ortho}, and the $N\to\infty$ limit (black) corresponding to the Wigner semicircle law from \cref{eq:wdscal}. Right: scaling to the lower edge, with the exact $N\to\infty$ limit (red) given by the Airy model from \cref{eq:airyf}, and its leading soft edge (black) given by \cref{eq:airedge}.
    }
    \label{fig:glue}
\end{figure}

These scalings for Gaussian potentials are presented in more detail in \cref{ssec:ortho}, for all WD, AZ, and BPS ensembles. For the AZ case, the counterpart of \cref{eq:wdscal} is the singular Mar\v{c}enko-Pastur law \cite{Pastur:1967zca},
\begin{equation}
\label{eq:singaz}
    \hat{\rho}_*(x) \reprel{GAZ}{=} \frac{1}{2\pi\upbeta}\sqrt{\frac{4\upbeta-x}{x}}, \qquad 0<x<4\upbeta,
\end{equation}
which diverges as $x\to0^+$ (see \cref{fig:mpsing}).
This is in stark contrast with the regularity near the lower edge of the spectrum in the WD case.
Spectral edges where the spectral density goes to zero (or some finite value)  like in \cref{eq:wdscal} are often referred to as soft edges, and those where it diverges like in \cref{eq:singaz} as hard edges. As we explain in \cref{sssec:univ}, this turns out to give a universal characterization of the edge statistics of spectral ensembles.

Notice how \cref{eq:singaz} in this single large-$N$ limit has no dependence on the $\upnu$ parameter, even though it is typically nonzero for AZ ensembles. 
In the standard AZ ensembles $\upnu$ takes fixed $O(1)$ values and the leading result in \cref{eq:singaz} is unavoidable; in the BPS ensembles, however, an interestingly different story arises. The BPS case is explained in \cref{sec:gaps}.

The scaling that zooms onto the $x=0$ edge of the spectrum in \cref{eq:singaz} for $\upbeta=2$ defines a family of Bessel matrix models, whose exact spectral density in $e^{S_0}$ is \cite{Tracy:1993xj}
\begin{equation}
\label{eq:besself}
    \hat{\rho}^{\smalltext{Bessel}}_{\upnu}(x) = 
    \frac{e^{S_0}}{2} \left(J_{\upnu }\left(e^{S_0}\sqrt{x}\right){}^2-J_{\upnu +1}\left(e^{S_0}\sqrt{x}\right) J_{\upnu -1}\left(e^{S_0}\sqrt{x}\right)\right), \qquad x>0.
\end{equation}
At large $e^{S_0}$, this reproduces the hard edge of the leading spectral density in \cref{eq:singaz},
\begin{equation}
\label{eq:hardw}
    \hat{\rho}_*^{\smalltext{Bessel}}(x) = \lim_{S_0\to\infty} \hat{\rho}_{\upnu}^{\smalltext{Bessel}}(x) = \frac{1}{\pi\sqrt{x}}, \qquad x>0.
\end{equation}
Note, however, that this singular behavior only arises in the limit of \cref{eq:hardw}; in fact, \cref{eq:besself} gives a perfectly regular spectral density as $x\to0^+$ for any $\upnu\geq0$.
It is only for $\upnu=-1/2$, a specific case allowed by \cref{eq:params} for AZ ensembles with $\upalpha=0$, that \cref{eq:besself} actually is singular.
Explicitly, near $x=0$ one finds $\hat{\rho}^{\smalltext{Bessel}}_{\upnu}(x) \sim x^{\upnu}$, which vanishes in the $x\to0^+$ limit for $\upnu>0$, gives the finite value $e^{S_0}/2$ for $\upnu=0$, and diverges as $1/\sqrt{x}$ for $\upnu=-1/2$. In other words, the exact spectral density of the Bessel model generically does not actually diverge as $x\to0^+$, despite the hard edge that the leading result in \cref{eq:hardw} exhibits. The behavior $\hat{\rho}^{\smalltext{Bessel}}_{\upnu}(x) \sim x^{\upnu}$ actually suggests that the growth of these models off $x=0$ can be made arbitrarily slow by increasing the $\upnu$ parameter, which is certainly possible in BPS ensembles. Indeed, as we describe next, $\upnu$ turns out to soften the hard edge of these ensembles by pushing their leading spectrum off $x=0$ and creating a spectral gap $\Delta>0$.

\begin{figure}
    \centering
    \includegraphics[width=0.45\textwidth]{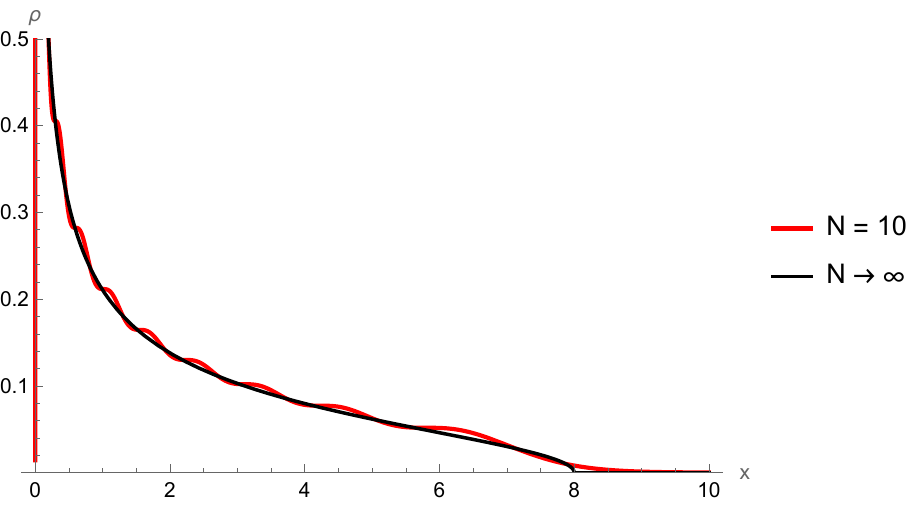}
    ~~
    \includegraphics[width=0.45\textwidth]{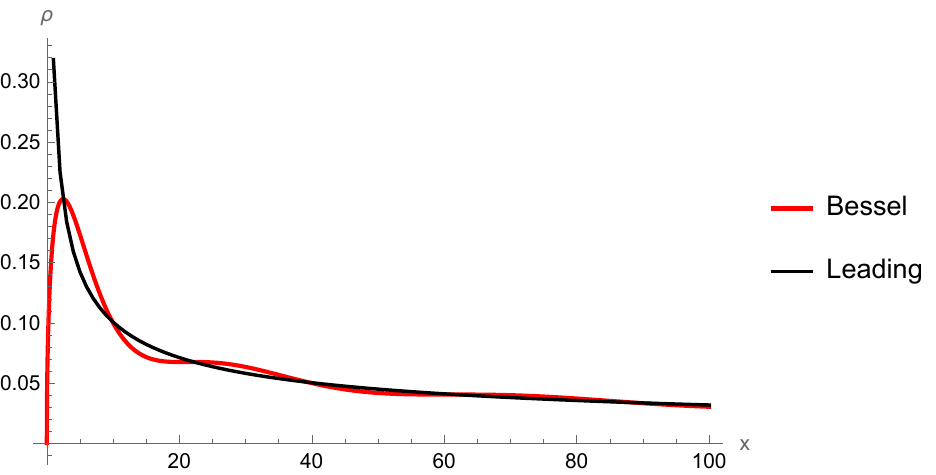}
    \caption{Illustration of spectral densities for the Gaussian AZ ensemble with $(\upalpha,\upbeta)=(2,2)$. Left: scaling to a compact distribution, with the $N=10$ result (red) obtained via the orthogonal polynomials techniques from \cref{ssec:ortho}, and the $N\to\infty$ limit (black) corresponding to the singular Mar\v{c}enko-Pastur law from \cref{eq:singaz}. Right: scaling to the lower edge, with the exact $N\to\infty$ limit (red) given by the Bessel model from \cref{eq:besself}, and its leading hard edge (black) given by \cref{eq:hardw}.
    }
    \label{fig:mpsing}
\end{figure}

\subsubsection{Spectral Gaps}
\label{sec:gaps}

Recall that by \cref{eq:params} in BPS ensembles $\upnu$ is linearly related to the free parameter $\bar{\upnu}\in\mathbb{N}$ which determines the size of the rectangular $(N+\bar{\upnu})\times N$ random matrix $M$. A natural large-$N$ scaling of rectangular matrices would preserve their aspect ratio $(N+\bar{\upnu})/N$, which requires $\bar{\upnu}$ to grow linear in $N$. For the Wishart matrices constructed out of $M$ as described in \cref{ssec:ensemmess}, this is indeed the standard scaling in the random matrix literature \cite{pastur2011eigenvalue,Perret_2016}.
Letting $\bar{\upnu}\sim O(N)$ makes $\upnu$ and $\bar{\upnu}$ coincide up to factors of $2$ from $\upbeta$ to leading order at large $N$. It is thus of interest for BPS ensembles to generally consider a scaling where $\upnu = \nu N$ with $\nu\sim O(1)$. 
Given that $\upnu\sim\bar{\upnu}$ at large $N$, for these ensembles we may at times refer to $\upnu$ itself as the degeneracy of the zero eigenvalue.

With this scaling, the potential $V_\upnu$ from \cref{eq:redef} does not trivialize to just $V$ at large $N$, but preserves a logarithmic pole at zero in the continuum,\footnote{\,The relevance of this pole for the emergence of a gap in the spectrum was already noted by \cite{Turiaci:2023jfa}. Here we additionally explore this gap quantitatively for general Wishart matrices, both singular and non-singular.}
\begin{equation}
\label{eq:vgapcon}
    V_\upnu(x) = V(x) - \nu \log|x|.
\end{equation}
Remarkably, this turns \cref{eq:singaz} into a regular Mar\v{c}enko-Pastur law,
\begin{equation}
\label{eq:gaznu}
    \hat{\rho}_*(x) \reprel{\!GBPS\!}{=} \frac{\sqrt{(a_+ - x)(x - a_-)}}{2\pi\upbeta x}, \qquad a_-<x< a_+,
\end{equation}
where the endpoints of the support interval are given by
\begin{equation}
\label{eq:endpoints}
    a_{\pm} = \upbeta \left(1 \pm \sqrt{1 + \frac{2 \nu }{\upbeta }}\right)^2.
\end{equation}
This is illustrated in \cref{fig:gaps}.
\begin{figure}
    \centering
    \includegraphics[width=0.45\textwidth]{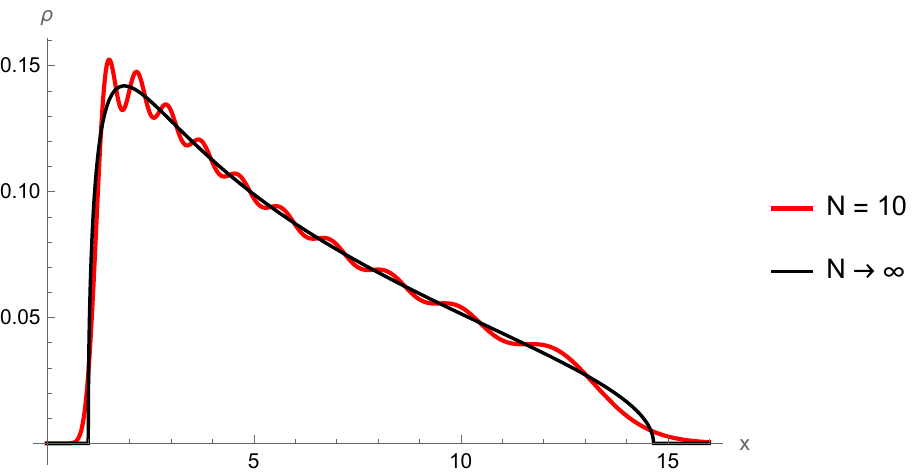}
    ~~
    \includegraphics[width=0.45\textwidth]{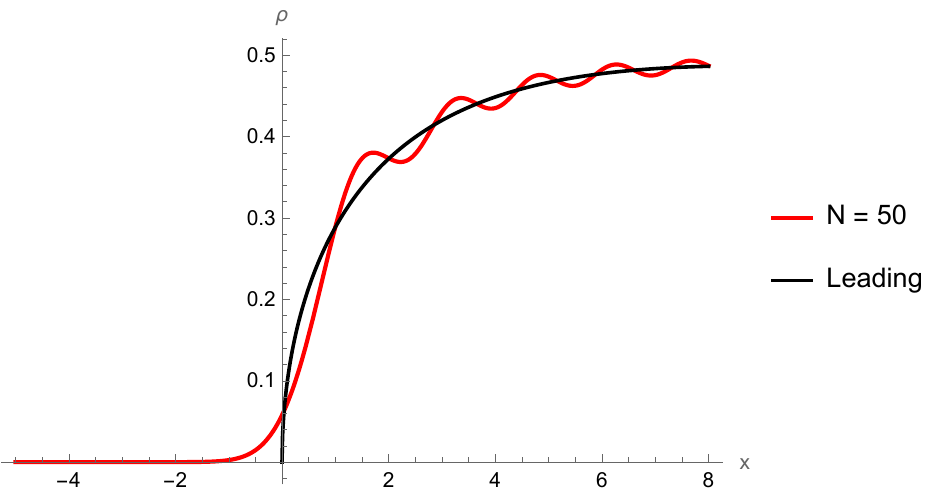}
    \caption{Illustration of spectral densities for the Gaussian BPS ensemble with $\upbeta=2$ and $\nu = \frac{1}{2}+\sqrt{2}$ such that \cref{eq:deltaam} gives gap size $\Delta=1$. Left: scaling to a compact distribution, with the $N=10$ result (red) obtained via the orthogonal polynomials techniques from \cref{ssec:ortho}, and the $N\to\infty$ limit (black) corresponding to the regular Mar\v{c}enko-Pastur law from \cref{eq:gaznu}. Right: scaling to the lower edge, with the exact $N=50$ result (red) approaching the Airy model from \cref{eq:airyf} as $N\to\infty$, and its leading soft edge (black) approaching \cref{eq:airedge} as $N\to\infty$ (cf. the strict Airy limit in \cref{fig:glue}).
    }
    \label{fig:gaps}
\end{figure}
Since \cref{eq:singaz} is recovered in the limit $\nu\to0$, we may generally refer to \cref{eq:gaznu} also when addressing standard AZ ensembles. In contrast with \cref{eq:singaz}, we see that \cref{eq:gaznu} is now non-singular and has no support on an open interval $(0,a_-)$ above zero for $\nu>0$.
Since zero is the smallest eigenvalue possible for the positive semidefinite random Hamiltonians in these ensembles, a nonempty interval of this kind constitutes a spectral gap.
This gap in the spectrum separates the discrete ground state with zero energy from all other higher-energy states. 
Its size is
\begin{equation}
\label{eq:deltaam}
    \Delta \equiv a_- = 
    \begin{cases}
        \frac{\nu^2}{\upbeta} + O(\nu^3), & \nu\ll1\\
        2\nu + O(1), & \nu\gg1\\
    \end{cases}
\end{equation}
where we have written out the scaling of $\Delta$ at small and large $\nu = \upnu/N$. In particular, note that the gap is never extensive in $N$, but determined by the ratio between the degeneracy of the zero eigenvalue and the number of excited states. The gap can be numerically seen to be populated by $O(1)$ eigenvalues in total as $N\to\infty$, thus becoming relatively depleted in the limit.
When the relative number of zero eigenvalues $\nu\to0$ at large $N$, corresponding to a non-extensive degeneracy, the spectrum becomes gapless.

In supersymmetric terms, this means that there arises a gap between the zero-energy ground state and the energy of the lightest non-BPS states if and only if $\bar{\upnu}$ is extensive in $N$. In supermultiplets where $\bar{\upnu}$ is the number of BPS states, this explains the origin of the spectral gap and predicts its size in terms of the BPS ground-state degeneracy.\footnote{We thank Misha Usatyuk for early conversations on this idea; a separate discussion just appeared on \cite{Johnson:2024tgg}.} 
In consistency with our observations here, we see that the $\mathcal{N}=2,4$ super-Schwarzians that arise in near-BPS black holes indeed exhibit an extensive degeneracy of BPS states given by $e^{S_0}$ when their spectrum is gapped. In particular, our results precisely reproduce the behavior of $\mathcal{N}=2$ supermultiplets in \cref{eq:susy2}: at small $q\ll \hat{q}$, the BPS degeneracy gives a ratio $\bar{\upnu}/e^{S_0} \approx q/\hat{q}$, and correspondingly the gap is seen to be quadratic in $q/\hat{q}$ (the identification $N\sim e^{S_0}$ will be explained shortly). 

Remarkably though, $\bar{\upnu}$ clearly must also be playing an important role even in purely non-BPS supermultiplets. Namely, the onset of the continuum in \cref{eq:susy2,eq:rq4} keeps growing for higher supermultiplets, even though these are unrelated to BPS states. This suggests that we should match the quadratic growth of $\Delta_q$ (or $\Delta_J$) in the R-charge $q$ (or $J$) to the large-gap limit of \cref{eq:deltaam}. This gives the identification $\nu\sim q^2$ asymptotically, and thus $\bar{\upnu} \sim q^2 e^{S_0}$. The interpretation of this extensive value of $\bar{\upnu}$ requires in this non-BPS case following the discussion in \cref{ssec:hamran} and results of \cref{sec:susyrmt}. In particular, the quadratic growth of $\bar{\upnu}$ in $q$ must be related to a similar growth in the number of non-BPS states in higher-$q$ supermultiplets. Letting $L_q e^{S_0}$ be the number of non-BPS states in a $q$-multiplet, explicitly we would expect $L_q \sim q^2$ (see comments before \cref{eq:n4ham}).

Another important consequence of the $\upnu\sim O(N)$ scaling is that now the leading spectral density in \cref{eq:gaznu} behaves regularly near both edges of the spectrum. In particular, \cref{eq:gaznu} has a soft edge at $x=\Delta$ where $\hat{\rho}_*$ just goes to zero. The hard edge of the AZ case in \cref{eq:singaz} can be recovered in the $\nu\to0$ limit. More generally, in terms of the original $\upnu$ parameter, the hard edge arises if $\upnu\sim O(1)$ or potentially if its growth with $N$ is slower than linear in the large-$N$ limit.\footnote{\,There are more sophisticated scalings of $\upnu$ one could consider which also allow for a soft edge with no gap. We do not consider these here, since the scaling limits for the most general soft-edge case remain open \cite{Ram_rez_2011}.} Summarizing, the BPS ensembles give rise to two qualitatively distinct scenarios: a gapless spectrum with a hard edge if $\upnu\sim O(1)$, and gapped spectrum with a soft edge if $\upnu\sim O(N)$,

The scaling limit that zooms onto the hard edge in the gapless case again gives rise to the Bessel model in \cref{eq:besself}. The gapped case requires first a shift of the zooming variable to center $x=\Delta$, and an inequivalent scaling in $N$ to capture the statistics of the soft edge.
Remarkably though, the appropriate scaling here turns out to be of the same form as in the WD ensembles, and yields the same Airy model from \cref{eq:airyf} in the large-$N$ limit \cite{Baker_1998},\footnote{\,Large-$N$ corrections here differ from WD ones, but this does not matter for us (cf. \cref{fn:ncorr}) \cite{Perret_2016,Forrester_2018,Forrester_2019}.} as illustrated in \cref{fig:gaps}. As a result, for BPS ensembles with a positive gap, the spectral statistics of the lower edge are in fact not governed by the Bessel models in \cref{eq:besself}, but by the Airy model in \cref{eq:airyf}. The surprising return of the BPS ensemble to Airy statistics is not a coincidence, but explainable by universality results in random matrix theory.

\subsubsection{Edge Universality}
\label{sssec:univ}

Scaling limits were described in previous sections referencing \cref{eq:wdscal,eq:singaz,eq:gaznu}, which are leading spectral densities for our invariant ensembles whose precise form assumes a Gaussian potential.
In fact, universal properties of random matrix statistics imply that upon edge scalings, the results in \cref{eq:airyf,eq:besself} actually arise much more broadly for general matrix ensembles and potentials.\footnote{\,The usual assumptions on the matrix potential $V$ are most merely needed for convergence of the matrix integral, together with mild regularity conditions which do not require $V$ to be analytic or even arbitrarily differentiable \cite{Erd_s_2011}. In particular, $V$ need not include a quadratic term characteristic of Gaussian potentials.}
Indeed, universality results hold beyond just invariant ensembles, and take the form of local spectral statistics which in specific scaling limits are fully determined by the eigenvalue measure and independent of details of the matrix potential \cite{tao2012random}.

Universality applies to bulk statistics \cite{pastur1997universality,erdos2009bulk} and to edge statistics \cite{BOWICK199121,Deift_2006,Erd_s_2011,Bourgade_2014}.
Given our interest in understanding low-energy fluctuations, we are in particular interested in edge universality. Our spectral ensembles give rise to edge statistics that fall into two universality classes: Airy and Bessel \cite{Dumitriu_2006,Forrester_2012,Bertola_2017}, shown in \cref{fig:glue,fig:mpsing}.
We emphasize however that these two classes are not in one to one correspondence with the distinction between supersymmetric and non-supersymmetric ensembles. 
Rather, the universality class is determined by where the relevant edge of the spectrum is soft or hard, which as explained in \cref{sec:gaps} depends on the number of BPS states in the supersymmetric case.

For WD ensembles, edge statistics fall under the universality class of Gaussian WD matrices. The eigenvalues of such matrices famously asymptote to \cref{eq:wdscal}, and exhibit the soft edge in \cref{eq:airedge} by which the spectral density goes to zero as $\sqrt{x}$.
The spectral statistics of ensembles in this universality class are governed by Hermite polynomials, which upon the large-$N$ scaling that focuses on this soft edge give rise to Airy processes like in \cref{fig:glue} \cite{BOWICK199121,Forrester:1993vtx,Tracy:1992kc} (cf. the near-extremal black holes spectrum in \cref{eq:bos} from the $\mathcal{N}=0$ Schwarzian).
For the standard AZ ensembles, the eigenvalues in the Gaussian case asymptote at large $N$ to \cref{eq:singaz}, and exhibit the hard edge in \cref{eq:hardw} by which the spectral density diverges as $1/\sqrt{x}$. The spectral statistics of ensembles in this universality class are governed by Laguerre polynomials, which upon the large-$N$ scaling that focuses on this hard edge give rise to Bessel processes like in \cref{fig:mpsing} \cite{Bronk1965ExponentialEF,nagao91,nagao93,Tracy:1993xj,Forrester:1993vtx} (cf. the near-non-BPS black holes spectrum in \cref{eq:susynon} from the $\mathcal{N}=1$ super-Schwarzian).

By our understanding from \cref{sec:gaps} about spectral gaps in the BPS ensembles, we learn that the Wishart generalization of the AZ ensembles is actually able to accommodate both Bessel and Airy statistics \cite{Dumitriu_2006,Ram_rez_2011,Forrester_2012,Bertola_2017}. The gapless BPS spectrum falls under Bessel universality, with a characteristic $1/\sqrt{x}$ hard edge just like in the standard AZ ensembles (cf. the near-BPS black hole spectrum in \cref{eq:susy2} from the $\mathcal{N}=2$ super-Schwarzian for gapless supermultiplets).
In contrast, the gapped BPS spectrum falls back under Airy universality, recovering the characteristic $\sqrt{x}$ soft edge of WD ensembles as illustrated in \cref{fig:gaps} (cf. the near-BPS black hole spectrum in \cref{eq:susy2,eq:rq4} from the $\mathcal{N}=2,4$ super-Schwarzians for gapped supermultiplets).

\subsubsection{Double Scaling Limits}
\label{ssec:doubscal}

The large-$N$ scalings with a finite $e^{S_0}$ parameter that gave rise to the Airy and Bessel models above are somewhat trivial examples of a more sophisticated type of matrix scalings that go under the name of double scaling limits (DSL) \cite{Gross:1989vs,Brezin:1990rb,Douglas:1989ve,Ginsparg:1993is}. 
Such limits consist of zooming on the edge of the spectrum while simultaneously tuning the coefficients in the matrix potential to certain critical values as $N$ becomes large. The motivation for doing so is the emergence of the structure of topological surfaces which connect matrix integrals to gravity in the continuum $N\to\infty$ limit.

A perturbative large-$N$ expansion of the matrix partition function $\mathcal{Z}$ can be represented by double line matrix diagrams which come with factors of $N^{\chi_g}$, where $\chi_g=2-2g$ is the Euler character of the corresponding graph, and $g$ its genus. Organizing diagrams this way gives an expansion for $\mathcal{Z}$ itself,
\begin{equation}
\label{eq:genusex}
    \mathcal{Z} = \sum_{g=0}^\infty N^{\chi_g} \mathcal{Z}_g,
\end{equation}
where $ \mathcal{Z}_g$ contains all contributions from genus-$g$ diagrams. The single large-$N$ limits considered above thus suppress all terms but $\mathcal{Z}_0$, meaning that only planar diagrams survive. When the couplings $\gamma$ in the matrix potential are allowed to vary, however, contributions from higher genus diagrams can be preserved at large $N$.
As it turns out, the $\mathcal{Z}_g$ functions in the genus expansion all exhibit divergences at the same critical values $\gamma=\gamma_c$ of the couplings, near which they behave as
\begin{equation}
    \mathcal{Z}_g \approx \mathcal{Z}_g^c (\gamma - \gamma_c)^{h \chi_g}, \qquad h >0,
\end{equation}
for some finite constant $\mathcal{Z}_g^c$ independent of $\gamma$. Near such critical points, the expansion in \cref{eq:genusex} can thus be rewritten as
\begin{equation}
\label{eq:dsl}
    \mathcal{Z} \approx \sum_{g=0}^\infty e^{\chi_g S_0} \, \mathcal{Z}_g^c, \qquad e^{S_0} \equiv N (\gamma - \gamma_c)^{h}.
\end{equation}
This expression realizes the possibility of a taking $N\to\infty$ strictly while preserving higher genus contributions by simply demanding that the new effective expansion parameter $e^{S_0}$ stays finite in the limit. As defined, this clearly requires that $\gamma\to\gamma_c$ as $N\to\infty$ in the precise way specified by \cref{eq:dsl}. By coordinating this coupling limit with the appropriate edge zooming variable as $N\to\infty$, the DSL succeeds in capturing the edge statistics exactly in $e^{S_0}$ of large-$N$ matrix models where diagrams of arbitrary topology contribute. Furthermore, there is a precise sense in which the diagrammatic contributions may be interpreted as actual topological surfaces in the DSL. Note that the number of vertices $n$ in a given diagram appears as the power of the coupling $\gamma$, so the expectation value of vertex number at fixed $g$ is computed by
\begin{equation}
    \expval{n}_g \equiv \gamma \frac{\partial}{\partial \gamma} \log \mathcal{Z}_g \approx \frac{\gamma \, h\chi_g}{\gamma-\gamma_c}.
\end{equation}
Since $\expval{n}_g\to\infty$ in the DSL, we see that diagrams may be understood as approaching a continuum limit of surfaces given by graphs with infinitely many vertices.
In this sense, the partition function of a matrix ensemble in the DSL takes the form of a sum over surfaces of arbitrary topology of very much the same flavor as quantum gravity. 
In particular, the tuning of matrix couplings to different critical values in the DSL nicely corresponds to introducing different choices of matter content coupled to gravity.
Of particular interest for us are the DSL matrix ensembles dual to the models of JT quantum gravity describing the physics of near-extremal black holes \cite{Saad:2019lba,Stanford:2019vob,Turiaci:2023jfa,Mertens:2022irh}.

Crucially, it can be shown that for matrix integrals in the DSL, the expansion of the spectral density to all orders in $e^{-S_0}$ is fully determined by its leading from at large $e^{S_0}$, together with a discrete choice out of the ten invariant matrix ensembles.\footnote{\,\label{fn:ordjt}
For example, in ordinary JT gravity one only considers orientable surfaces, which correspond to choosing a WD ensemble with $\upbeta=2$. The other two classical WD ensembles are realized by allowing for non-orientable surfaces as well. The $\upbeta=1,4$ cases differ by whether or not one includes a factor of $(-1)^{n_c}$ in the sum over topologies, where $n_c$ is the number of crosscaps. These are just the purely bosonic cases; theories with fermions, with or without supersymmetry, involve the symmetry structure of the AZ ensembles \cite{Stanford:2019vob,Turiaci:2023jfa}.} This is accomplished through topological recursion relations which can be extracted from the loop equations \cite{Eynard:2004mh,Eynard:2007kz,Stanford:2019vob} (see \cite{Eynard:2014zxa} for an overview on topological recursion). As mentioned in \cref{sec:bhspec}, in JT the leading spectral density at large $e^{S_0}$ is computed on a trivial topology, where the dynamics reduce to a Schwarzian theory governing the fluctuations of the boundary of a hyperbolic disk.
In other words, the near-extremal black hole spectral densities in \cref{eq:bos,eq:susynon,eq:susypro,eq:susy2,eq:susypro2,eq:rq4} coming from Schwarzian theories precisely correspond to the leading spectral densities of the dual DSL matrix ensembles. More explicitly, with our choice of ``normalization'' in \cref{sssec:zoom} (cf. \cref{eq:rhoslim}),
\begin{equation}
\label{eq:noscal}
    \hat{\rho}_*^{\mathcal{N}}(E) \equiv e^{-S_0} \, \rho^{\mathcal{N}}(E).
\end{equation}
with the quotation marks emphasizing that actually $\hat{\rho}_*$ is not even normalizable, and all we are doing is stripping off the extensive factor of $e^{S_0}$ from $\rho^{\mathcal{N}}$.
With this convention, one easily verifies that near the lower edge of the continuous parts of their spectra, $\hat{\rho}_*^{\mathcal{N}}$ exhibits exactly the soft edge in \cref{eq:airedge} for $\mathcal{N}=0$ and gapped $\mathcal{N}=2,4$, and the hard edge in \cref{eq:hardw} for $\mathcal{N}=1$ and gapless $\mathcal{N}=2$. This is how edge universality manifests itself in the DSL, i.e., in the near-edge behavior of the DSL spectral density, which itself is a near-edge spectral density. In particular, we expect the exact DSL spectral densities to be governed by Airy statistics near soft edges, and by Bessel statistics near hard edges.

\section{Near-Extremal Saddles}
\label{sec:canemsem}

The goal of this section is to investigate the replica trick that quenches the thermal entropy in a way that elucidates its workings in quantum gravity.
To do so we pursue the calculation in the near-extremal regime employing only those tools from random matrix theory with a clear gravitational dual description.
Our results show that a treatment strictly to leading order suffices, the crucial realization being that a new saddle point arises at large $\beta\gtrsim O(e^{S_0})$ and becomes dominant.
For the replica trick in \cref{eq:reptrick}, the behavior of this saddle point is shown to be determined by spectral edge statistics, implying that to leading order the quenched entropy behaves universally and is completely independent of non-perturbative details.
On the gravity side, this saddle point is seen to be realized by a particular type of branes associated to energy eigenvalues.
Because these branes are needed at leading order, our findings imply that their inclusion on the gravity side is non-optional for a genuine duality between JT gravity and matrix integrals.\footnote{\,We thank Netta Engelhardt for emphasizing this point of view.}

To implement \cref{eq:reptrick} we need to study the moments of $Z(\beta)$ in a matrix ensemble,
\begin{equation}
\label{eq:Zmom}
    \expval{ Z(\beta)^m} = \int\displaylimits_{\mathbb{R}^N} d\Lambda \, p(\Lambda) \left( \Tr e^{-\beta \Lambda }\right)^m,
\end{equation}
where $p$ is given in \cref{eq:pIZ} and $\Lambda$ denotes the diagonal matrix of $N$ eigenvalues of the random Hamiltonian.
As in \cref{eq:Iint}, it is again convenient to express this as an action integral,
\begin{equation}
\label{eq:momentint}
    \expval{Z(\beta)^m} = \frac{1}{\mathcal{Z}} \int\displaylimits_{\mathbb{R}^N} d\Lambda \, e^{-N I_m(\Lambda)}, \qquad I_m(\Lambda) \equiv I(\Lambda) + \frac{m}{N} \log \Tr e^{-\beta \Lambda },
\end{equation}
with the matrix action $I$ given by \cref{eq:Iout}.
This integral admits a saddle-point approximation at large $N$. 
One could follow the same strategy as in \cref{ssec:contl} to study this integral in the continuum limit to leading order. 
If done na\"ively, however, the resulting functional integral and variational problem would in fact be ill-posed for any $m>0$. This is because the $Z(\beta)$ continuum one would get by replacing a discrete measure $\hat{\rho}_\Lambda$ by a continuous one $\hat{\rho}$ (cf. \cref{eq:discon,eq:gencon}),
\begin{equation}
\label{eq:naive}
    \int\displaylimits_{\mathbb{R}} dx \, \hat{\rho}_\Lambda(x) \, e^{-\beta x} \quad\longrightarrow\quad \int\displaylimits_{\mathbb{R}} dx \, \hat{\rho}(x) \, e^{-\beta x},
\end{equation}
would generally cause the action to be unbounded from below. This fundamental issue would be caused by shallow tails of the continuum measure $\hat{\rho}$ extending to arbitrarily negative values.\footnote{\,\label{fn:unbound}This issue can be easily exemplified. The na\"ive continuum limit for e.g. the GUE yields (cf. \cref{eq:contIm})
$$
    \widehat{I}_m[\hat{\rho}] \reprel{!}{=} \frac{1}{2} \int\displaylimits_\mathbb{R} dx \, \hat{\rho}(x) \, x^2 - \int\displaylimits_\mathbb{R} dx \, dy \, \hat{\rho}(x) \hat{\rho}(y) \log|x-y| - \frac{m}{N^2} \log \left( N  \int\displaylimits_\mathbb{R} dx \, \hat{\rho}(x) \, e^{-\beta x} \right).
$$
Consider letting $\hat{\rho}(x)$ decay like $e^{-\frac{1}{\tau} \left( \frac{|x|}{a} \right)^\tau}$ as $x\to-\infty$ for $1<\tau<2$. The lower bound guarantees that the integral in the $\beta$-dependent term converges for any $\beta$, while the upper bound implies that this term will grow negative faster in $a$ than the potential term grows positive. Then one finds that $\widehat{I}_m$ behaves as $-\frac{\tau-1}{\tau} m (a\beta)^{\frac{\tau}{\tau-1}}$ at large $a$, which implies the action is unbounded from below since $\widehat{I}_m$ can be made arbitrarily negative by taking $a\to\infty$.} But such tails would clearly not make sense from the discretum: there must be a lowest eigenvalue which, in a saddle-point approximation, would lie at some finite value. 
The effect of $Z(\beta)$ in \cref{eq:momentint} thus has to be better understood before going to the large-$N$ continuum.

Valuable intuition can be gained from a finite-$N$ evaluation of the object whose large-$N$ limit gives the saddle point we are after. Namely, thinking of \cref{eq:momentint} altogether as just another matrix integral, we may define its joint eigenvalue PDF by (cf. \cref{eq:Iint,eq:pIZ})
\begin{equation}
    p_\beta^m(\Lambda) = \frac{e^{-N I_m(\Lambda)}}{\mathcal{Z}\expval{Z(\beta)^m}},
\end{equation}
and correspondingly by \cref{eq:pnPN,eq:densn1} obtain an average spectral density via
\begin{equation}
\label{eq:specbetam}
    \expval{\rho^m_{\beta}(\lambda)} \equiv N \int\displaylimits_{\mathbb{R}^{N-1}} d\lambda_2\cdots d\lambda_N \, p_\beta^m(\lambda,\lambda_2,\dots,\lambda_N).
\end{equation}
This is precisely the object we need to understand in the large-$N$ limit. In particular, by \cref{eq:rhoslim},
\begin{equation}
\label{eq:betames}
    {\hat{\rho}^m_{\beta}}{}_*(x) \equiv \lim_{N\to\infty} \frac{\expval{\rho^m_{\beta}(x)}}{N},
\end{equation}
is the equilibrium measure that a saddle-point evaluation of \cref{eq:momentint} must give if the continuum limit is taken correctly. As an example, we integrate \cref{eq:specbetam} exactly for the GUE at finite $N$ and show the result of varying $m$ and $\beta$ in \cref{fig:instanton}.
The observed behavior of $\expval{\rho^m_{\beta}(x)}$ suggests that in a large-$N$ limit where $\beta\sim O(N)$ the lowest eigenvalue naturally behaves in a discrete manner, explaining why \cref{eq:momentint} does not admit a na\"ive continuum limit. With the intuition that \cref{fig:instanton} provides, we now derive a general expression for \cref{eq:momentint} that actually admits such a limit.

\begin{figure}
    \centering
    \includegraphics[width=0.8\textwidth]{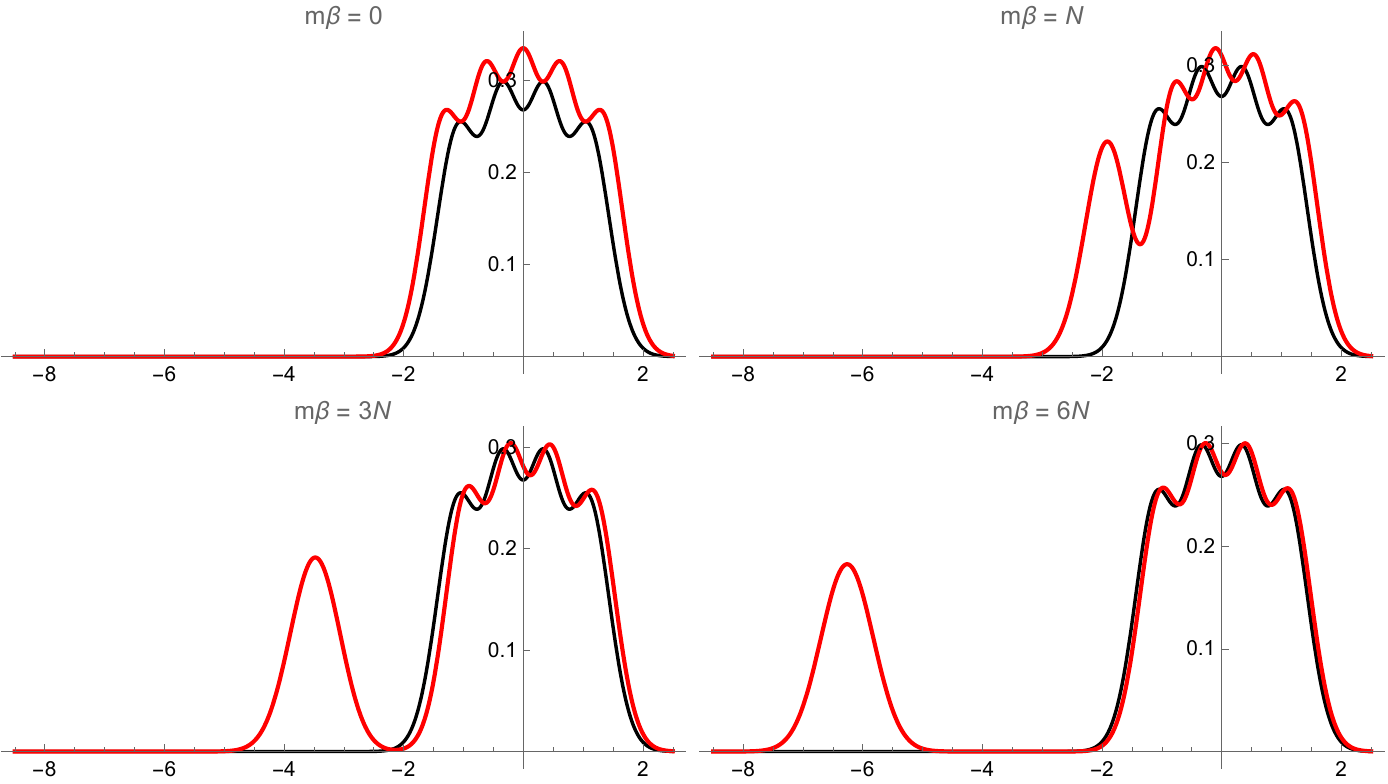}
    \caption{Illustration of $\langle{\rho^m_{\beta}(x)}\rangle$ by exact evaluation of the integral in \cref{eq:specbetam} for the GUE at finite $N$. The red curve corresponds to $N=5$ and is plotted for the values of $m\beta$ shown; here $m=2$, but the relevant qualitative features are seen to only depend on the ratio $m\beta/N$. The black curve corresponds to $N=4$ with $m=0$ fixed, and serves as a reference for the behavior of the red curve. As $m\beta/N$ increases, the peak associated to the lowest eigenvalue of the $N=5$ curve separates from the other $N-1$ eigenvalues, which converge down to the $N=4$ curve. In other words, only one eigenvalue depends significantly on $\beta$, while the others are mostly insensitive to the $Z(\beta)$ insertion. The single dynamical eigenvalue can be analytically shown to asymptote with $m\beta/N$ to a Gaussian of variance $1/N$ around $-m\beta/N$. Importantly, this peak shrinks as $N\to\infty$, giving no overlap with the bulk of eigenvalues at large $N$ for any constant $m\beta/N>0$.}
    \label{fig:instanton}
\end{figure}

\subsection{Low Temperatures}

The form of $\log Z(\beta)$ at small and large $\beta$ can be read off from \cref{eq:hight,eq:lowt}, respectively. At small $\beta$, this goes as $\log Z(\beta) \sim \log N$, whereas at large $\beta$ it behaves as $\log Z(\beta) \sim -\beta\min\Lambda$. Hence at large $N$, if $\beta$ is kept fixed, the contribution from the $\beta$-dependent term in \cref{eq:momentint} gets suppressed by an extra factor of $N$ relative to $I(\Lambda)$. In a large-$N$ treatment of this expectation value integral, the saddle points would thus be solely determined by the original matrix integral and completely unaffected by the insertion of $Z(\beta)^m$.

The more interesting regime arises at large $\beta$, when $\log Z(\beta) \sim -\beta\min\Lambda$. In this case, by making $\beta\sim O(N)$, the two action terms in \cref{eq:momentint} become of the same order in the large-$N$ limit. This must lead to a saddle point for the expectation value integral strictly different from that of the original matrix integral due to the effect of the $Z(\beta)^m$ insertion. The new saddle point clearly appears for any $m>0$ and gives a distinct role to the lowest eigenvalue $\min\Lambda$ in the integral, as observed also in \cref{fig:instanton}. 
In particular, we expect this lowest eigenvalue to be particularly sensitive to $\beta$, and thus capture the dependence of $\expval{Z(\beta)^m}$ on it in the large-$\beta$ regime we are interested in.
Since $\beta\sim O(N)$ translates into $\beta\sim O(e^{S_0})$ in the DSL, note that this is precisely the near-extremal regime where the gravitational entropy generically begins to turn negative, as discussed in \cref{ssec:gsint}. In other words, the appearance of this new saddle point is perfectly consistent with the expectation that annealed and quenched entropies must depart from each other near extremality.

To track the behavior of the lowest eigenvalue, consider splitting $\Lambda$ into\footnote{\,The integral in \cref{eq:momentint} treats all eigenvalues symmetrically, so the $\lambda_k$ labelling is unrelated to any ordering.}
\begin{equation}
\label{eq:lambdasplit}
    \hat{\lambda} \equiv \min \Lambda, \qquad \widehat{\Lambda} \equiv \diag\{\lambda_k \neq \hat{\lambda} \}_{k=1}^N.
\end{equation}
Separating eigenvalues this way, the canonical partition function may equivalently be written
\begin{equation}
\label{eq:zbetasplit}
    \log Z(\beta) = - \beta\hat{\lambda} + \log \left( 1 + \sum_{\lambda\in\widehat{\Lambda}} e^{-\beta (\lambda-\hat{\lambda})} \right),
\end{equation}
an expression which is natural for large $\beta$ but holds for any $\beta$. 
Any global eigenvalue degeneracies coming from the symmetries of the ensemble can be factored out from this treatment.
As for dynamical degeneracies, recall that the Vandermonde determinant exerts a repulsion across all eigenvalues which forbids them. In particular, this means that $\lambda>\hat{\lambda}$ can be assumed strictly for all $\lambda \in \widehat{\Lambda}$. 
In general though, every eigenvalue $\lambda \in \widehat{\Lambda}$ can get arbitrarily close to $\hat{\lambda}$, so it is not \`a priori clear whether the sum in \cref{eq:zbetasplit} can be assumed to be exponentially suppressed in $\beta$. The key observation is that the typical separation between eigenvalues that the matrix integral induces is no smaller than $O(N^{-1})$, and often larger than that near spectral edges (cf. \cref{sssec:zoom}). Since we will be interested in scaling $\beta$ linearly in $N$, this statistical separation between eigenvalues will always be counteracted by $\beta$ in the exponent. Hence we expect the logarithmic term in \cref{eq:zbetasplit} to statistically be at most $O(\log N)$, and thus subleading relative to the $\beta\sim O(N)$ term. 
As a result, the large-$\beta$ approximation
\begin{equation}
\label{eq:appzbeta}
    \log Z(\beta) \approx - \beta\hat{\lambda},
\end{equation}
holds generically across the ensemble.
In fact, notice that this $-\beta\hat{\lambda}$ term will always be favoring strictly smaller values for $\hat{\lambda}$ in the matrix integral than for any $\lambda \in \widehat{\Lambda}$. Therefore we actually expect the typical difference $\lambda-\hat{\lambda} \gtrsim O(N^{-1})$, thereby enhancing the approximation in \cref{eq:appzbeta} to be accurate up to at most $O(1)$ corrections, as our results turn out to confirm (cf. \cref{fig:instanton}).
The strategy in what follows is to use the large-$\beta$ behavior of $Z(\beta)$ in \cref{eq:appzbeta} to break down $\expval{ Z(\beta)^m}$ into matrix integrals which are computable using a saddle-point approximation.

\subsection{Eigenvalue Splitting}
\label{ssec:split}

Given the special status the lowest eigenvalues acquires, it will be useful to identify it inside the matrix integral. As expressed throughout, spectral integrals run over all of $\mathbb{R}$ for every eigenvalue in $\Lambda$, so any one of them can be the smallest in some integration domain. However, all integrals are fully symmetric under permutations of the eigenvalues. This allows one to reorder them so as to always keep track of which eigenvalue is the lowest. Using the notation in \cref{eq:lambdasplit}, one easily arrives at the following identity
\begin{equation}
\label{eq:intsep1}
    \int\displaylimits_{\mathbb{R}^N} d\Lambda \, (\,\cdot\,) = N \int\displaylimits_{\mathbb{R}^{N-1}} d\widehat{\Lambda} \int\displaylimits_{-\infty}^{\min\widehat{\Lambda}} d\hat{\lambda}  \, (\,\cdot\,).
\end{equation}
This is useful because under this reorganization of the eigenvalue integral we can identify $\hat{\lambda}$ as a specific integration variable corresponding to $\min\Lambda$, and thus study what happens to the lowest eigenvalue inside the full matrix integral.
The representation in \cref{eq:intsep1} can be read as upper bounding the $\hat{\lambda}$ eigenvalue by all other $\widehat{\Lambda}$ eigenvalues. However, one could also think of all eigenvalues $\widehat{\Lambda}$ as being upper bounded by $\hat{\lambda}$, which gives the alternative representation
\begin{equation}
\label{eq:intsep2}
    \int\displaylimits_{\mathbb{R}^N} d\Lambda \, (\,\cdot\,) = N \int\displaylimits_{\mathbb{R}} d\hat{\lambda} \int\displaylimits_{
        \mathbb{R}_{\geq \hat{\lambda}}^{N-1}
        } d\widehat{\Lambda} \, (\,\cdot\,).
\end{equation}
Both \cref{eq:intsep1,eq:intsep2} become useful and interchangeable at the saddle-point level if one trades between an upper bound constraint for $\hat{\lambda}$ and a lower bound constraint for $\widehat{\Lambda}$.
In particular, \cref{eq:intsep2} allows one to define an induced PDF on the lowest eigenvalue. Namely, applying \cref{eq:intsep2} to the joint PDF of eigenvalues $p(\Lambda) = p(\hat{\lambda},\widehat{\Lambda})$ gives the PDF for the lowest eigenvalue,\footnote{\,For the WD ensembles, with the right large-$N$ edge scaling this would give the Tracy-Widom distribution \cite{Tracy:1992kc}.}
\begin{equation}
\label{eq:twdis}
    \mathcal{F}(\hat{\lambda}) \equiv N \int\displaylimits_{
    \mathbb{R}_{\geq \hat{\lambda}}^{N-1}} d\widehat{\Lambda} \, p(\hat{\lambda},\widehat{\Lambda}).
\end{equation}
Applying this splitting to \cref{eq:Zmom} gives the PDF induced on $\hat{\lambda}$ under the insertion of $Z(\beta)^m$,
\begin{equation}
\label{eq:fbmn}
    \mathcal{F}_\beta^m(\hat{\lambda}) \equiv \frac{N}{\expval{Z(\beta)^m}} \int\displaylimits_{
    \mathbb{R}_{\geq \hat{\lambda}}^{N-1}} d\widehat{\Lambda} \, p(\hat{\lambda},\widehat{\Lambda}) \left( e^{-\beta\hat{\lambda}} + \Tr e^{-\beta \widehat{\Lambda}} \right)^m.
\end{equation}
An illustration of this PDF for finite $N$ is shown in \cref{fig:lowest}.
In consistency with \cref{eq:appzbeta}, we observe that $\mathcal{F}(\hat{\lambda}) \, e^{-m\beta\hat{\lambda}}$ approximates $\mathcal{F}_\beta^m(\hat{\lambda})$ very well as $\beta\to\infty$. However, making such a large-$\beta$ approximation from the onset would be too drastic and not very illuminating. Instead, our strategy will be to take a continuum limit, use saddle-point methods to simplify the integral, and only make large $N$ and $\beta$ at the controlled level of the equations of motion.\footnote{\,Even if we wanted to work with \cref{eq:appzbeta,eq:twdis}, a saddle-point analysis is unavoidable for all practical purposes. The standard definition of $\mathcal{F}$ in terms of Fredholm determinants of the self-reproducing kernel in \cref{eq:cdkernel} becomes an infinite dimensional operator at large $N$ which is extremely hard to construct, even numerically \cite{Johnson:2021zuo}.}

\begin{figure}
    \centering
    \includegraphics[width=0.9\textwidth]{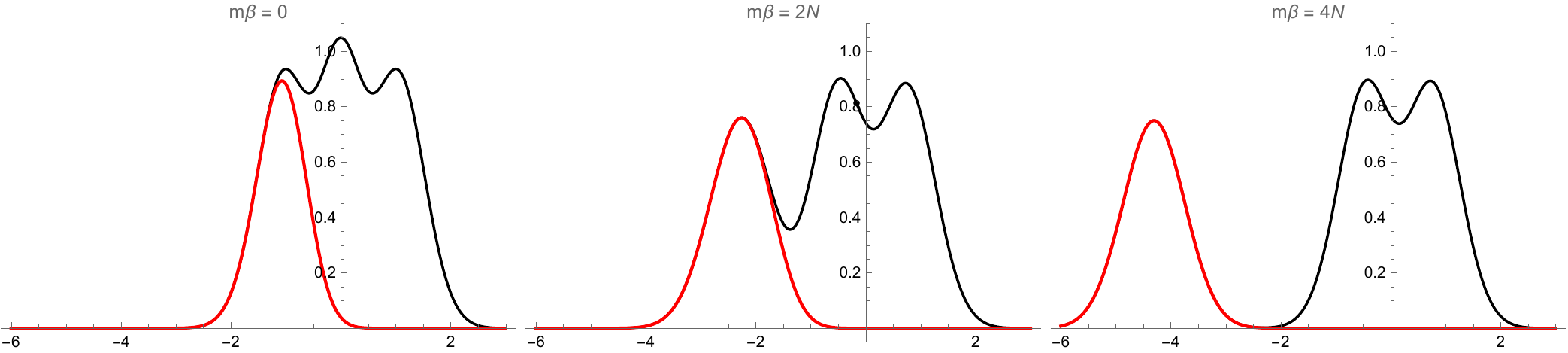}
    \caption{Illustration of $\mathcal{F}_\beta^m$ (red) against $\langle{\rho^m_{\beta}}\rangle$ (black) by direct integration of \cref{eq:specbetam,eq:fbmn}, respectively, for the GUE at $N=3$ (cf. \cref{fig:instanton}). Here $\beta$ is being varied at fixed $m=6$, but the qualitative features are seen to only be sensitive to $m\beta/N$. As desired, the eigenvalue splitting from \cref{eq:intsep2} allows for $\mathcal{F}_\beta^m$ to isolate the PDF of just the statistically lowest eigenvalue within the total spectral density $\langle{\rho^m_{\beta}}\rangle$. At large $\beta$, $\mathcal{F}_\beta^m(\hat{\lambda})$ tends to simply $\mathcal{F}(\hat{\lambda}) \, e^{-m\beta\hat{\lambda}}$, while the other eigenvalues converge to the $N=2$ spectrum.}
    \label{fig:lowest}
\end{figure}

With this strategy in mind, consider breaking up the integral in \cref{eq:momentint} in terms of the decomposition of $\Lambda$ into $\hat{\lambda}$ and $\widehat{\Lambda}$. Firstly, writing out $Z(\beta)$, the total action is
\begin{equation}
\label{eq:imin}
    I_m(\Lambda) = I(\Lambda) - \frac{m}{N} \log( e^{-\beta\hat{\lambda}} + \Tr e^{-\beta \widehat{\Lambda}} ).
\end{equation}
The original action $I$ splits into
\begin{equation}
\label{eq:actiondec}
    I(\Lambda) = I(\widehat{\Lambda}) + V_\upnu(\hat{\lambda}) - \frac{\upbeta}{N} \log | {\det} ( \widehat{\Lambda} - \hat{\lambda} ) |,
\end{equation}
where the interaction is the Vandermonde determinant. Recalling \cref{eq:fpotwi}, we recognize the $\hat{\lambda}$-dependent terms in \cref{eq:actiondec} as precisely giving the effective potential for the $\hat{\lambda}$ eigenvalue,
\begin{equation}
    \widehat{V}(\hat{\lambda}) =  V_\upnu(\hat{\lambda}) - \frac{\upbeta}{N} \log | {\det} ( \widehat{\Lambda} - \hat{\lambda} ) |.
\end{equation}
Putting this altogether in \cref{eq:momentint} and using the integral representation from \cref{eq:intsep2} gives
\begin{equation}
\label{eq:beast}
    \expval{Z(\beta)^m} = \frac{N}{\mathcal{Z}} 
    \int\displaylimits_{\mathbb{R}} d\hat{\lambda} 
    \,
    e^{-N V_\upnu(\hat{\lambda})}
    \int\displaylimits_{
    \mathbb{R}_{\geq \hat{\lambda}}^{N-1}
    } d\widehat{\Lambda} 
    \, e^{- N I(\widehat{\Lambda})}
    \, 
    | {\det} ( \widehat{\Lambda} - \hat{\lambda} ) |^\upbeta
    \left( e^{-\beta\hat{\lambda}} + \Tr e^{-\beta \widehat{\Lambda}} \right)^m.
\end{equation}
This matrix integral is now in a form that is amenable to a continuum treatment. Crucially, in contrast with \cref{eq:momentint}, the distinct behavior of the $\hat{\lambda}$ eigenvalue has now been isolated such that a continuum limit can be taken for the $\widehat{\Lambda}$ integral alone (cf. \cref{fn:unbound}).
The reader may recognize the determinant in \cref{eq:beast} as the first hint that branes will be important in the gravitational avatar of this calculation.

\subsection{Single-Eigenvalue Instanton}
\label{ssec:insti}

As discussed so far, inserting $Z(\beta)^m$ inside the matrix integral and making $\beta\sim O(N)$ causes the lowest eigenvalue $\hat{\lambda}$ to behave differently from the rest. This single eigenvalue behaves in such a way that it dynamically separates from the others at large $\beta$, a phenomenon which we already illustrated for the GUE in \cref{fig:instanton}.
In more general matrix integrals this may be recognized as the characteristic behavior of a single-eigenvalue instanton (see e.g. \cite{Okuyama:2018gfr}). 
The purpose of this section is to work out the large-$N$ instanton that characterizes $\expval{ Z(\beta)^m}$ at large $\beta$.\footnote{We thank Luca Iliesiu for insightful conversations which inspired the study of this low-temperature instanton.}

To tackle the integral in \cref{eq:beast} by saddle-point approximation, it becomes relevant to take a continuum large-$N$ limit first. Although at large $N$ the distinction between $N$ and $N-1$ becomes immaterial to leading order, it will be useful to distinguish such terms in order to more reliably connect results back to finite $N$ later on. Following \cref{ssec:contl}, in the continuum the integral over $N-1$ eigenvalues $\widehat{\Lambda}$ bounded from below by $\hat{\lambda}$ becomes a functional integral over spectral densities $(N-1) \int \mathcal{D}\hat{\sigma}$ where $\hat{\sigma}$ is unit-normalized and restricted to have $\supp\hat{\sigma} \subseteq [\hat{\lambda},\infty)$. Correspondingly, the action $I(\widehat{\Lambda})$ becomes a functional $(N-1) \widehat{I}[\hat{\sigma}]$ where the summations over $N-1$ eigenvalues are replaced by integrals $(N-1) \int_{\mathbb{R}} dx \, \hat{\sigma}(x)$. The upshot is\footnote{\,This is an $(N-1)$-eigenvalue action obtained as part of an originally $N$-eigenvalue integral, hence the ratio in front of the term that is nonlinear in $\hat{\sigma}$ relative to the expression for $\widehat{I}[\hat{\rho}]$ in \cref{eq:contIm}.}
\begin{equation}
\label{eq:contIm1}
    \widehat{I}[\hat{\sigma}] = \int\displaylimits_{\mathbb{R}} dx \, \hat{\sigma}(x) \, V_\upnu(x) - \frac{N-1}{N} \frac{\upbeta}{2} \int\displaylimits_{\mathbb{R}} dx \, dy \, \hat{\sigma}(x) \hat{\sigma}(y) \log|x-y|,
\end{equation}
where we could have equivalently written the integrals over just $[\hat{\lambda},\infty)$ given the constrained support of $\hat{\sigma}$. 
By \cref{eq:effV}, functionally varying $\widehat{I}$ in $\hat{\sigma}$ gives the effective potential
\begin{equation}
\label{eq:effVv}
    \widehat{V}[\hat{\sigma};x] = V_\upnu(x) - \frac{N-1}{N} \upbeta \int\displaylimits_{\mathbb{R}} dy \, \hat{\sigma}(y) \log|x-y|.
\end{equation}
As for the determinant in \cref{eq:beast}, using the identity $\log\det M = \Tr\log M$, in the continuum
\begin{equation}
\label{eq:contdet}
    | {\det} ( \widehat{\Lambda} - \hat{\lambda} ) |^\upbeta \quad\longrightarrow\quad \Psi[\hat{\sigma};\hat{\lambda}] \equiv \exp[ (N-1) \, \upbeta \int\displaylimits_{\mathbb{R}} dx \, \hat{\sigma}(x) \log | x - \hat{\lambda} |].
\end{equation}
which can also be written in terms of the effective potential in \cref{eq:effVv} as
\begin{equation}
\label{eq:deteff}
    \Psi[\hat{\sigma};x] \, e^{- N V_\upnu(x)}  = e^{- N \widehat{V}[\hat{\sigma};x]}.
\end{equation}
The choice of notation here is suggestive of the brane operators of \cite{Saad:2019lba}, which correspond to the full left-hand side of \cref{eq:deteff}.
Finally, in the continuum $Z(\beta)$ as written in \cref{eq:beast} becomes
\begin{equation}
\label{eq:Zcont}
    Z(\beta) \quad\longrightarrow\quad Z[\hat{\sigma};\hat{\lambda}] \equiv e^{-\beta\hat{\lambda}} + (N-1) \int\displaylimits_{\mathbb{R}} dx \, \hat{\sigma}(x) \, e^{-\beta x},
\end{equation}
where the $\beta$-dependence is left implicit. This continuum limit of $Z$ should be contrasted with the na\"ive one in \cref{eq:naive}.
The complete continuum form of \cref{eq:beast} thus reads
\begin{equation}
\label{eq:contzbm}
    \expval{ Z(\beta)^m} = \frac{N(N-1)}{\mathcal{Z}} 
    \int\displaylimits_{\mathbb{R}} d\hat{\lambda} 
    \,
    e^{-N V_\upnu(\hat{\lambda})}
    \int\displaylimits \mathcal{D}\hat{\sigma} 
    \, e^{- N(N-1) \widehat{I}[\hat{\sigma}]}
    \, 
    \Psi[\hat{\sigma};\hat{\lambda}]
    \,
    Z[\hat{\sigma};\hat{\lambda}]^m.
\end{equation}
In the large-$N$ limit, note how the action term above scales as $N^2$, whereas the $\Psi$ insertion only goes as $N$, as can be seen from \cref{eq:contdet}. This means that, to leading order, the saddle point $\hat{\sigma}=\hat{\sigma}_*$ that dominates the functional integral in \cref{eq:contzbm} is insensitive to $\Psi$. Intuitively, the determinant insertion of a single eigenvalue that $\Psi$ encodes has a subleading effect on the dynamics of the $O(N)$ eigenvalues that participate in $\widehat{I}$; the effect of the latter on the former will however not be negligible. 
Additionally, from the discussion below \cref{eq:zbetasplit} we know that $Z$ is dominated at large $\beta$ by the lowest eigenvalue $\hat{\lambda}$ alone with $\hat{\sigma}$ only contributing exponentially suppressed corrections, so $\hat{\sigma}_*$ will also be insensitive to $\beta$ at leading order.
Altogether, the extremization problem that defines $\hat{\sigma}_*$ does not depend on $\hat{\lambda}$ at all to leading order at large $N$ and $\beta$, in perfect consistency with the examples in \cref{fig:instanton,fig:lowest}.

Extremizing \cref{eq:contzbm} thus amounts to applying the equilibrium condition, \cref{eq:deffV0}, to the effective potential in \cref{eq:effVv}. As a result, the saddle point $\hat{\sigma}_*$ corresponds to the unit-normalized measure that solves the singular integral equation
\begin{equation}
\label{eq:singint}
    \fint\displaylimits_{\mathbb{R}} dy \, \frac{\hat{\sigma}_*(y)}{x-y} = \frac{N}{N-1} \frac{V_\upnu'(x)}{\upbeta}, \qquad \forall x\in\supp\hat{\sigma}_*.
\end{equation}
Just like \cref{eq:expcond}, this equation can be solved exactly for $\hat{\sigma}_*$, with the general solution taking a form analogous to \cref{eq:rho0main}.
When approximating the $\hat{\sigma}$ integral in \cref{eq:contzbm} by its saddle point $\hat{\sigma}_*$, consistency with the unit-normalization of the expectation value at $m=0$ requires to also evaluate $\mathcal{Z}$ on the same saddle. The saddle-point approximation $\mathcal{Z}_*$ for $\mathcal{Z}$ reads
\begin{equation}
\label{eq:Zstar}
    \mathcal{Z}_* \equiv N(N-1) e^{- N(N-1) \widehat{I}[\hat{\sigma}_*]} \, \widehat{\mathcal{Z}}_*, \qquad \widehat{\mathcal{Z}}_* \equiv
    \int\displaylimits_{-\infty}^{\hat{\lambda}_0} d\hat{\lambda} \, \Psi[\hat{\sigma}_*;\hat{\lambda}] \, e^{- N V_\upnu(\hat{\lambda})} = \int\displaylimits_{-\infty}^{\hat{\lambda}_0} d\hat{\lambda} \, e^{- N \widehat{V}_*(\hat{\lambda})},
\end{equation}
where we have used \cref{eq:deteff}, the notation $V_*$ refers to the leading effective potential with $\hat{\sigma} = \hat{\sigma}_*$ (cf. \cref{eq:leadeff}), and we have defined the lower edge of the leading spectrum
\begin{equation}
\label{eq:groundinf}
    \hat{\lambda}_0 \equiv \inf\supp\hat{\sigma}_*.
\end{equation}
The new constant $\widehat{\mathcal{Z}}_*$ is convenient because after evaluating the saddle-point approximation for \cref{eq:contzbm}, all other constants cancel out and one is left with
\begin{equation}
\label{eq:semilarge}
    \expval{ Z(\beta)^m}_*
    =
    \frac{1}{\widehat{\mathcal{Z}}_*} 
    \int\displaylimits_{-\infty}^{\hat{\lambda}_0} d\hat{\lambda} 
    \,
    e^{- N \widehat{V}_*(\hat{\lambda})}
    \,
    Z[\hat{\sigma}_*;\hat{\lambda}]^m.
\end{equation}
Note that $\expval{ Z(\beta)^m}_*$ is indeed consistently normalized to unity for $m=0$. The integrand above provides a large-$N$ approximation to the lowest eigenvalue PDF we encountered in \cref{eq:fbmn},
\begin{equation}
\label{eq:saddledist}
    \mathcal{F}_\beta^m{}_*(\hat{\lambda}) = \frac{1}{\widehat{\mathcal{Z}}_* \expval{ Z(\beta)^m}_*} \, e^{- N \widehat{V}_*(\hat{\lambda})} Z[\hat{\sigma}_*;\hat{\lambda}]^m, \qquad \hat{\lambda}\leq\hat{\lambda}_0.
\end{equation}
This PDF is illustrated for the GUE in \cref{fig:halfhalf}.
Physically, \cref{eq:saddledist} captures the statistics of the ground-state energy $\hat{\lambda}$ in the presence of a classical thermal gas of $O(N)$ higher energies that localize in the $\hat{\sigma}_*$ saddle. 
Importantly, $\hat{\lambda}$ must by construction lie below the lower edge of $\hat{\sigma}_*$, sometimes referred to as the classically forbidden region. However, we emphasize that there is nothing non-classical or forbidden about this instanton: it is a leading saddle (i.e., classical), and its location is perfectly allowed due to the large-$\beta$ back-reaction that $Z(\beta)$ sources.

\begin{figure}
    \centering
    \includegraphics[width=0.9\textwidth]{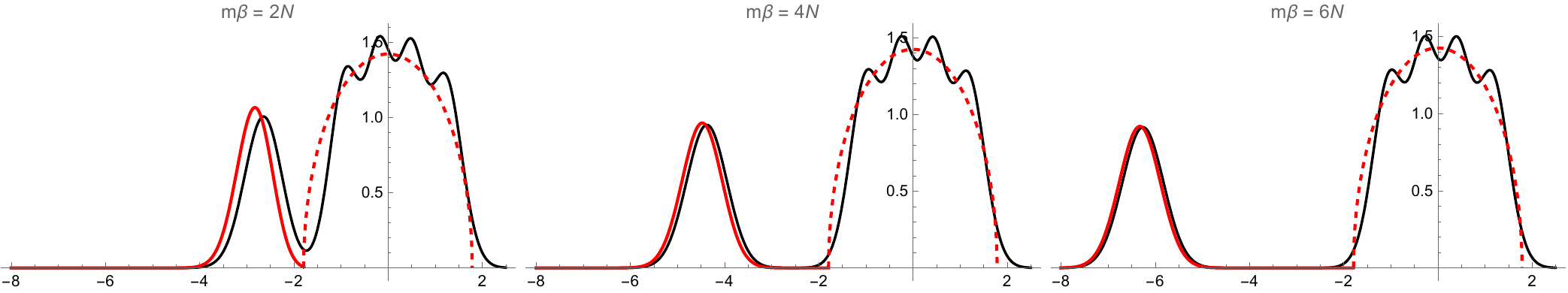}
    \caption{
    Illustration of ${\mathcal{F}_\beta^m}{}_*$ (red, solid) given by \cref{eq:saddledist} evaluated on the continuous $\hat{\sigma}_*$ saddle that solves \cref{eq:singint} for the GUE at $N=5$ (dashed). Together, ${\mathcal{F}_\beta^m}{}_*(x) + (N-1) \hat{\sigma}_*(x)$ is an approximation to $\langle{\rho_\beta^m}(x)\rangle$ in which the lowest eigenvalue is only sensitive to the saddle form of the higher eigenvalues. 
    The black curve corresponds to the exact result for $\langle{\rho_\beta^m}\rangle$ using \cref{eq:fbmn} (cf. \cref{fig:instanton}), and which also captures ${\mathcal{F}_\beta^m}(x)$ from \cref{eq:fbmn} exactly (cf. \cref{fig:lowest}). Here $m=1$, but once again, qualitative features mostly depend on $m\beta/N$. Already at such small $N=5$, we see that ${\mathcal{F}_\beta^m}{}_*$ rapidly converges to ${\mathcal{F}_\beta^m}$ at large $\beta$, even though $\hat{\sigma}_*$ is a very coarse approximation to the $N=4$ spectral density that the higher eigenvalues asymptote to.
    }
    \label{fig:halfhalf}
\end{figure}

If the result for the integrand in \cref{eq:semilarge} is sufficiently simple, one could now attempt to directly evaluate the integral over $\hat{\lambda}$ to obtain a final answer for $\expval{ Z(\beta)^m}_*$. However, since \cref{eq:semilarge} itself already is the result of a large-$N$ saddle-point approximation, it is in fact more natural to also approximate the $\hat{\lambda}$ integral by its saddle-point value. Indeed, this is a very reasonable approximation at large $N$, since the associated PDF for the lowest eigenvalue in \cref{eq:saddledist} clearly has a characteristic width which shrinks as $1/N$ at large $N$ (cf. \cref{fig:instanton}). By reducing this PDF to the dynamics of where it peaks, we will be treating the lowest eigenvalue as a point particle and understanding its motion as a function of $\beta$. Recalling that we are in a regime in which $\beta$ and $N$ are simultaneously large, the extremization problem we are after will be sensitive to the presence of the $Z$ factor. Fortunately, from \cref{eq:zbetasplit} we see that $Z$ at large $\beta$ will only depend on $\hat{\lambda}$ up to exponentially suppressed corrections from higher eigenvalues. Hence using the approximation in \cref{eq:appzbeta} and extremizing, the saddle point $\hat{\lambda}=\hat{\lambda}_*$ we are after is the solution to\footnote{\,\label{fn:compsad}We are assuming $\hat{\lambda}_*\in\mathbb{R}$. It is in principle possible for this integral to be dominated by stationary points along a Lefschetz thimble off $\mathbb{R}$. In this case the integral representation in \cref{eq:intsep1} is preferable. The same stationary conditions would apply, but with $\hat{\lambda}_*$ allowed to be complex and lying on the appropriate contour of steepest descent.}
\begin{equation}
    \evalat{\frac{d}{d\hat{\lambda}}}{\hat{\lambda} = \hat{\lambda}_*} e^{- N \widehat{V}_*(\hat{\lambda}) - m \beta \hat{\lambda}} = 0, \qquad \hat{\lambda}_* \leq \hat{\lambda}_0,
\end{equation}
This reduces to a simple equation for the location of $\hat{\lambda}_*$ in terms of the effective potential it feels,
\begin{equation}
\label{eq:effintr}
    \widehat{V}_*'(\hat{\lambda}_*) =  - \frac{m \beta}{N}, \qquad \hat{\lambda}_* \leq \hat{\lambda}_0,
\end{equation}
where recall that, on the support of $\hat{\sigma}_*$, the equilibrium condition from \cref{eq:deffV0} says
\begin{equation}
\label{eq:equieffv}
    \widehat{V}_*'(\hat{\lambda}) = 0, \qquad \forall\hat{\lambda}\in\supp\hat{\sigma}_*.
\end{equation}
Intuitively, this means that the lowest eigenvalue feels a force proportional to $m\beta/N$ pushing it to smaller values away from the rest, which stay at equilibrium on $\supp\hat{\sigma}_*$ feeling no force.
Written out, \cref{eq:effintr} gives the equation of motion for the instanton that we were after:
\begin{equation} 
\label{eq:eigenwan}
    V_\upnu'(\hat{\lambda}_*) +
    \frac{N-1}{N} \, \upbeta \int\displaylimits_{\mathbb{R}} dx \, \frac{\hat{\sigma}_*(x)}{x - \hat{\lambda}_*} = - \frac{m \beta}{N}, \qquad \hat{\lambda}_* \leq \hat{\lambda}_0.
\end{equation}
It is not necessary to write the integral as a principal value now because $\hat{\lambda}_*$ is required to lie off the support of $\hat{\sigma}_*$ anyway.
This general analysis explains the behavior exemplified in \cref{fig:instanton,fig:lowest,fig:halfhalf}, and realizes \cref{eq:betames} as
\begin{equation}
    {\hat{\rho}^m_\beta}{}_*(x) = \delta(x - \hat{\lambda}_*) + (N-1) \, \hat{\sigma}_*(x).
\end{equation}
But recall that off $\supp\hat{\sigma}_*$, the left-hand side of \cref{eq:effintr} is precisely what defines the spectral curve in \cref{eq:speyc}. Hence in terms of $y$, \cref{eq:effintr} reduces to simply\footnote{\,Remember that $\upbeta$ and $\beta$ have nothing to do with each other: the former is the matrix ensemble parameter specifying a symmetry group, whereas the latter is the inverse temperature of the canonical partition function.}
\begin{equation}
\label{eq:specintr}
    y(\hat{\lambda}_*) = \frac{m \beta}{\upbeta N}, \qquad \hat{\lambda}_* \leq \hat{\lambda}_0.
\end{equation}
This final expression for the location of $\hat{\lambda}_*$ is particularly powerful for several reasons. We were after a saddle-point evaluation of \cref{eq:contzbm}, which required us to find saddle points $\hat{\lambda}_*$ and  $\hat{\sigma}_*$ for each integral. The equation of motion we found for $\hat{\sigma}_*$ is \cref{eq:singint}, which at large $N$ describes nothing but the usual leading spectral density of the ensemble and is explicitly solved by \cref{eq:rho0main}. As for $\hat{\lambda}_*$, what we found is the algebraic equation in \cref{eq:specintr}, where $y$ follows directly from analytically continuing $\hat{\sigma}_*$ as in \cref{eq:rhotoy}. All in all, this means that our single-eigenvalue instanton for $\expval{Z(\beta)^m}$ is fully characterized just from knowledge of the leading spectral density of the ensemble.

Very importantly, our final results do not depend on objects like the matrix action and potential, which are not defined in scaling limits. Instead, all we need is access to the leading spectral density of the matrix ensemble, which is well-defined under matrix scaling. This obviously includes matrix ensembles in the DSL, as relevant for gravity. For these, not only are the requisite ingredients well-defined on the matrix side, but also accessible from the gravity side. Indeed, as we elaborate on in \cref{ssec:lessons}, all one needs is to account for a particular type of brane effects that the theory requires.

\subsection{Instanton Action}
\label{ssec:insact}

We would like to finally obtain an expression for \cref{eq:momentint}, the matrix integral that computes moments of $Z(\beta)$. At the saddle-point level, this amounts to evaluating the action of the single-eigenvalue instanton found in \cref{ssec:insti}. This instanton corresponds to the full saddle-point approximation of both $\hat{\sigma}$ and $\hat{\lambda}$ applied to \cref{eq:beast}. Recalling the saddle-point result for $\hat{\sigma}_*$ from \cref{eq:semilarge} and further extremizing in $\hat{\lambda}$ gives
\begin{equation}
\label{eq:finalZ}
    \expval{ Z(\beta)^m}_*
    =
    \frac{
    e^{- N \widehat{V}_*(\hat{\lambda}_*)}
    }{
    e^{- N \widehat{V}_*(\hat{\lambda}_0)}}
    Z[\hat{\sigma}_*;\hat{\lambda}_*]^m,
\end{equation}
where in evaluating $\widehat{\mathcal{Z}}_*$ we have to approximate \cref{eq:Zstar} by a saddle point $\hat{\lambda}_0$ for $\hat{\lambda}$ corresponding to $m=0$. This is the solution $\hat{\lambda}_*=\hat{\lambda}_0$ to \cref{eq:specintr} for $m=0$, which is given by \cref{eq:groundinf} as can be easily seen by comparing \cref{eq:effintr,eq:equieffv}.
For $m=0$ we thus obtain unity, in consistency with the proper normalization of the saddle-point approximation. 
Using now \cref{eq:deteff}, we reduce \cref{eq:finalZ} to the final expression
\begin{equation}
\label{eq:instZbeta}
    \expval{ Z(\beta)^m}_* = e^{- N I_*(m)},
\end{equation}
where we have defined the instanton action,
\begin{equation}
\label{eq:instI}
    I_*(m) \equiv \widehat{V}_*(\hat{\lambda}_*) - \widehat{V}_*(\hat{\lambda}_0) - \frac{m}{N} \log Z[\hat{\sigma}_*;\hat{\lambda}_*].
\end{equation}
Here, recall that $\hat{\sigma}_*$ and $\hat{\lambda}_*$ are the solutions to \cref{eq:singint,eq:effintr}, respectively, and that $\hat{\lambda}_*$ implicitly depends on $m$, with $\hat{\lambda}_*=\hat{\lambda}_0$ for $m=0$. This is our general result for the saddle-point approximation to $\expval{ Z(\beta)^m}$ at large $\beta$ and $N$ for any value of $m$. Written out in full glory,
\begin{equation}
\label{eq:zgenmom}
     \expval{ Z(\beta)^m}_* = e^{-N(\widehat{V}_*(\hat{\lambda}_*) - \widehat{V}_*(\hat{\lambda}_0))}
    \left(
    e^{-\beta\hat{\lambda}_*} + (N-1) \int\displaylimits_{\mathbb{R}} dx \, \hat{\sigma}_*(x) \, e^{-\beta x}
    \right)^m.
\end{equation}
Because the instanton eigenvalue $\hat{\lambda}_*$ depends non-trivially on $m$, this result clearly does not factorize for $m\geq1$ moments. This is consistent with the expectation that factorization in $m$ should only occur at small $\beta$. Although we are primarily interested in large $\beta$, note that our results also capture this asymptotic factorization in the small-$\beta$ regime. Namely, at small $\beta$ the insertion of $Z(\beta)$ does not modify the integral saddle, and indeed we see from \cref{eq:effintr,eq:equieffv} that to leading order the solution for $\hat{\lambda}_*$ recedes to the spectral edge at $\hat{\lambda} = \hat{\lambda}_0$. Hence \cref{eq:zgenmom} reduces to
\begin{equation}
\label{eq:forbapp}
    \expval{ Z(\beta)^m}_* \approx 
    \left(
    e^{-\beta\hat{\lambda}_0} + (N-1) \int\displaylimits_{\mathbb{R}} dx \, \hat{\sigma}_*(x) \, e^{-\beta x}
    \right)^m \approx \left(
    N \int\displaylimits_{\mathbb{R}} dx \, \hat{\sigma}_*(x) \, e^{-\beta x}
    \right)^m, \qquad \beta\ll N
\end{equation}
where in the last approximation we used large $N$. Not only do we see factorization at small $\beta$, but also consistency with the finite-$\beta$ result that one can easily compute for $m=1$. This is the simple case in which the matrix integral only depends on the average spectral density (cf. \cref{eq:anlogexp}). A saddle-point approximation for this case would just give
\begin{equation}
\label{eq:Zm1}
    \langle Z(\beta) \rangle_* = N \int\displaylimits_\mathbb{R} dx \, \hat{\rho}_*(x) \, e^{-\beta x}, 
\end{equation}
with $\hat{\rho}_*$ the solution to \cref{eq:expcond}. Clearly \cref{eq:forbapp} agrees with \cref{eq:Zm1} for $m=1$, up to the separate treatment of the lowest eigenvalue in \cref{eq:forbapp}, or in the last approximation up to the mismatch between $N$ and $N-1$ in the equations for $\hat{\rho}_*$ and $\hat{\sigma}_*$, negligible at large $N$. Namely, note that \cref{eq:eigenwan} to leading order at large $N$ with $\beta$ kept fixed reduces to \cref{eq:expcond}.
If $\beta$ is large, however, the spectral density term in \cref{eq:zgenmom} in fact becomes small against the eigenvalue instanton contributions. In this regime $\hat{\lambda}_*$ also does not necessarily lie close to $\hat{\lambda}_0$, so the large-$\beta$ behavior of \cref{eq:zgenmom} turns out to actually be dominated by
\begin{equation}
\label{eq:mmzlarge}
     \expval{ Z(\beta)^m}_* \approx  e^{-m\beta\hat{\lambda}_* -N(\widehat{V}_*(\hat{\lambda}_*) - \widehat{V}_*(\hat{\lambda}_0))}.
\end{equation}
Clearly, this behavior would be completely missed were the eigenvalue instanton to be ignored. In particular, this has implications for any moment $m$, and thus not only for the quenched logarithm, but also for the annealed one, as we discuss in \cref{ssec:quni}.

We emphasize that the large-$N$ result for $\expval{ Z(\beta)^m}_*$ obtained here applies in full generality to any matrix integral with an arbitrary matrix potential and eigenvalue measure of the general form posed in \cref{eq:genjac}. Furthermore, notice that \cref{eq:instI} only depends on quantities like the effective potential, which are well-defined in the DSL ensembles relevant for gravity.

As explained in \cref{ssec:doubscal}, for DSL integrals the typical starting point is the leading spectral density $\hat{\rho}_*$ itself. 
We already pointed out at the end of \cref{ssec:insti} that $\hat{\rho}_*$ is actually all one needs to define the equation of motion for $\hat{\lambda}_*$, given that $\hat{\rho}_*$ uniquely determines the spectral curve. Furthermore, for the same reason, we see that $\hat{\rho}_*$ itself also defines the effective potential through \cref{eq:speyc} up to an integration constant. This constant is irrelevant, as it is cancelled out by the difference of effective potentials in \cref{eq:instI}. It thus follows that everything that is needed to define and evaluate the instanton action is $\hat{\rho}_*$, which can be taken for granted in DSL ensembles.

Although $\hat{\rho}_*$ in the DSL is generally non-normalizable due to the scaling, the correct identification with $\hat{\sigma}_*$ is unambiguous. Namely, the relevant DSL spectral density is to be obtained by scaling a normalized spectral density, thus giving a $\hat{\rho}_*$ that is non-extensive in the DSL parameter $e^{S_0}$ (cf. the left-hand side of \cref{eq:noscal}). As a result, \cref{eq:zgenmom} in the DSL becomes simply
\begin{equation}
\label{eq:zgendsl}
     \expval{ Z(\beta)^m}_* = e^{-e^{S_0} \widehat{V}_*(\hat{\lambda}_*)}
    \left(
    e^{-\beta\hat{\lambda}_*} + e^{S_0} \int\displaylimits_{\mathbb{R}} dx \, \hat{\rho}_*(x) \, e^{-\beta x}
    \right)^m,
\end{equation}
where we have used the convention that $\widehat{V}_*(0) = 0$ at the leading spectral edge.

\subsection{Quenched Universality}
\label{ssec:quni}

The crucial question that remains is whether our results for $\expval{Z(\beta)^m}_*$ succeed in implementing the replica trick in \cref{eq:reptrick}, which takes the subtle $m\to0$ limit. In other words, we still have to understand whether working with just a saddle-point approximation for simultaneously large $N$ and $\beta$ suffices to restore the non-negativity of the quenched entropy, \cref{eq:quenS}, against the arbitrary negativity of the annealed one, \cref{eq:annS}.

All of our results follow from a large-$N$ treatment while accounting for a $\beta\sim O(N)$ scaling, but no assumption was ever made about $m$. In this paper $m$ started life as a non-negative integer determining the moment $m$ of $Z(\beta)$ under consideration. But in fact, notice that beginning already in \cref{eq:Zmom}, our equations and limits all make sense for any real $m\geq0$.
Our analytic continuation of $m$ from integers to reals is thus unambiguously unique and well-defined from the onset.\footnote{\,
This surmounts the difficulties that arise in \cite{Engelhardt:2020qpv,Janssen:2021stl}: uniqueness of our continuation is unaffected by whether or not the moments of $Z(\beta)$ satisfy Carleman's condition, and the replica trick can be unambiguously applied.}
A controlled way of understanding what happens to our results as $m\to0$ is to study the equations that govern our saddle points in this limit (cf. the approach of \cite{Chandrasekaran:2022asa} in semiclassical gravity). 

The pertinent equations of motion are \cref{eq:singint,eq:specintr}, where clearly only the latter depends on $m$. Actually, both $m$ and $\beta$ appear just in \cref{eq:specintr}, and only through the effective parameter $m\beta$.
It would thus seem that the large-$\beta$ regime we are interested in is in direct conflict with the $m\to0$ replica limit (cf. the discussion at the end of \cref{ssec:qexp}). In fact, every large-$\beta$ approximation we made was perfectly consistent with $m$ being arbitrarily small. In addition, the replica trick crucially takes $m\to0$ at fixed $\beta$, so there is no ambiguity in the order of limits. Namely, $m$ becomes arbitrarily small in the replica trick, so the effective parameter $m\beta$ must be treated as small, even at large $\beta$.\footnote{\,\label{fn:mbeff}While we can take $m\beta\geq0$ all the way to zero, it is still important that in our treatment we identified a modified saddle point that involves large $\beta$. In other words, because our interest is in the large-$\beta$ regime, small $m\beta$ generally just means that $m$ is much smaller than $1/\beta$ or, equivalently, $m\beta\ll1\ll\beta\sim O(N)$.} This observation has far-reaching consequences which we elaborate on in \cref{ssec:lessons}.

As $m\to0$, we expect the saddle $\hat{\lambda}_*$ to return to the lower edge of the spectrum. Indeed, the solution to \cref{eq:specintr} in this limit is $\hat{\lambda}_* = \hat{\lambda}_0$ as can be easily seen by comparing \cref{eq:effintr,eq:equieffv}. By continuity of the spectral curve away from $\supp\hat{\sigma}_*$, for small $m\gtrsim0$ we thus expect $\hat{\lambda}_*$ to be in the neighborhood of $\hat{\lambda}_0$. The universality discussion in \cref{sssec:univ} implies that the near-edge region of the spectral curve of any matrix integral is solely determined by the ensemble measure, and independent of any choice of matrix potential. This already leads us to expect that the behavior of our large-$\beta$ saddle in the replica trick will be universal, and thus yield a universal result for the quenched logarithm. Importantly, this in particular means that the kind of non-perturbative details in matrix integrals that gravity does not have access to will be completely irrelevant in our leading-order results.

Applying the replica trick in \cref{eq:reptrick} to \cref{eq:instZbeta},
\begin{equation}
\label{eq:zbtrick}
    \left\langle \log Z(\beta) \right\rangle_* = \lim_{m\to0} \frac{d}{dm} \expval{ Z(\beta)^m}_* = -N \lim_{m\to0} \frac{d I_*(m)}{dm},
\end{equation}
where we have used that $I_*(0)=0$. The evaluation of the total derivative in $m$ is particularly simple if one recalls that $I_*$ is precisely a stationary point of the matrix action in $\hat{\lambda}$. Hence any implicit dependence on $m$ coming from the stationary value $\hat{\lambda} = \hat{\lambda}_*$ will give zero, implying that
\begin{equation}
    \frac{d I_*(m)}{dm} = \frac{\partial I_*(m)}{\partial m}.
\end{equation}
Since the only explicit dependence of $I_*(m)$ on $m$ is through the prefactor of the logarithm, the result for the quenched logarithm in \cref{eq:zbtrick} is simply
\begin{equation}
    \left\langle \log Z(\beta) \right\rangle_* = \lim_{m\to0} \log Z[\hat{\sigma}_*;\hat{\lambda}_*] = \log Z[\hat{\sigma}_*;\hat{\lambda}_0],
\end{equation}
Explicitly, using \cref{eq:Zcont} our final expression is
\begin{equation}
\label{eq:finqlog}
    \left\langle \log Z(\beta) \right\rangle_* = \log( e^{-\beta\hat{\lambda}_0} + N \int\displaylimits_{\mathbb{R}} dx \, \hat{\sigma}_*(x) \, e^{-\beta x}),
\end{equation}
where recall $\hat{\lambda}_0 = \inf \supp \hat{\sigma}_*$ and $\hat{\sigma}_*$ is the leading spectral density of the matrix ensemble. If normalizable, $\hat{\sigma}_*$ is unit-normalized, but if not (as in the DSL), $\hat{\sigma}_*$ is a spectral density with no scaling in the scale parameter $e^{S_0}$ (cf. the discussion around \cref{eq:zgendsl}). In extreme regimes of scaling $\beta$ against $N$, this quenched result behaves as\footnote{\,The extreme $\beta\gg N$ result realizes the expectations of \cite{Okuyama:2020mhl}, but using the replica trick.}
\begin{equation}
\label{eq:qlogcases}
    \left\langle \log Z(\beta) \right\rangle_* \approx 
    \begin{cases}
        \displaystyle \log N + \log\int\displaylimits_{\mathbb{R}} dx \, \hat{\sigma}_*(x) \, e^{-\beta x}, \qquad &\beta\ll N, \\
        \displaystyle -\beta\hat{\lambda}_0 + N \int\displaylimits_{\mathbb{R}} dx \, \hat{\sigma}_*(x) \, e^{-\beta (x-\hat{\lambda}_0)}
        , \qquad  &\beta \gg N.
    \end{cases}
\end{equation}
Note that the last integral is a Laplace transform for a canonically shifted spectral density, and thus tends to zero in the large-$\beta$ limit by the results of \cref{sec:fact}.
These quenched result should be contrasted with the annealed one, which corresponds to taking the logarithm of the $m=1$ moment. This case corresponds to just \cref{eq:zgenmom} with $m=1$.
At small or finite $\beta$, we already argued around \cref{eq:forbapp} that this approximately reduces to
\begin{equation}
\label{eq:wignertype}
     \log\left\langle Z(\beta) \right\rangle_* \approx \log N + \log \int\displaylimits_{\mathbb{R}} dx \, \hat{\sigma}_*(x) \, e^{-\beta x}, \qquad \beta\ll N,
\end{equation}
thereby matching the small-$\beta$ behavior of the quenched logarithm in \cref{eq:qlogcases}.
If $\beta$ is large, however, this term becomes small relative to the eigenvalue instanton, leading to
\begin{equation}
\label{eq:m1zlarge}
     \log \left\langle Z(\beta) \right\rangle_* \approx  -\beta\hat{\lambda}_* -N(\widehat{V}_*(\hat{\lambda}_*) - \widehat{V}_*(\hat{\lambda}_0)), \qquad \beta\gg N,
\end{equation}
which is drastically different from the quenched result in \cref{eq:qlogcases}, as will more explicitly be seen in concrete examples below.
This large-$\beta$ behavior of $\expval{Z(\beta)}$ also differs significantly from what one would obtain by solving the large-$N$ matrix integral without scaling $\beta$, which would just give \cref{eq:wignertype} as well. So we observe that our large-$\beta$ eigenvalue instanton is not only necessary to obtain the quenched logarithm, but it also qualitatively modifies the behavior of $\expval{Z(\beta)}$ from \cref{eq:wignertype} to \cref{eq:m1zlarge} at large $\beta$. Intuitively, this is because our instanton captures statistics of the tail of eigenvalues in the classically forbidden region, which always matter at large $\beta$.

It only remains to apply our results for the logarithm of $Z(\beta)$ to the calculation of entropies. The extreme $\beta\gg N$ is of particular interest, given the expectation that the quenched entropy must remain non-negative and that the annealed entropy must become arbitrarily negative.
For the quenched entropy, applying \cref{eq:quenS} to \cref{eq:finqlog},
\begin{equation}
    S_q(\beta) \approx N \int\displaylimits_{\mathbb{R}} dx \, \hat{\sigma}_*(x) \, e^{-\beta (x-\hat{\lambda}_0)}, \qquad \beta\gg N.
\end{equation}
As mentioned below \cref{eq:qlogcases}, this term goes to zero as $\beta\to\infty$, giving a vanishing quenched entropy in the extremal limit. This is precisely the expected result for the non-degenerate eigenvalues that the leading spectral density captures, and demonstrates the success of our saddle-point analysis in implementing the replica trick for the quenched entropy. For the annealed entropy, applying \cref{eq:annS} to \cref{eq:m1zlarge},
\begin{equation}
    S_a(\beta) \approx  -N(\widehat{V}_*(\hat{\lambda}_*) - \widehat{V}_*(\hat{\lambda}_0))
    , \qquad \beta\gg N.
\end{equation}
For well-defined matrix integrals, $\widehat{V}_*$ generally grows unbounded off $\hat{\sigma}_*$, and so does the magnitude of its gradient. Hence it follows from \cref{eq:effintr} that $S_a(\beta)\to -\infty$ as $\beta\to\infty$, as expected.
Explicit examples of matrix integral calculations of quenched and annealed entropies using our saddle-point analysis are given in \cref{sec:egs}.

\subsection{Lessons for Gravity}
\label{ssec:lessons}

In connection with JT gravity, \cite{Saad:2019lba} interpreted single-eigenvalue instantons as the analogue of Liouville theory ZZ branes \cite{Zamolodchikov:2001ah}, and determinants like that in \cref{eq:contdet} as FZZT branes \cite{Fateev:2000ik,Teschner:2000md}. Our single-eigenvalue instantons require a reconciliation of these two notions of brane effects. The key point is that the determinant in \cref{eq:beast} is not inserted here as a probe with eigenvalue $\hat{\lambda}$ fixed by hand (unlike FZZT branes), but arises naturally inside the integral with $\hat{\lambda}$ actually being dynamical (like ZZ branes). The gravitational dual of our single-eigenvalue instantons should thus correspond to branes with the energy eigenbrane boundary conditions of \cite{Blommaert:2019wfy,Goel:2020yxl}, but treated dynamically like the end-of-the-world (EoW) branes of \cite{Gao:2021uro}.

As understood in \cref{ssec:split}, what gives rise to these branes is the large-$\beta$ back-reaction of $Z(\beta)$ itself on the integral (with no additional probes), and their effect alone dominates $\expval{Z(\beta)^m}$ at sufficiently large $\beta$ to leading order.
In other words, the brane is not a non-perturbative object we add by hand to complete the theory and probe $\expval{Z(\beta)^m}$; the brane is already part of the theory and $\expval{Z(\beta)^m}$ does not even make sense at large $\beta$ without it.
This narrative translates directly to the gravity side, and says that EoW eigenbranes are necessarily part of any quantum gravity theory with a matrix dual. Once again, this is not a statement about optional non-perturbative corrections; it is a leading-order effect without which the theory does not make sense.\footnote{\,In gravity, branes are often associated to tiny doubly non-perturbative corrections of $O(e^{-e^{S_0}})$ which are negligible semiclassically \cite{Saad:2019lba}. In our regime of interest, $\beta\sim e^{S_0}$ enhances these effects to the point that they become dominant to leading order at large $e^{S_0}$. For instance, branes are responsible for turning a negatively divergent annealed entropy into a non-negative quenched entropy, which is not a tiny effect.} The inclusion of such branes would not supplement the theory with a fixed discrete spectrum, much like there is no unique discrete spectrum in a matrix integral. Rather, their inclusion is needed for consistency with a random discrete spectrum and an ensemble description.

A gravity calculation of $\expval{Z(\beta)^m}$ would correspond to the usual gravitational path integral $\mathcal{P}(Z(\beta)^m)$ where one inserts $m$ boundaries all with $Z(\beta)$ boundary conditions. The inclusion of branes in the theory gives the option for geometries to end not only on $Z(\beta)$ boundaries, but also on EoW eigenbranes. It is well-understood that the contribution of fully-connected $m$-boundary wormholes to $\mathcal{P}(Z(\beta)^m)$ becomes important at large $\beta$, but that such geometries are off-shell and do not admit a semiclassical treatment \cite{Saad:2019lba,Engelhardt:2020qpv}. In fact, not even the off-shell contributions make sense at sufficiently large $\beta$, since the topological expansion breaks down when $\beta$ becomes comparable to $e^{S_0}$.\footnote{\,The statement that connected correlators dominate can still be made precise in the matrix formalism \cite{Johnson:2020mwi}.} Our findings suggest that a new semiclassical regime arises for $\beta\sim O(e^{S_0})$, effectively capturing a resummation of the full topological expansion to all genus. In particular, the new large-$\beta$ saddle point to leading order at large $e^{S_0}$ should be matched on the gravity side by a semiclassical geometry that dominates the gravitational path integral. Of course, such a geometry must necessarily involve eigenbrane boundaries, the preferred energy of which has to be dynamically determined by extremizing the action.\footnote{\,Thinking of the spectral curve as the target space on which branes propagate, \cref{eq:specintr} says how $\beta$ modulates their classical location. This is close in spirit to the Kodaira-Spencer formulation of JT gravity by \cite{Post:2022dfi,Altland:2022xqx}.} For the $m=1$ case that more directly descends from the near-horizon geometry of a near-extremal black hole, the dominance of a brane configuration would correspond to the black hole throat ending at a finite distance, rather than becoming infinitely deep as $\beta\to\infty$.

The emergence of an effective gravitational description in terms of branes when some other perturbative topological expansion breaks down is certainly not unheard of in string theory \cite{Polchinski:1995mt}. An example which beautifully illustrates this in a duality between gravity and matrix integrals is that of \cite{Drukker:2000rr,Drukker:2005kx,Okuyama:2018aij}.
This setting involves the computation of Wilson loops in $\mathcal{N}=4$ supersymmetric Yang-Mills using AdS/CFT at large $N$. For $1/2$ BPS Wilson loops the boundary calculation can be performed by a GUE matrix integral. Holographically, these can be computed from the world-sheet of a fundamental string anchored to the loop asymptotically. When the Wilson loop is multiply wrapped with winding number $k\geq 1$, the strings may coincide and interact among themselves. 
For $k\sim O(1)$, corrections are technically hard but in principle computable by including higher genus contributions. For $k\sim O(N)$, the topological expansion does not make sense anymore, and the many coincident strings are better described in terms of the dynamics of the D3-brane they end on. The classical action of the D3-brane solution is then seen to match precisely the result of the matrix dual. This is not merely analogous to what we are arguing for: a precise match between the winding number $k$ and our inverse temperature $\beta$ can be made and, as we show in \cref{sec:gwd}, our single-eigenvalue instantons exactly match the D-brane results of \cite{Drukker:2000rr,Drukker:2005kx,Okuyama:2018aij}. This provides a highly non-trivial consistency check of our saddle-point analysis of matrix integrals at large $\beta$, and the expectation of a dual description in terms of semiclassical geometries.

In trying to reproduce our large-$\beta$ results with a purely gravitational calculation, some previous explorations of JT gravity with branes and its matrix dual are useful. The work of \cite{Blommaert:2019wfy,Okuyama:2021eju} studied how $N$ eigenvalues react to probe branes fixing $O(1)$ of them. Instead, we are interested in how a dynamical brane eigenvalue reacts to the other $N-1$, which to leading order at large $\beta$ are fixed. The half-wormholes or correlators between FZZT branes and $Z(\beta)$ operators of \cite{Blommaert:2019wfy,Okuyama:2021eju} would seem to provide the right tool if the eigenbrane energies are treated dynamically rather than fixed by hand. A purely gravitational exploration of our large-$\beta$ instantons in terms of branes is currently underway \cite{intrepid}.

The discussion above applies to $\expval{Z(\beta)^m}$ for any integer $m\geq1$, but the $m\to0$ replica trick for quenching requires extra care. That branes are indispensable to capture the large-$\beta$ behavior of the quenched entropy is already clear.\footnote{\,An approach to computing quenched entropies unrelated to the replica trick but where the relevance of branes is also manifest was pursued within the matrix integral framework by \cite{Okuyama:2021pkf}.} How to make sense of non-integer $m$ in gravitational correlators, however, is not so clear. A natural approach one may be tempted to follow consists of computing moments for integer $m\geq1$, find their limiting behavior at large $\beta$ and, if their $m$-dependence happens to be sufficiently simple, analytically continue them to $m\to0$. As we explain next, even if the analytic continuation of the large-$\beta$ moments is unique (which typically is not the case), this strategy is fundamentally flawed and bound to give wrong results.

The basic reason the above does not work was already anticipated at the end of \cref{ssec:qexp} and in \cref{fn:mbeff}: $\beta\to\infty$ and $m\to0$ are opposing limits, and the order in which they are taken does not commute. This is obvious from the saddle-point equations of motion, where $m$ and $\beta$ only appear in \cref{eq:specintr} combined into $m\beta$. However, it is much obscurer from any explicit expression for the moments of $Z(\beta)$ like \cref{eq:zgenmom}, where $m$ and $\beta$ appear separately and also implicitly in $\hat{\lambda}_*$. In other words, at the level of $\expval{Z(\beta)^m}$ it is unclear how large $\beta$ and small $m$ compete with each other, and what approximations are allowed at large $\beta$ given that the strict $m\to0$ limit must be taken first.
This hidden competition in $\expval{Z(\beta)^m}$ explains the failure of previous attempts in the literature at performing this replica trick even on the matrix side. Indeed, as \cite{Johnson:2020mwi} observed, the pursuit of this strategy leads to a pathological $m\to0$ limit even when the matrix integral is treated fully non-perturbatively. In particular, one finds that the continuation in $m$ of the large-$\beta$ approximation for $\expval{Z(\beta)^m}$ gives a power-law divergence as $m\to0$, in perfect inconsistency with the requirement that $\expval{Z(\beta)^m}\to1$ in the no-replica limit. A ``replica-scaling'' on top of the DSL was proposed by \cite{Johnson:2020mwi} in order to obtain a sensible replica limit; unfortunately, by our general arguments above, we do not expect this to work.

The question thus remains of how to make sense of $\expval{Z(\beta)^m}$ for non-integer $m$ in gravity. A promising direction is that of \cite{Chandrasekaran:2022asa}, where the semiclassical quotient construction of \cite{Lewkowycz:2013nqa} was used to make the analytic continuation in $m$ well-defined. By studying the replica trick at the level of the equations of motion, this gravitational prescription allows for a controlled treatment of the competition between $m$ and $\beta$ just like on the matrix side. As understood here though, it would be essential to include branes in the calculation of \cite{Chandrasekaran:2022asa} in order to capture the large-$\beta$ regime.

Finally, let us comment on an additional subtlety that may arise in gravity in relation with replica symmetry. In gravity, the moments $\expval{Z(\beta)^m}$
correspond to gravitational correlators involving all possible wormholes consistent with boundary conditions, connected and disconnected. On the matrix side it is simple to study these moments directly, without even having to think about breaking them down into more fundamental cumulants $\expval{Z(\beta)^m}_c$. On the gravity side, however, the quantities one naturally computes are connected correlators involving connected wormholes, which correspond precisely to $\expval{Z(\beta)^m}_c$. Hence to implement the replica trick in gravity one has to work harder, and figure out how to combine connected correlators in a way that can be analytically continued in $m$. This raises the question of which $\expval{Z(\beta)^k}_c$ contributions with $1\leq k \leq m$ dominate $\expval{Z(\beta)^m}$, and how to even make sense of this question when $m<1$. This problem was addressed by \cite{Chandrasekaran:2022asa}, where a gravitational ansatz to account for replica symmetry breaking (RSB) was also proposed.

\begin{figure}
    \centering
    \includegraphics[width=0.7\textwidth]{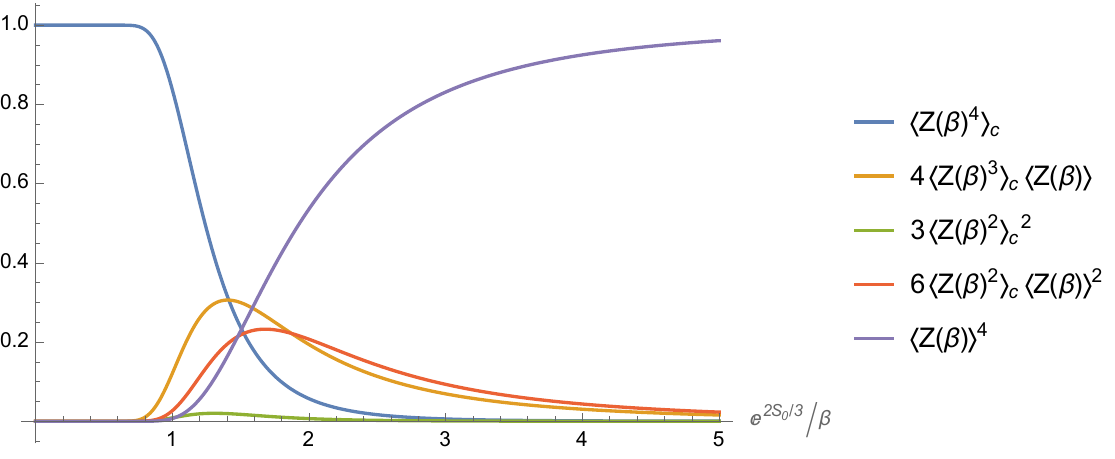}
    \caption{Illustration of the relative contribution to $\expval{Z(\beta)^m}$ from different connected correlator terms involving $\expval{Z(\beta)^k}_c$ for $1\leq k \leq m$. The plot considers $m=4$, for which $\expval{Z(\beta)^4}$ is equal to the sum of all terms appearing on the legend. The figure shows the ratio of each of these terms relative to $\expval{Z(\beta)^4}$ as a function of $\beta$ for the GUE with $e^{S_0}\sim N=10^6$. The small and large $\beta$ limits can be seen to follow \cref{eq:connbeta}. The transition between $\expval{Z(\beta)^4}$ being dominated by the replica-symmetric contributions $\expval{Z(\beta)^4}_c$ (fully connected) and $\expval{Z(\beta)}^4$ (fully disconnected) occurs at $\beta \sim e^{2S_0/3}$. This regime receives contributions of comparable order by partially-connected correlators which necessarily break replica symmetry.}
    \label{fig:cumul}
\end{figure}

Whether or not RSB plays a role in the gravitational implementation of the no-replica trick has remained a puzzle since the work of \cite{Engelhardt:2020qpv}. Our random matrix results seem to answer this in the affirmative: the $m\to0$ limit is fundamentally dominated by RSB configurations. To reach this conclusion, we may follow the strategy in \cite{Okuyama:2019xvg} of decomposing $\expval{Z(\beta)^m}$ into connected contributions $\expval{Z(\beta)^k}_c$ for each $1\leq k \leq m$ as a function of $m$ and $\beta$. The results of this exercise are illustrated in \cref{fig:cumul}.\footnote{\,The relation between moments and cumulants can be worked out from the standard generating function identity $$\left\langle e^{t Z(\beta)} \right\rangle = \exp(\sum_{k=1}^\infty \frac{t^k}{k!} \expval{ Z(\beta)^k}_c).$$} In general, for any integer $m\geq1$, one easily observes the following behavior:
\begin{equation}
\label{eq:connbeta}
    \lim_{\beta\to\infty} \expval{Z(\beta)^m} = \expval{Z(\beta)^m}_c, \qquad \lim_{\beta\to0} \expval{Z(\beta)^m} = \expval{Z(\beta)}^m.
\end{equation}
In gravitational terms, and without branes, this would mean that $\expval{Z(\beta)^m}$ is dominated by a fully connected wormhole at low temperatures, but by a fully disconnected topology at high temperatures. These two topologies are the only ones with only $Z(\beta)$ boundaries which can possibly accommodate full replica symmetry, which means that the transition between the two $\beta$ limits must involve RSB with partially connected wormholes being important (see \cref{fig:cumul}). In the large-$\beta$ regime, \cref{eq:connbeta} would suggest approximating $\expval{Z(\beta)^m}$ by $\expval{Z(\beta)^m}_c$ and analytically continuing to $m\to0$. But as we already explained, the $\beta\to\infty$ and $m\to0$ limits do not commute, and this approximation would indeed lead to pathological replica limit \cite{Okuyama:2020mhl}. On the other hand, were we to approximate $\expval{Z(\beta)^m}$ by $\expval{Z(\beta)}^m$, we would obtain an annealed answer. Since these are the only two replica-symmetric options without branes, an immediate corollary is that RSB is unavoidable in the $m\to0$ limit.
With branes, however, it is plausible that replica symmetry could be restored.
We leave an exploration of RSB and branes in the gravitational implementation of this replica trick to future work \cite{intrepid}.

\section{Explicit Examples}
\label{sec:egs}

We first do the exercise for Gaussian WD and BPS matrix ensembles without scaling, then for the soft edge of the universal Airy limit, and finally for the JT gravity matrix model in the DSL.

\subsection{Gaussian WD}
\label{sec:gwd}

For a WD ensemble, the eigenvalue measure for the random matrix is the general one in \cref{eq:genjac} with $\upnu=0$, and according to \cref{eq:WDH} the Hamiltonian is taken to be the random matrix itself. With a Gaussian potential $V(x)=\frac{1}{2}x^2$, the exact unit-normalized solution to \cref{eq:singint} for the saddle $\hat{\sigma}_*$ is just a Wigner semicircle,
\begin{equation}
\label{eq:leadgwd}
    \hat{\sigma}_*(x) = \frac{N}{N-1} \frac{1}{\pi\upbeta} \sqrt{2\upbeta \left(1-\frac{1}{N}\right) - x^2} = \frac{1}{\pi\upbeta} \sqrt{2\upbeta - x^2} + O(N^{-1}),
\end{equation}
when the radicand is non-negative, and zero otherwise. The corresponding spectral curve off the support of $\hat{\sigma}_*$ is given by
\begin{equation}
    y(x) = \sqrt{x^2 - 2\upbeta \left(1-\frac{1}{N}\right)} = \frac{1}{\upbeta} \sqrt{x^2 - 2\upbeta} + O(N^{-1}) 
\end{equation}
and the solution to the eigenvalue instanton condition, \cref{eq:specintr}, is
\begin{equation}
\label{eq:instl0}
\hat{\lambda}_* = - \sqrt{2\upbeta \left(1-\frac{1}{N}\right) + \left(\frac{m\beta}{N} \right)^2} = - \sqrt{2\upbeta + \left(\frac{m\beta}{N} \right)^2} + O(N^{-1}).
\end{equation}
As it should, note that $\hat{\lambda}_*$ lies below the support of $\hat{\sigma}_*$, and satisfies $\hat{\lambda}_*\to\inf\supp\hat{\sigma}_*$ as $m\to0$. 
Together, \cref{eq:leadgwd,eq:instl0} completely characterize our single-eigenvalue instanton. Dropping pesky factors that are $O(N^{-1})$ and thus negligible at large $N$ anyway, the leading effective potential for the spectral density in \cref{eq:leadgwd} can be written
\begin{equation}
    \widehat{V}_*(x) = \frac{\upbeta}{2}\left( 1 - x \, y(x) + 2 \log( - \frac{x}{\upbeta} - y(x)) \right).
\end{equation}
On the instanton eigenvalue from \cref{eq:instl0} this evaluates to
\begin{equation}
\label{eq:insteffv}
    \widehat{V}_*(\hat{\lambda}_*) = \upbeta \left( \frac{1-\log \upbeta/2}{2} + \kappa \sqrt{1+\kappa^2} - \sinh^{-1} \kappa\right), \qquad \kappa \equiv \frac{m\beta}{\sqrt{2\upbeta} N},
\end{equation}
where $\kappa$ is generally $O(m)$ given that both $N$ and $\beta$ scale together.
It only remains to obtain an expression for \cref{eq:Zcont} on the instanton solution,
\begin{equation}
    Z[\hat{\sigma}_*;\hat{\lambda}_*] = e^{\beta \sqrt{2\upbeta} \sqrt{1+\kappa^2}} + 2 N \frac{ I_1(\sqrt{2\upbeta} \beta)}{\sqrt{2\upbeta} \beta},
\end{equation}
where $I_1$ is the modified Bessel function of the first kind. Evaluating \cref{eq:instI}, upon a convenient rewriting we arrive at the following instanton action
\begin{equation}
\label{eq:exactmi}
    I_*(m) = - \upbeta \left(\kappa \sqrt{1+\kappa^2} + \sinh^{-1} \kappa\right)
    - \frac{m}{N} \log( 1 + 2N e^{-\beta \sqrt{2\upbeta} \sqrt{1+\kappa^2}} \frac{I_1(\sqrt{2\upbeta} \beta)}{\sqrt{2\upbeta} \beta}).
\end{equation}
The final result for \cref{eq:zgenmom} at the saddle-point level is thus
\begin{equation}
\label{eq:gwdzb}
    \expval{ Z(\beta)^m}_* = \left( 1 + 2N e^{-\beta \sqrt{2\upbeta} \sqrt{1+\kappa^2}} \frac{I_1(\sqrt{2\upbeta} \beta)}{\sqrt{2\upbeta} \beta}\right)^m \, \displaystyle e^{ N \upbeta\left(\kappa \sqrt{1+\kappa^2} + \sinh^{-1} \kappa\right)}.
\end{equation}
At large $\beta$ one may further reduce \cref{eq:gwdzb} to
\begin{equation}
\label{eq:fiol}
    \expval{ Z(\beta)}_* = \displaystyle e^{ N F(\kappa)}, \qquad F(\kappa) \equiv \upbeta \kappa \sqrt{1+\kappa^2} + \upbeta \sinh^{-1} \kappa. 
\end{equation}
For the GUE case ($\upbeta=2$) and $m=1$, this provides a remarkable consistency check of our formalism: \cref{eq:fiol} identically matches the results of \cite{Drukker:2000rr,Drukker:2005kx,Okuyama:2018aij} for $1/2$ BPS Wilson loops of $\mathcal{N}=4$ supersymmetric Yang-Mills at large $N$, 't Hooft coupling $\lambda$, and winding number $k\sim O(N)$, upon the identification $\beta=k\sqrt{\lambda}/2$. 
Here $F$ was obtained from the action of a single-eigenvalue instanton that arises at $\beta\sim O(N)$. In \cite{Drukker:2000rr,Drukker:2005kx,Okuyama:2018aij}, $F$ is the on-shell action of the D3-brane holographically dual to multiply wrapped Wilson loops at winding number $k\sim O(N)$. This is strongly suggestive that our large-$\beta$ instantons may generally be describable gravitationally in terms of dynamical branes. Furthermore, since \cref{eq:gwdzb} is also valid for $\beta\sim O(1)$, our results provide an interpolation between the effective brane description that arises at $k\sim O(N)$, and the worldsheet description in terms of weakly interacting fundamental strings that is natural for $k\sim O(1)$.

We would like to use \cref{eq:gwdzb} to compare quenched and annealed entropies.
For the former, using the replica trick from \cref{eq:reptrick} leads to a result of the universal form in \cref{eq:finqlog}, which for this Gaussian WD matrix integral reads
\begin{equation}
\label{eq:fiollog}
\begin{aligned}
    \expval{\log Z(\beta)}_* &=
    \log\left[ e^{\sqrt{2 \upbeta } \beta } + N \frac{2 I_1\left(\sqrt{2 \upbeta } \beta \right)}{\sqrt{2 \upbeta } \beta} \right] \\
    &\approx
    \begin{cases}
        \log N + \frac{\sqrt{2 \upbeta }}{N} \beta+\frac{\upbeta}{4} \beta^2, \qquad &\beta \ll N,\\
        \frac{\sqrt{2 \upbeta }}{2 N} \beta + \frac{2 N}{\sqrt{2 \pi } (2 \upbeta )^{3/4}} \beta^{-3/2}, \qquad &\beta \gg N.
    \end{cases}
\end{aligned}
\end{equation}
Using \cref{eq:quenS}, the extreme $\beta$ limits of the quenched entropy for this matrix integral are
\begin{equation}
\label{eq:sqguelims}
    S_q(\beta) \approx
    \begin{cases}
        \log N - \frac{\upbeta}{4} \beta^2,
        \qquad &\beta \ll N,\\
        \frac{5 N}{\sqrt{2 \pi } (2 \upbeta )^{3/4}} \beta^{-3/2}, \qquad &\beta \gg N.
    \end{cases}
\end{equation}
Clearly, at small $\beta$ we recover the expected maximal entropy for a system with $N$ degrees of freedom, whereas at large $\beta$ we obtain a consistently non-negative result that goes to zero in consistency with having a non-degenerate lowest eigenvalue. See \cref{fig:sngue} for an illustration.

\begin{figure}
    \centering
    \includegraphics[width=0.5\textwidth]{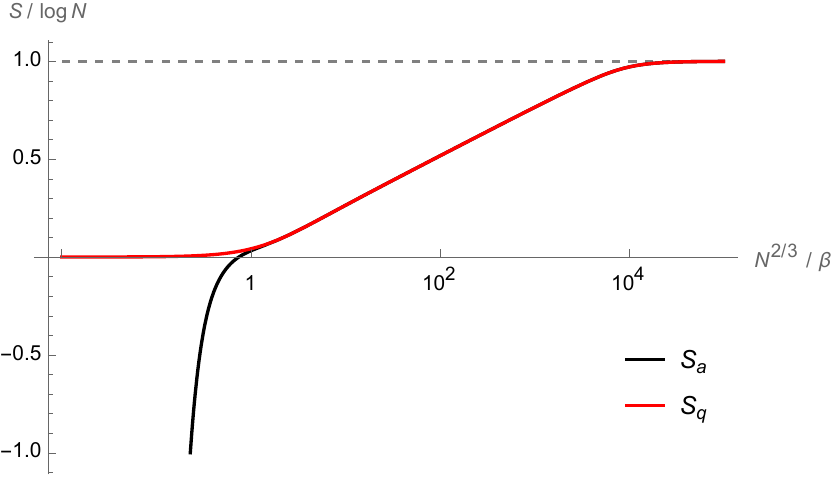}
    \caption{Comparison between annealed (black) and quenched (red) entropies for the GUE using the instanton result from \cref{eq:gwdzb}. These entropies begin to differ at $\beta \sim O(N^{2/3})$, which determines the transition to a low-temperature regime where the large-$\beta$ instanton dominates both quantities and cannot be neglected (cf. \cref{fig:aqgue}). Here $N=10^6$; in the thermodynamic limit $N\to\infty$, with an appropriate scaling $S_q(\beta)$ becomes non-differentiable at $\beta \sim O(N^{2/3})$, signaling a second-order phase transition. The extreme regimes $\beta \gg N$ and $\beta\ll N$ reproduce \cref{eq:sqguelims,eq:sague}, and are respectively consistent with \cref{thm:fact} and the statistical expectation that $e^{S}$ asymptote to the total number of degrees of freedom $N$ as $\beta\to0$.}
    \label{fig:sngue}
\end{figure}

These quenched results shall be compared to annealed ones, obtained by simply taking the logarithm of the $m=1$ expectation value in \cref{eq:gwdzb},
\begin{equation}
\label{eq:qague}
    \log \left\langle Z(\beta) \right\rangle_* = -I_*(1) \approx
    \begin{cases}
        \log N + \frac{\sqrt{2 \upbeta }}{N} \beta +\frac{\upbeta}{4} \beta^2, \qquad &\beta \ll N,\\
        \frac{\beta^2}{2 N}+\upbeta  N \log \frac{\beta}{N}, \qquad &\beta \gg N.
    \end{cases}
\end{equation}
which for the annealed entropy give the limiting behaviors
\begin{equation}
\label{eq:sague}
    S_a(\beta) \approx
    \begin{cases}
        \log N - \frac{\upbeta}{4} \beta^2,
        \qquad &\beta \ll N,\\
        -\frac{\beta^2}{2 N}+\upbeta  N \log \frac{\beta}{N}, \qquad &\beta \gg N.
    \end{cases}
\end{equation}
In consistency with the general arguments of previous sections, we see that quenched and annealed results agree at small $\beta$, but strongly disagree at large $\beta$. In particular, \cref{eq:sague} gives a negatively divergent annealed entropy as $\beta\to\infty$, exemplifying \cref{thm:fact}. This is shown in \cref{fig:sngue}.

It is also illustrative to compare these results with the na\"ive answer one would obtain by neglecting the large-$\beta$ instanton. In other words, if one simply evaluates $Z(\beta)$ on the WD semicircle saddle from \cref{eq:wdscal} (cf. using \cref{eq:Zm1}), one obtains
\begin{equation}
\label{eq:nail}
    \log \left\langle Z(\beta) \right\rangle_{\smalltext{na\"ive}} \approx
    \begin{cases}
        \log N+\frac{\upbeta}{4} \beta^2, \qquad &\beta \ll N,\\
        \sqrt{2\upbeta}\beta - \frac{3}{2} \log \beta, \qquad &\beta \gg N,
    \end{cases}
\end{equation}
which for the entropy gives
\begin{equation}
\label{eq:naiS}
    S_{\smalltext{na\"ive}}(\beta) \approx
    \begin{cases}
        \log N-\frac{\upbeta}{4} \beta^2, \qquad &\beta \ll N,\\
        - \frac{3}{2} \log \beta, \qquad &\beta \gg N,
    \end{cases}
\end{equation}
As previously pointed out, this gives the correct behavior at small $\beta$, but fails to capture the leading large-$\beta$ behavior of even the annealed entropy. In other words, our large-$\beta$ instantons are also important for annealed quantities, as we illustrate in \cref{fig:aqgue}.

\begin{figure}
    \centering
    \includegraphics[width=0.49\textwidth]{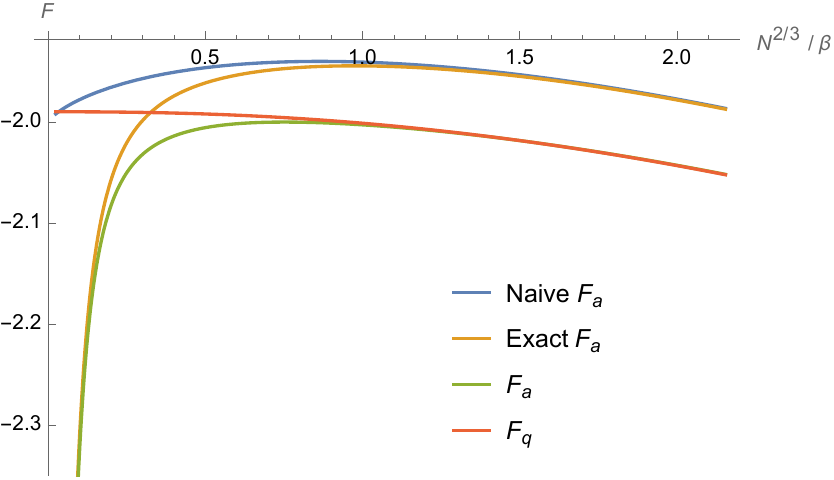}
    \hfill
    \includegraphics[width=0.49\textwidth]{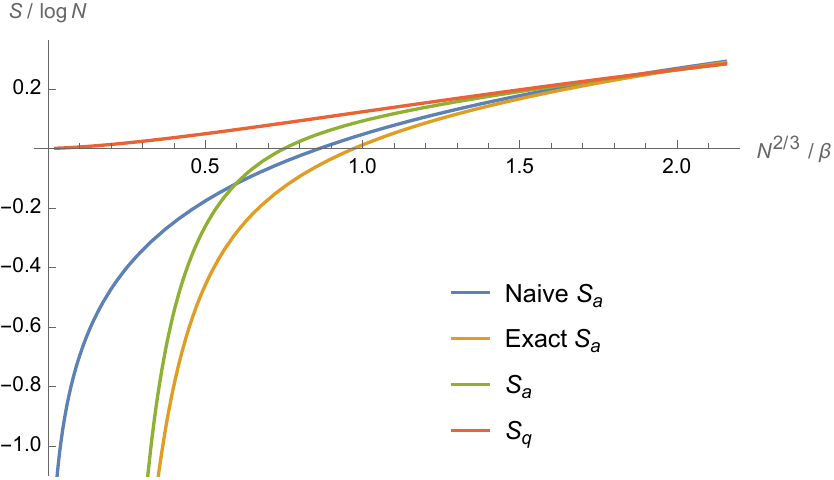}
    \caption{Free energies (left) and entropies (right) for the GUE with $N=100$. Na\"ive annealed results (blue) correspond to using the standard Wigner saddle from \cref{eq:wdscal} which does not take into account the large-$\beta$ instanton (cf. \cref{eq:nail,eq:naiS}). Exact annealed results (yellow) correspond to evaluating $\expval{Z(\beta)}$ exactly for $N=100$ using the orthogonal polynomial techniques from \cref{ssec:ortho}. Annealed results (green) correspond to using the large-$\beta$ instanton form for $\expval{Z(\beta)}_*$ from \cref{eq:gwdzb} at $m=1$ (cf. \cref{eq:qague,eq:sague}. Quenched results (red) correspond to using the large-$\beta$ instanton form for $\expval{Z(\beta)^m}_*$ from \cref{eq:gwdzb} in the $m\to0$ replica trick (cf. \cref{eq:fiollog,eq:sqguelims}). Only quenched quantities preserve physical properties such as the monotonicity of the free energy and the corresponding non-negativity of the entropy. Importantly, our instanton is seen to also be necessary for capturing the correct $\beta$ dependence of the annealed entropy at large $\beta$. In particular, although the na\"ive finite-$\beta$ saddle gives the desired qualitative negativity, only the large-$\beta$ instanton gives a quantitative approximation to the exact result at large $\beta$.}
    \label{fig:aqgue}
\end{figure}

As a matter of fact, the na\"ive result in \cref{eq:naiS} is actually the one that should be compared to the gravitational entropy in \cref{eq:genstt}. Indeed, \cref{eq:naiS,eq:genstt} exhibit the exact same logarithmic divergence at large $\beta$ for the bosonic $s=3/2$ case, whereas the leading divergence of our annealed result from \cref{eq:sague} is actually a quadratic power law. The reason the gravitational entropy behaves like \cref{eq:naiS} is simply that \cref{eq:genstt} is obtained to leading order at large $e^{S_0}$ at finite $\beta$, which generally corresponds to a lower-bounded leading spectrum with a classically forbidden region below it. By working at large $\beta\sim e^{S_0}$, our single-eigenvalue instanton is effectively accounting for corrections at higher $O(e^{-S_0})$, which capture the statistics of the tail of eigenvalues in the classically forbidden region. A gravitational calculation including $O(e^{-S_0})$ corrections would make the gravitational entropy reproduce \cref{eq:sague}.

\subsection{Gaussian BPS}
\label{ssec:gbps}

For BPS ensembles, the random matrix $M$ is $(N+\bar{\upnu})\times N$ and the Hamiltonian $H$ can be constructed out of $M$ as in \cref{eq:hcool}.\footnote{\,Once again, this is the case for $\mathcal{N}=1$ supersymmetry. See \cite{Turiaci:2023jfa} for the construction with higher supersymmetry.} The parameter $\bar{\upnu}$ determines the number of zero eigenvalues, which corresponds to the degeneracy of the ground state. The measure for the positive eigenvalues of $H$ is the general one in \cref{eq:genjac}.
The natural large-$N$ scaling for rectangular matrices corresponds to fixing their aspect ratio. Hence we let $\upnu\equiv N\nu$ and consider taking the large-$N$ limit keeping $(N+\upnu)/N = 1+\nu$ fixed. This amounts to keeping the proportion of zero eigenvalues against nonzero eigenvalues constant.
The case in which the ground-state degeneracy $\upnu$ is kept finite in the limit can be simply recovered by taking $\nu\to0$.

With a Gaussian potential for the singular values of $M$, the squared eigenvalues of $H$ feel a linear potential, so in particular $V(x)=\frac{1}{2}x$ in \cref{eq:vgapcon}.
The solution $\hat{\sigma}_*$ to \cref{eq:singint} for a support interval $\supp\hat{\sigma}_* = [a_-,a_+]$ takes the Mar\v{c}enko-Pastur form of \cref{eq:gaznu}
\begin{equation}
\label{eq:mpd}
    \hat{\sigma}_*(x) = \frac{N}{N-1} \frac{\sqrt{(a_+-x)(x-a_-)}}{2\pi\upbeta \, x}.
\end{equation}
Unit normalization of the measure requires
\begin{equation}
\label{eq:rho0daznorm}
    \sqrt{a_+} = \sqrt{a_-} + 2\sqrt{\upbeta}\sqrt{\frac{N-1}{N}},
\end{equation}
while direct evaluation of the principal value integral in \cref{eq:singint} leads to
\begin{equation}
    \fint\displaylimits_{[a_-,a_+]} dy \, \frac{\hat{\sigma}_*(y)}{x-y} = 
    \frac{N}{N-1} \frac{1}{2\upbeta} \left(1 - \frac{\sqrt{a_- a_+}}{x} \right).
\end{equation}
Matching to the potential gives a simple relation between the endpoints and the $\nu$ parameter,
\begin{equation}
    \sqrt{a_- a_+} = 2\nu.
\end{equation}
The explicit solution for the endpoints can be compactly written (cf. \cref{eq:endpoints})
\begin{equation}
    a_\pm = \upbeta \, \frac{N-1}{N} \, \left(1 \pm \sqrt{1+\frac{N}{N-1} \frac{2\nu}{\upbeta}}\right)^2= \upbeta  \left(1 \pm \sqrt{1 + \frac{2 \nu }{\upbeta }}\right)^2 + O(N^{-1}).
\end{equation}
As usual, notice that for $\nu>0$ the lower end of the $[a_-,a_+]$ interval is strictly positive, corresponding to a positive $\Delta\equiv a_->0$ gap in the spectrum between the zero eigenvalues and the non-trivial eigenvalues that $\hat{\sigma}_*$ describes.
The $\nu\to0$ limit recovers the gapless $\Delta\to0$ spectrum associated to a ground-state degeneracy $\bar{\upnu}$ that does not scale with $N$. In this case \cref{eq:mpd} degenerates into
\begin{equation}
\label{eq:degnuzero}
    \hat{\sigma}_*(x) \reprel{$\nu=0$}{=} \frac{1}{2\pi\upbeta} \sqrt{\frac{4\upbeta - x}{x}} + O(N^{-1}),
\end{equation}
which exhibits no gap and develops a pole at $x=0$. As explained in \cref{sec:linazens}, this is just the distribution that a Gaussian WD ensemble induces on squared eigenvalues (cf. \cref{eq:sqmu2})
Here we study \cref{eq:mpd} for general $\nu\geq0$, from which \cref{eq:degnuzero} is just the limiting case $\nu\to0$. The spectral curve corresponding to \cref{eq:mpd} for $x>0$ off the support of $\hat{\sigma}_*$ is
\begin{equation}
    y(x) = \frac{\sqrt{(a_+ - x) (a_- - x)}}{2 \upbeta  \, x}.
\end{equation}
The corresponding eigenvalue instanton that solves \cref{eq:specintr} is
\begin{equation}
\label{eq:instl0az}
\begin{aligned}
    \hat{\lambda}_* &= \frac{a_- + a_+ - \sqrt{(a_+ - a_-)^2+4 a_- a_+ \xi^2}}{2 \left(1-\xi ^2\right)} \\
    &= \frac{2 \left(\upbeta +\nu - \sqrt{\upbeta ^2+\nu  \left(2 \upbeta +\nu  \xi ^2\right)}\right)}{1-\xi ^2} + O(N^{-1}), \qquad \xi \equiv \frac{2m\beta}{N}.
\end{aligned}
\end{equation}
As required, note that $\hat{\lambda}_*$ lies below the support of $\hat{\sigma}_*$ but still on the allowed region of non-negative eigenvalues, i.e., $0\leq \hat{\lambda}_*\leq a_-$. In the limit $\nu\to0$ one has $\hat{\lambda}_*\to0$, consistent with the tightening of these bounds due to the spectral gap shrinking to zero size.
Integrating \cref{eq:speyc} gives the effective potential up to an unimportant constant as
\begin{equation}
    \widehat{V}_*(x) = -\upbeta\,x \,y(x)-\frac{1}{2} (a_-+a_+) \tanh ^{-1}\sqrt{\frac{a_- - x}{a_+ - x}} + \sqrt{a_- a_+} \tanh ^{-1}\sqrt{\frac{a_+ (a_- - x)}{a_- (a_+-x)}}.
\end{equation}
The evaluation of the effective potential for the instanton in \cref{eq:instl0az} does not simplify in any particular way at finite $\nu$ and $\xi$ even if $N\to\infty$, so it is not worth quoting explicitly. Furthermore, the integral in \cref{eq:Zcont} does not take a simple form either for the spectral density in \cref{eq:mpd} for general finite $\nu>0$. We nonetheless evaluate these results numerically in \cref{fig:aqbps}.

\begin{figure}
    \centering
    \includegraphics[width=0.5\textwidth]{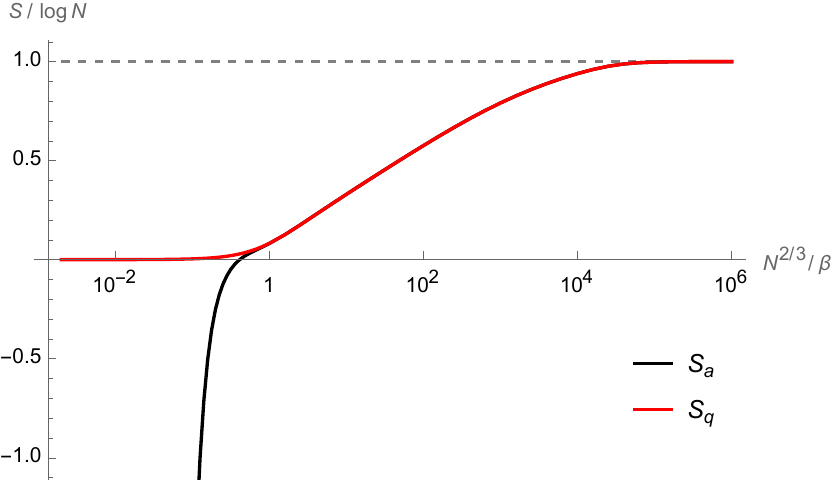}
    \caption{Comparison between annealed (black) and quenched (red) entropies for the non-zero eigenvalues of the Gaussian BPS ensemble with $\upbeta=2$ and $\nu=1$. The instanton solution from \cref{eq:mpd,eq:instl0az} is in this case evaluated numerically for the moments $\langle Z(\beta)^m \rangle_*$ setting $N=10^6$. As in the GUE case from \cref{fig:sngue}, the soft edge that the gapped BPS ensemble exhibits gives a transition to the large-$\beta$ instanton regime also at $\beta\sim O(N^{2/3})$. The extreme regimes $\beta \gg N$ and $\beta\ll N$ reproduce \cref{eq:gp2,eq:ap1}.}
    \label{fig:aqbps}
\end{figure}

The BPS ensembles are interesting both for $\nu>0$ and $\nu=0$, corresponding to a positive gap and no gap in the spectrum, respectively. These are qualitatively distinct regimes and, while $\nu=0$ can be studied as the $\nu\to0$ limit of the former, the converse is not true. Namely, it is not possible to study finite-gap $\nu>0$ results as an expansion about the gapless $\nu=0$ model. The reason is simply that the spectral density in \cref{eq:degnuzero} has support all the way to zero, and so will capture any quantity computed by expanding \cref{eq:mpd} about $\nu=0$. For instance, the allowed range for the eigenvalue instanton $\hat{\lambda}_*$ will remain to be just zero at any order in $\nu$, which is inconsistent with the existence of a finite gap in the spectrum for any finite $\nu>0$. Another way of seeing this failure is to simply note that the terms in an expansion of the spectral density in \cref{eq:mpd} about $\nu=0$ are not even integrable on the region of real support beyond second order.

Since a small-gap expansion makes no sense, we will only be interested in keeping finite $\nu>0$ or setting $\nu=0$ exactly. The latter is also better understood as a limit of the former, for otherwise at $\nu=0$ the instanton dynamics become trivial as its only allowed value is zero.
In order to obtain explicit results, we will consider the large-$\beta$ limit which, after all, is the regime of interest for us. As was the case in the WD example, our quantities of interest in this limit are also here governed by the large-$\beta$ instanton.
For the moments of $Z(\beta)$, using the large-$\beta$ result from \cref{eq:mmzlarge} and further expanding the exponent at large $\beta$ we get
\begin{equation}
\label{eq:larbm}
    \expval{Z(\beta)^m}_* = e^{N F(\xi)}, \qquad F(\xi) \equiv \frac{a_-+a_+}{2} \left( \tanh^{-1} \sqrt{\frac{a_-}{a_+}} - \frac{1}{2\xi} \right) - \nu \log \left(\frac{8\nu\xi}{a_+ - a_-} \right).
\end{equation}
This result should be compared to the one in \cref{eq:fiol} for the WD case. In particular, it would be interesting to understand if $F$ in \cref{eq:larbm} similarly matches the on-shell action of any known gravitational brane solutions.

As emphasized previously, while \cref{eq:larbm} makes sense for integer $m\geq1$, this large-$\beta$ expansion is incompatible with taking the $m\to0$ replica limit. Instead, directly applying \cref{eq:qlogcases},
\begin{equation}
\label{eq:gp1}
    \expval{\log Z(\beta)}_* = - \beta\Delta + O(Ne^{-\beta\Delta}), \qquad \beta\gg N,
\end{equation}
with $\Delta=a_-$, and the exponentially small errors coming from the spectral density integral. It is worth keeping track of these as they emphasize that already the next-to-leading order corrections at large-$\beta$ are exponentially suppressed. This contrasts with the large-$\beta$ approximation for the annealed result that follows from \cref{eq:larbm} for $m=1$, where corrections appear at all orders in $1/\beta$.
The entropy that results from quenching behaves as
\begin{equation}
\label{eq:gp2}
    S_q(\beta) = O(\beta \, e^{-\beta\Delta}), \qquad \beta \gg N,
\end{equation}
whereas the annealed entropy yields
\begin{equation}
\label{eq:ap1}
    S_a(\beta) = - \nu N \log \frac{\beta}{N} + O(1), \qquad \beta\gg N.
\end{equation}
Once again, this verifies the general expectation from \cref{sec:fact} that the annealed entropy negatively diverges as $\beta\to\infty$, whereas the quenched entropy remains non-negative.

These are the results one obtains for the continuous spectrum of positive eigenvalues in BPS ensembles. Remember though that the spectrum of such ensembles also includes a discrete zero eigenvalue with degeneracy $\upnu = \nu N$. In the above we have analyzed the dynamical eigenvalues, and the entropic properties of their spectrum excluding the highly degenerate ground state from the partition function.
However, if we want to study the full spectrum of the Hamiltonian altogether, we may want to also take into account the non-dynamical zero eigenvalues. Denoting the partition function with zero eigenvalues included by $Z^0$, its moments are simply related to those of the partition functions above by
\begin{equation}
    \expval{Z^0(\beta)^m}_* = \expval{\left(\upnu + Z(\beta)\right)^m}_*.
\end{equation}
As expected from \cref{sec:fact}, the presence of a nonzero number of discrete ground states, extensive or not in $N$, drastically changes the large-$\beta$ behavior of the annealed entropy from negatively divergent to non-negative and finite.
In particular, for the extensive case above we obtain
\begin{equation}
\label{eq:ann0}
    S_a^0(\beta) = \log \upnu + O(\beta^{-\upnu}),  \qquad \beta \gg N
\end{equation}
whereas the quenched entropy becomes\footnote{\,This form of the subleading correction assumes $\Delta>0$, i.e., an extensive degeneracy $\upnu=\nu N$ with $\nu>0$. In the gapless case $\Delta=0$, even if $\upnu>0$, the leading correction is only suppressed at large $\beta$ by a power law.}
\begin{equation}
\label{eq:que0}
    S_q^0(\beta) = \log \upnu + O(\beta \, e^{-\beta\Delta}),  \qquad \beta \gg N
\end{equation}
Hence we see that both entropies now asymptote to the same value capturing the degeneracy $\upnu$ of the ground state of the Hamiltonian. This is precisely what happens when supersymmetry protects the ground-state energy, and why gravity despite computing an annealed quantity succeeds in computing the extremal entropy of BPS black holes. Nonetheless, as is clear from \cref{eq:ann0,eq:que0}, the approach to the strict $\beta\to\infty$ value of annealed and quenched entropies is very distinct qualitatively. In particular, the quenched entropy stays exponentially close to its extremal value up to $\beta^{-1}$ scales of the order of the gap size $\Delta$. In contrast, the annealed entropy departs from its extremal value polynomially in $\beta^{-1}$ at order $\upnu$. Hence the behavior of the entropies and their derivatives at gap scales will differ significantly whether one is quenching or annealing, as illustrated in \cref{fig:gapqa}.

\begin{figure}
    \centering
    \includegraphics[width=0.48\textwidth]{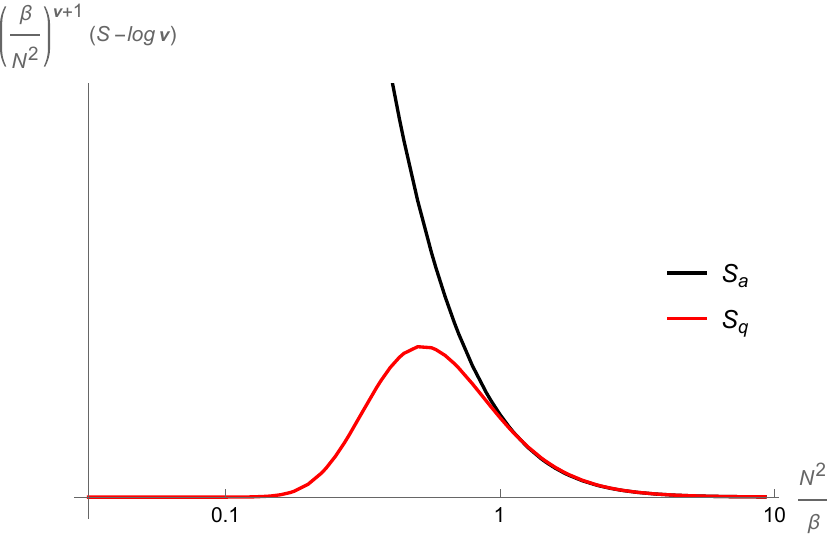}
    ~~
    \includegraphics[width=0.48\textwidth]{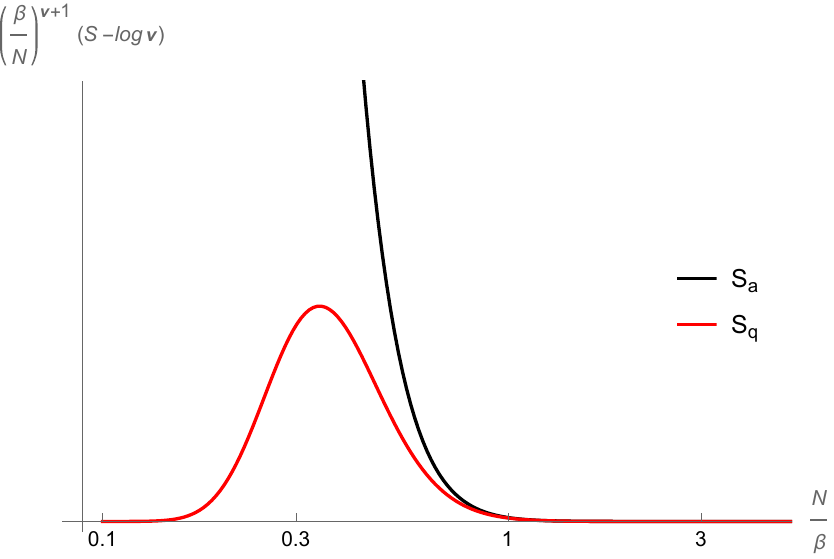}
    \caption{Entropies for Gaussian BPS ensembles with $\upbeta=2$ including the ground state with degeneracy $\upnu=\nu N$ in the spectrum. On the left $\upnu=1$, i.e., a non-extensive degeneracy giving a gapless BPS spectrum as $N\to\infty$; on the right $\nu=1$, corresponding to an extensive degeneracy $\upnu=N$ giving a gapped BPS spectrum as $N\to\infty$. In these plots $N=10^6$. In contrast with the situation in \cref{fig:aqbps}, in this case both the annealed (black) and quenched (red) entropies give $\log\upnu$ in the strict extremal limit. However, the approach to this asymptotic value is power-law for $S_a$ and exponential for $S_q$, as quoted in \cref{eq:ann0,eq:que0}, respectively. We demonstrate this discrepancy here by showing that as $\beta$ increases, $(\beta/N)^{\upnu+1}(S(\beta)-\log\upnu)$ diverges monotonically for $S_a$ but goes to zero for $S_q$. The scale of the vertical axis is irrelevant; on the horizontal axis, however, the entropies are generally observed to depart from each other at $\beta\sim O(N/\nu)$.}
    \label{fig:gapqa}
\end{figure}

The attentive reader might wonder why the gravitational entropy in \cref{eq:genstt}, despite being an annealed entropy, exhibits the exponential suppression of a gapped quenched entropy. This lucky accident occurs for the same reason explained at the end of \cref{sec:gwd}. Namely, \cref{eq:genstt} is a leading-order result at large $e^{S_0}$ and finite $\beta$, the underlying spectrum of which has a classically forbidden region $E<\Delta$. The statistics of eigenvalues in the gap region require $O(e^{-S_0})$ effects which would make the gravitational entropy in \cref{eq:genstt} actually behave like \cref{eq:ann0}, and not like \cref{eq:que0}.

To address the gapless case it is preferable to reintroduce $\upnu = N\nu$, and take the large-$N$ limit keeping $\upnu$ finite. The gap size is identically zero for $\upnu=0$ at any order in $N$, but if $\upnu>0$ then the gap size $\Delta\sim O(N^{-2})$ at large $N$, i.e., the gapless spectrum only occurs in the strict $N\to\infty$ limit. Without a gap, the spectral density is given by \cref{eq:degnuzero}, for which the integral in \cref{eq:zgenmom} is easily computable and reads
\begin{equation}
\label{eq:gapless}
    \int\displaylimits_{\mathbb{R}} dx \, \hat{\sigma}_*(x) \, e^{-\beta x} \reprel{$\nu=0$}{=} (I_0(2 \upbeta \beta)+I_1(2 \upbeta \beta )) e^{-2 \upbeta \beta}.
\end{equation}
Both annealed and quenched quantities behave equivalently at leading order at small $\beta$ and are governed by \cref{eq:gapless}. The small-$\beta$ results can be easily worked out to be
\begin{equation}
    \expval{\log Z(\beta)}_* \approx \log \expval{Z(\beta)}_* \approx \log N - \upbeta \beta + \frac{1}{2} \upbeta^2 \beta^2, \qquad \beta\ll N,
\end{equation}
for the logarithms, and
\begin{equation}
    S_q(\beta) \approx S_a(\beta) \approx \log N - \frac{1}{2} \upbeta^2 \beta^2, \qquad \beta\ll N,
\end{equation}
for the entropies.
At large $\beta$ the eigenvalue instanton once again dominates and the contributions from \cref{eq:gapless} become subleading. The results above still apply so long as $\upnu>0$, a parameter that can just be restored by resetting $\nu=\upnu/N$. In other words, for $\upnu>0$ one can just take the limit $N\to\infty$ at finite $\upnu$ in the expressions above to obtain gapless results. For quenched quantities, the $m\to0$ limit sends the instanton to the spectral edge $\hat{\lambda}_* \to \Delta$, which does not depend on $\beta$, and then the large-$N$ limit at finite $\upnu$ sends $\Delta\to0$. The upshot is a quenched logarithm of the form of \cref{eq:finqlog} with $\hat{\lambda}_0=0$ and a continuous part given by involving \cref{eq:gapless}. The resulting quenched entropy is clearly non-negative, as usual. For annealed quantities, however, $m=1$ and the instanton saddle $\hat{\lambda}_*$ remains non-trivial and $\beta$-dependent. Scaling $\beta$ with $N$, this gives rise to a regime in which the annealed logarithm becomes arbitrarily negative, which again results in an arbitrarily negative annealed entropy. This regime, however, only makes sense at nonzero $\upnu>0$, whether or not $\upnu$ scales with $N$.

The case $\upnu=0$ is strictly qualitatively different. According to \cref{eq:params}, this case corresponds to AZ ensembles with $\upalpha=1$ or BPS ensembles with $\bar{\upnu}=1$ and $\upbeta=1$ (a BPS ensemble with $\bar{\upnu}=0$ is just an AZ ensemble). In either case, the corresponding eigenvalue measure can be understood as coming from a WD ensembles upon squaring the random matrix. As such, the location of zero in the spectrum is not characterized by any particular eigenvalue, but rather by whichever eigenvalue is closer to zero in the original WD matrix. In other words, the statistics that govern the lowest eigenvalue in this case are not edge statistics, but rather bulk statistics of the eigenvalues that were squared. 
For completeness, we note that $\upnu=-1/2$ is also allowed for AZ ensembles with $\upalpha=0$. 
Our main motivation for studying BPS ensembles in this paper was to explore the gapped case, so we will not discuss these cases any further. A thorough exploration of the spectral statistics of ensembles which include $\upnu=0$ can be found in \cite{Johnson:2022wsr}.

\subsection{Airy Edge}
\label{ssec:airy}

As explained in \cref{sssec:univ}, the statistics of soft edges are universally captured by the Airy model in \cref{eq:airyf}. This includes WD and gapped BPS ensembles upon edge scaling, not only for the Gaussian cases studied in \cref{sec:gwd,ssec:gbps}, but for any matrix potential.
Recalling \cref{eq:airedge}, the leading spectral density of the Airy model is supported on $x>0$ and reads
\begin{equation}
\label{eq:airys}
    \hat{\rho}_*(x) = \frac{\sqrt{x}}{\pi}.
\end{equation}
The corresponding spectral curve along $x<0$ is simply
\begin{equation}
    y(x) = \sqrt{-x}.
\end{equation}
By \cref{eq:specintr} we obtain the instanton solution 
\begin{equation}
\label{eq:instairy}
    \hat{\lambda}_* = - \kappa^2, \qquad \kappa \equiv \frac{m\beta}{2e^{S_0}},
\end{equation}
where since $\upbeta$ can be absorbed in the zooming variable to leading order, we have simply chosen to fix $\upbeta=2$. Once again, the edge of the spectrum $\hat{\lambda}_*\to0$ is approached as $m\to0$. 
The effective potential along $x<0$ can be obtained by integrating \cref{eq:speyc}, which gives
\begin{equation}
    \widehat{V}_*(x) = \frac{4}{3} \left( - x\right)^{3/2}.
\end{equation}
For the instanton in \cref{eq:instairy} this evaluates to
\begin{equation}
    \widehat{V}_*(\hat{\lambda}_*) = \frac{4}{3} \kappa^3.
\end{equation}
The canonical partition function from \cref{eq:Zcont} yields
\begin{equation}
    Z[\hat{\sigma}_*;\hat{\lambda}_*] = e^{\beta \kappa^2} + \frac{e^{S_0}}{2\sqrt{\pi} \, \beta^{3/2}}.
\end{equation}
With a convenient rearrangement of terms, the resulting instanton action is
\begin{equation}
    I_*(m) = - \frac{2}{3} \kappa^3 - m \, e^{-S_0} \log(1 + \frac{e^{S_0 -\beta \kappa^2}}{2\sqrt{\pi} \, \beta^{3/2}} ).
\end{equation}
The final result for \cref{eq:zgendsl} in the Airy model is
\begin{equation}
\label{eq:airyzb}
    \expval{ Z(\beta)^m}_* = \left( 1 + \frac{ e^{S_0-\beta\kappa^2}}{2\sqrt{\pi} \, \beta^{3/2}} \right)^m e^{\frac{2}{3} e^{S_0} \kappa^3}.
\end{equation}
Using the replica trick from \cref{eq:reptrick} on \cref{eq:airyzb} gives the extreme limits
\begin{equation}
\label{eq:aiq}
    \expval{\log Z(\beta)}_* = 
    \log\left( 1 + \frac{e^{S_0}}{2\sqrt{\pi} \, \beta^{3/2}} \right) \approx
    \begin{cases}
        S_0 - \frac{3}{2} \log \beta + 2 \sqrt{\pi} e^{-S_0} \beta^{3/2}, \qquad &\beta \ll e^{S_0},\\
        \frac{e^{S_0}}{2\sqrt{\pi }} \beta^{-3/2}, \qquad &\beta \gg e^{S_0}.
    \end{cases}
\end{equation}
For the quenched entropy, these give
\begin{equation}
\label{eq:ais}
    S_q(\beta) \approx
    \begin{cases}
        S_0 - \frac{3}{2} \log \beta, \qquad & \beta \ll e^{S_0} \\
        \frac{5 e^{S_0}}{4\sqrt{\pi }} \beta^{-3/2}, \qquad & \beta \gg e^{S_0} \\
    \end{cases}
\end{equation}
In consistency with quenched universality, this large-$\beta$ result matches precisely that of the GUE, i.e., the Gaussian WD result for $\upbeta=2$ in \cref{eq:sqguelims}. 
In this case though, the entropy grows unbounded as $\beta\to0$, as expected from the non-normalizability of the spectral density of scaled matrix integrals (see \cref{fig:snairy}).
For the annealed entropy,
\begin{equation}
\label{eq:tred}
    S_a(\beta) \approx
    \begin{cases}
        S_0 - \frac{3}{2} \log \beta, \qquad & \beta \ll e^{S_0} \\
        - \frac{e^{-2S_0}}{6} \beta^3, \qquad & \beta \gg e^{S_0} \\
    \end{cases}
\end{equation}
Once again, as expected, small-$\beta$ behaviors match, while for large $\beta$ the annealed entropy diverges negatively. Finally, a na\"ive application of \cref{eq:Zm1} to \cref{eq:airys} yields
\begin{equation}
    S_{\smalltext{na\"ive}}(\beta) \approx
    \begin{cases}
        S_0 - \frac{3}{2} \log \beta, \qquad & \beta \ll e^{S_0} \\
         - \frac{3}{2} \log \beta, \qquad & \beta \gg e^{S_0} \\
    \end{cases}
\end{equation}
As was also observed in \cref{sec:gwd}, this na\"ive result reproduces the large-$e^{S_0}$ gravitational entropy from \cref{eq:genstt} for the bosonic $s=3/2$ case, but fails to capture the $O(e^{-S_0})$ corrections that our large-$\beta$ instanton accounts for in \cref{eq:tred}.

\begin{figure}
    \centering
    \includegraphics[width=0.5\textwidth]{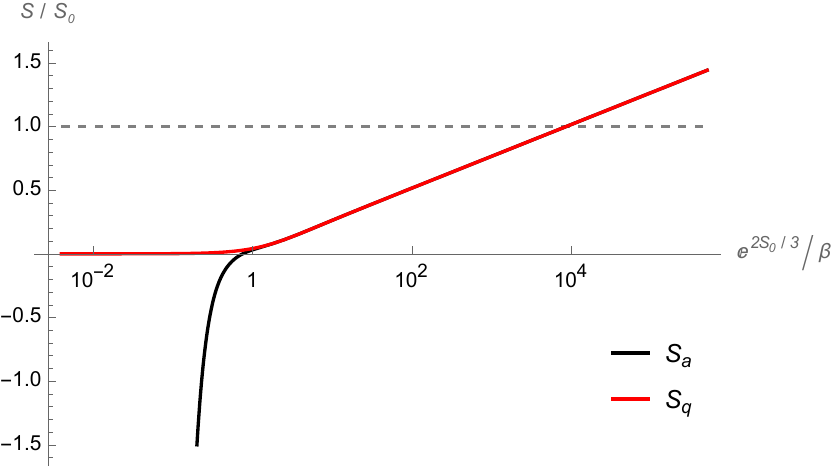}
    \caption{Comparison between annealed (black) and quenched (red) entropies for the Airy model using the instanton result from \cref{eq:airyzb}, and setting $e^{S_0}=10^6$. The transition to the instanton-dominated regime occurs at $\beta \sim e^{2S_0/3}$. Under the identification of the scaling parameter $N \sim e^{S_0}$, this is consistent with the examples in \cref{fig:sngue,fig:aqbps} given that both fall under the same soft edge universality class of the Airy model. The extreme regimes $\beta \gg N$ and $\beta\ll N$ reproduce the expectations from \cref{eq:ais,eq:tred}. In particular, the entropies grow unbounded as $\beta\to0$, as expected from the $N\to\infty$ limit of edge scalings.}
    \label{fig:snairy}
\end{figure}

\subsection{JT Gravity}
\label{ssec:jtsec}

The leading spectral density for pure JT gravity is given by \cref{eq:bos}.
With the desired scaling from \cref{eq:noscal} for DSL matrix integrals,
\begin{equation}
\label{eq:JTsigma}
    \hat{\rho}_*(x) = \frac{1}{2\pi^2} \sinh( 2\pi \sqrt{x}),
\end{equation}
where $x>0$. Near $x=0$, \cref{eq:JTsigma} reduces to the Airy soft edge from \cref{eq:airys}, as expected from universality.
The corresponding spectral curve along $x<0$ is
\begin{equation}
\label{eq:yjt}
    y(x) = \frac{1}{2\pi} \sin (2\pi \sqrt{-x}).
\end{equation}
This spectral curve is oscillatory and bounded to $2\pi y(x)\in[-1,1]$, so \cref{eq:specintr} would seem to admit either infinitely many or no instanton solutions. This is a manifestation of the fact that \cref{eq:JTsigma} alone does not actually completely define a matrix integral, but requires a non-perturbative completion \cite{Saad:2019lba,Johnson:2019eik}.
However, regardless of the non-perturbative completion, we know the only correct instanton solution for $m=0$ (or $\beta=0$) is \cref{eq:groundinf}, i.e., $\hat{\lambda}_*$ must lie at the $x=0$ edge of the spectrum. Hence independently of how one completes the model non-perturbatively, we can safely just focus on instanton solutions which are continuously connected to the edge of the spectrum as $m\to0$. In other words, when solving \cref{eq:specintr} we want to restrict our solutions to a domain connected to $x=0$ where the spectral curve $y$ is injective. Since $y$ attains its first maximum at $x_b=-1/16$, our desired domain is $x\in [x_b,0]$. The maximum here is $y(x_b) = 1/2\pi$, and thus the only possible solutions without a non-perturbative completion are
\begin{equation}
\label{eq:instjt}
    \hat{\lambda}_* = - \left(\frac{\theta}{2 \pi} \right)^2, \qquad \sin \theta \equiv \frac{\pi \, m \beta}{e^{S_0}}  \leq 1, \qquad 0\leq\theta\leq\frac{\pi}{2},
\end{equation}
where we have used $\upbeta=2$ for ordinary JT (cf. \cref{fn:ordjt}).
The parameter constraint above implies that these real instantons can only capture finite $m\geq1$ moments if $\beta$ is not too large.\footnote{\,As pointed out in \cref{fn:compsad}, it is in principle possible for the saddle points that dominate the $\hat{\lambda}$ integral in \cref{eq:semilarge} to be complex. Studying these would first require completing the JT matrix integral non-perturbatively with a valid eigenvalue contour of integration \cite{Johnson:2019eik}.}
Fortunately, this still allows one to probe $\beta \sim O(e^{S_0})$, even if there is a finite upper bound on $\beta/e^{S_0}$. Furthermore, this restriction does not affect the replica trick at all: as the limit $m\to0$ is taken, the real instanton exists for arbitrarily large $\beta$. In other words, the results we obtain for quenched quantities are valid at any $\beta$, including the limit $\beta\to\infty$.

By \cref{eq:speyc}, integrating \cref{eq:yjt} we find the leading effective potential on $x<0$ to be\footnote{\,The instanton ceases to exist at $x=-1/16$, sooner than this potential reaches its first $x=-1/4$ local maximum.}
\begin{equation}
    \label{eq:sickveff}
    \widehat{V}_*(x) = \frac{1}{2\pi ^3} \left( \sin \left(2 \pi  \sqrt{-x}\right) - 2\pi \sqrt{-x} \,\cos \left(2 \pi  \sqrt{-x}\right)
    \right).
\end{equation}
On the instanton from \cref{eq:instjt}, this evaluates to
\begin{equation}
    \label{eq:sickinst}
    \widehat{V}_*(\hat{\lambda}_*) = \frac{1}{2\pi ^3} \left( \sin\theta - \theta\,\cos\theta \right).
\end{equation}
The canonical partition function from \cref{eq:Zcont} gives
\begin{equation}
\label{eq:zjt}
    Z[\hat{\sigma}_*;\hat{\lambda}_*] = e^{\beta \left(\frac{\theta}{2\pi} \right)^2} + \frac{e^{S_0 + \pi^2/\beta}}{2\sqrt{\pi} \, \beta^{3/2}}.
\end{equation}
The second term is the familiar $1$-loop exact canonical partition function from \cref{eq:bos} that JT gravity yields on the disk. The first term is a novel contribution from the large-$\beta$ instanton. When $\theta\ll1$, corresponding to $m\beta e^{-S_0}\ll1$, this contribution behaves as
\begin{equation}
\label{eq:dnonp}
    e^{\beta \left(\frac{\theta}{2\pi} \right)^2} = e^{\gamma^2} + O(\theta^3), \qquad \gamma \equiv \frac{m  \beta^{3/2}}{2 \, e^{S_0}}
\end{equation}
This approximation clearly holds at large $S_0$ for any $\beta < O(e^{2S_0/3})$, but also for arbitrarily large $\beta$ in the $m\to0$ replica limit.
Gravitational effects that are doubly non-perturbative in $S_0$ are typically expected to give small corrections of $O(e^{-e^{S_0}})$; \cref{eq:dnonp} has a different flavor. From its form we see that in \cref{eq:zjt} the instanton contribution grows exponentially in $\gamma^2$, while the disk contribution goes to zero as $1/\gamma$. The two terms exchange dominance when $\gamma\sim O(1)$ or $\beta\sim O(e^{2S_0/3})$ at finite $m>0$. In particular, in this case the instanton contribution becomes dominant for $\beta \gtrsim O(e^{2S_0/3})$.
Hence there is a parametrically large regime of values of $\beta$ for which the instanton in \cref{eq:instjt} exists and dominates \cref{eq:zjt}, namely,
\begin{equation}
\label{eq:largebjt}
    O(e^{2S_0/3}) \lesssim \beta \lesssim O(e^{S_0}/m).
\end{equation}
Once again, the upper bound above matters for finite $m\geq1$ moments, but disappears as $m\to0$.
The instanton action that results from \cref{eq:sickinst,eq:zjt} can be written
\begin{equation}
    I_*(m) = - \frac{1}{2\pi^3} \left[ \theta\,\cos\theta - \left(1 - \frac{\theta^2}{2} \right)\,\sin\theta \right] - m \, e^{-S_0} \log(1 + \frac{e^{S_0+ \pi^2/\beta- \beta \left(\frac{\theta}{2\pi}\right)^2}}{2\sqrt{\pi} \, \beta^{3/2}} ),
\end{equation}
which leads to the following final result for the moments in JT gravity:
\begin{equation}
\label{eq:jtzm}
    \expval{ Z(\beta)^m}_* = \left( 1 + \frac{e^{S_0+ \pi^2/\beta- \beta \left(\frac{\theta}{2\pi}\right)^2}}{2\sqrt{\pi} \, \beta^{3/2}} \right)^m e^{
    \frac{e^{S_0}}{2\pi^3} \left[ \theta\,\cos\theta - \left(1 - \frac{\theta^2}{2} \right)\,\sin\theta \right]
    }.
\end{equation}
The term in the exponential above agrees with the leading $m=1$ contribution that \cite{Okuyama:2019xbv,Okuyama:2020ncd} finds in what they call the 't Hooft expansion of $\log\expval{Z(\beta)}$.
In the $m\to0$ replica limit, this JT instanton result precisely matches the Airy form of \cref{eq:aiq} at large $\beta$, as expected from \cref{ssec:quni}. As for the quenched entropy, the results for JT gravity exactly match \cref{eq:ais} in both extreme $\beta$ regimes. For integer $m\geq1$ moments, given the large-$\beta$ bracket from \cref{eq:largebjt}, it is of interest to specifically consider the $\beta\sim O(e^{2/3S_0})$ regime. The leading form of \cref{eq:jtzm} this gives is
\begin{equation}
\label{eq:airlim}
    \expval{ Z(\beta)^m}_* \approx \left(e^{\gamma^2/3} + \frac{m \, e^{-2 \gamma^2/3}}{4 \sqrt{\pi } \gamma}\right)^m, \qquad \beta\sim O(e^{2/3S_0}).
\end{equation}
where $\gamma$ was defined in \cref{eq:dnonp}. This controlled large-$\beta$ regime still allows for a successful implementation of the replica trick, i.e., \cref{eq:airlim} still yields the Airy result from \cref{eq:ais}. Intuitively, this is because $\beta\sim O(e^{2/3S_0})$ precisely corresponds to the transition between the disk and instanton dominating; had we trivialized the disk contribution, the replica trick would fail as discussed in \cref{ssec:lessons}. As for the annealed entropy, both \cref{eq:jtzm,eq:airlim} also yield consistent results which precisely match the Airy expressions we obtained in \cref{eq:tred}. The only differences with the Airy model occur at small $\beta$, a regime no longer dominated by the universality of edge statistics. This is illustrated in \cref{fig:snJT}.

\begin{figure}
    \centering
    \includegraphics[width=0.5\textwidth]{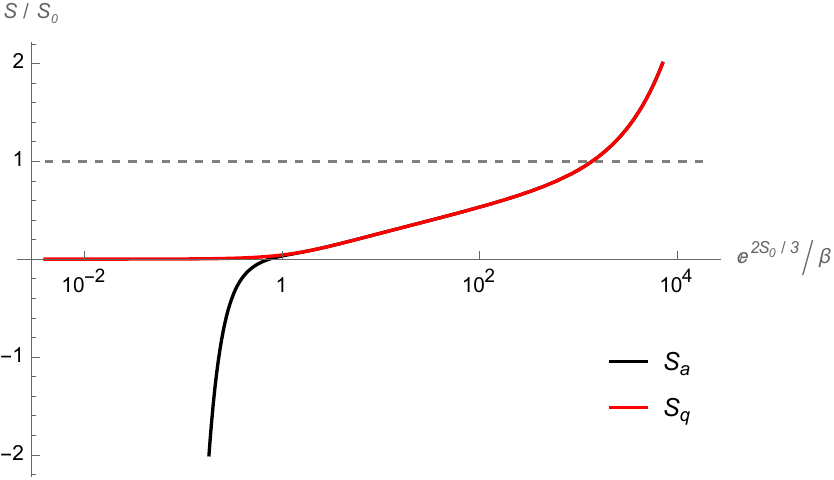}
    \caption{Comparison between annealed (black) and quenched (red) entropies for JT gravity using the instanton result from \cref{eq:jtzm}, and setting $e^{S_0}=10^6$. Both entropies match precisely the corresponding ones in the Airy model from \cref{fig:snairy} at large $\beta$, and only differ from those at small $\beta\sim O(1)$ when $S\sim O(S_0)$. At this point $S_a$ and $S_q$ are already indistinguishably uninteresting and dominated by bulk statistics.}
    \label{fig:snJT}
\end{figure}

As it turns out, $\beta\sim O(e^{2/3S_0})$ is also a controlled regime in the topological expansion of JT gravity. The first attempt at implementing the quenched replica trick in gravity made use of precisely this large-$\beta$ limit \cite{Engelhardt:2020qpv}. A difficulty \cite{Engelhardt:2020qpv} faced is that the growth in $m$ of these moments is $O(e^{m^3})$, i.e., much faster than exponential. As a result, an analytic continuation to $m\to0$ would be on general grounds highly non-unique. In our case, what makes the analytic continuation unique is the fact that \cref{eq:airlim} is a saddle-point result obtained by solving equations of motion which hold for all $m$ (cf. the strategy of \cite{Chandrasekaran:2022asa}). As a result, \cref{eq:airlim} actually applies automatically for any real $m\geq0$. Because \cref{eq:airlim} captures a transition to dominance of branes where the topological expansion should still make sense, we expect this result to provide a unique analytic continuation for the analysis of \cite{Engelhardt:2020qpv} which should succeed in implementing the replica trick.

To relate our results here to those of \cite{Blommaert:2019wfy,Okuyama:2021eju} one must treat branes dynamically rather than as probes. 
A plausible expectation is that the $\theta$-dependent terms in \cref{eq:jtzm} may be sourced by half-wormhole contributions to the gravitational path integral. Such effects are generally studied off-shell, which makes them practically inaccessible beyond $2$-dimensional models of quantum gravity like JT. However, our treatment here is strictly to leading order at large $e^{S_0}$ in terms of new saddle points that arise at large $\beta\sim O(e^{S_0})$. This strongly suggests that a semiclassical description of these half-wormholes may actually arise in the appropriate regime. In other words, this kind of non-perturbative effects would seem to admit an on-shell analysis which would easily generalize to higher-dimensional theories of gravity. A purely semiclassical calculation of these effects in quantum gravity will be pursued in \cite{intrepid}.

\section*{Acknowledgements}

It is a pleasure to thank
Netta Engelhardt,
Daniel Harlow,
Tom Hartman,
Gary Horowitz,
Luca Iliesiu,
Finn Larsen,
Samuel Leutheusser,
Henry Maxfield,
Marija Toma\v{s}evi\'{c},
Misha Usatyuk,
Nico Valdes-Meller,
Evita Verheijden, 
and Wayne Weng for stimulating discussions.
This research was supported by the Templeton Foundation via the Black Hole Initiative (award number 62286) and the MIT Center for Theoretical Physics.
The opinions expressed in this publication are those of the author and do not necessarily reflect the views of the John Templeton Foundation.

\appendix

\section{Near-Extremal Black Holes}
\label{sec:nebh}

This section gives a general overview of the current understanding of the gravitational entropy of near-extremal black holes, commenting on the results and expectations in semiclassical, perturbative, and non-perturbative temperature regimes relevant to this paper.

\subsection{Semiclassical Entropy}
\label{ssec:semiS}

The approach to black hole thermodynamics using the Euclidean gravitational path integral was pioneered at a semiclassical level by Gibbons-Hawking \cite{Gibbons:1976ue}, who using \cref{eq:slogZ} successfully reproduced the famous Bekenstein-Hawking entropy formula \cite{PhysRevD.7.2333,Hawking:1975vcx},
\begin{equation}
\label{eq:semibh}
    \mathcal{S}_{\smalltext{BH}}(\beta) \equiv \frac{A(\beta)}{4 G_{\smalltext{N}}},
\end{equation}
where $A(\beta)$ is the horizon area of the black hole at a given temperature.
Generically, this area remains non-zero even in the extremal limit $\beta\to\infty$, and grows linearly to leading order at low temperature. As a result, the general form of \cref{eq:semibh} at large $\beta$ is\footnote{\,Since we are interested in black holes with a nonzero area at zero temperature, we generally measure quantities relative to the extremal black hole radius and thus effectively treat them as dimensionless.}
\begin{equation}
\label{eq:sbhs0}
    \mathcal{S}_{\smalltext{BH}}(\beta) = S_0 + \frac{\#
}{\beta} + O(\beta^{-2}),
\end{equation}
where $S_0$ is often referred to as the extremal entropy.
Addressing the precise microscopic meaning of this entropy and relation or lack thereof to counting extremal states is one of the goals of this paper.
A standard statistical interpretation of $S_0$ as counting microstates would imply a degeneracy of the ground state given by $e^{S_0}$.
Since the extremal black hole area is typically macroscopically large in Planck units, $e^{S_0}$ would in particular be a very large number. Indeed, the result in \cref{eq:sbhs0} itself should be understood as an approximation $\mathcal{S}(\beta) \approx \mathcal{S}_{\smalltext{BH}}(\beta)$ capturing a large-$\beta$ expansion of just the classical, leading large-$S_0$ contribution to $\mathcal{Z}$ in \cref{eq:slogZ}.

The Strominger-Vafa enumeration of BPS states provided a celebrated confirmation of this result for a class of extremal black holes with supersymmetry in string theory \cite{Strominger:1996sh,David:2002wn}.
On symmetry grounds, such a large degeneracy is not surprising for black holes which preserve supersymmetry. However, that extremal black holes with no underlying symmetries could generally have such a high ground-state degeneracy of $O(e^{S_0})$ has historically been a puzzling possibility \cite{Preskill:1991tb,Teitelboim:1994az,Gibbons:1994ff,Hawking:1994ii,Hawking:1995fd,Maldacena:1998uz,Page:2000dk,Dabholkar:2006tb}.\footnote{\,This puzzle stems from conferring a microstate-counting meaning to the area of a black hole horizon, and is thus not dissimilar in spirit from modern issues interpreting this area as a von Neumann entropy in holography \cite{Engelhardt:2017aux,Engelhardt:2018kcs}.}

There is actually a simple argument for why one should be wary of the semiclassical analysis near extremality. Semiclassically, as the temperature of a black hole goes to zero, so does its heat capacity. Hence at low temperatures the emission of a typical Hawking quantum would drastically change the black hole temperature. But this would contradict the adiabaticity assumption of the semiclassical derivation of Hawking radiation \cite{Preskill:1991tb,Maldacena:1998uz,Iliesiu:2020qvm}. On dimensional grounds, this argument implies that the semiclassical result in \cref{eq:semibh} should only be trusted for $\beta\ll S_0$.

\subsection{Perturbative Corrections}
\label{ssec:quanS}

The near-extremal regime we are interested in corresponds to $\beta\gg S_0$, where the new large parameter $\beta$ implies that quantum effects that are subleading at large $S_0$ may no longer be negligible. 
Because \cref{eq:slogZ} gives the thermal entropy linearly in $\log\mathcal{Z}$, leading quantum corrections in a large-$S_0$ expansion are logarithmic in $1$-loop determinants. The computation of these has been tackled through a variety of methods for a plethora of black hole spacetimes, with and without supersymmetry, and at and near extremality \cite{Solodukhin:1994yz,Solodukhin:1994st,Sen:1995in,Mann:1997hm,Kaul:2000kf,Das:2001ic,Sen:2007qy,Banerjee:2010qc,Sen:2012dw,Sen:2012kpz,Sen:2012cj,Bhattacharyya:2012ye,Castro:2018hsc,Ghosh:2019rcj,Heydeman:2020hhw,Iliesiu:2022onk,Banerjee:2023quv,H:2023qko,Rakic:2023vhv,Kapec:2023ruw}. 
In general this turns out to be an extremely technical and laborious endeavor, with the tools at one's disposal strongly depending on the specific setting and on which contributions among those from matter fields, gravitons, or even strings, are being accounted for. Fortunately though, the functional form of these logarithmic corrections happens to be quite universal and straightforward to characterize.

With a single exception to be described shortly, all $1$-loop contributions near extremality yield logarithmic corrections to the extremal entropy $S_0$ of the form \cite{Solodukhin:1994yz,Solodukhin:1994st,Sen:1995in,Mann:1997hm,Kaul:2000kf,Das:2001ic,Sen:2007qy,Banerjee:2010qc,Sen:2012dw,Sen:2012kpz,Sen:2012cj,Bhattacharyya:2012ye,Castro:2018hsc}
\begin{equation}
\label{eq:slogc}
    \mathcal{S}_{\smalltext{log}} \equiv c_{\smalltext{log}} \log S_0.
\end{equation}
where $c_{\smalltext{log}}$ is an $O(1)$ factor that only depends on the number of massless fields (massive fields contribute at subleading orders).
This is the type of logarithmic correction generated by all non-zero modes (both bosonic and fermionic, including gravitons), and also every zero mode but one kind \cite{Banerjee:2010qc,Sen:2012dw,Sen:2012kpz,Sen:2012cj}.\footnote{\,Within the near-horizon region, zero modes correspond to diffeomorphisms (or their supersymmetric avatars) and gauge transformations which are large (i.e., non-trivial asymptotically), and thus physical. Except for scalars, a discrete family of these also arises for all types of fields and contribute to $c_{\smalltext{log}}$, once again including metric fluctuations.} The only exception is a single family of relevant metric tensor deformations which become strongly coupled at low temperatures.
These are exact zero modes of the extremal metric that have to be included, but the integral over which is formally divergent. Regulating their contribution requires flowing off the strict $\beta\to\infty$ limit, which gives as a result a temperature-dependent logarithmic correction of the form \cite{Ghosh:2019rcj,Heydeman:2020hhw,Iliesiu:2022onk,Banerjee:2023quv,H:2023qko,Rakic:2023vhv,Kapec:2023ruw,Pal:2023cgk}\footnote{\,As it turns out, these zero modes also give further contributions to the correction in \cref{eq:slogc} \cite{Banerjee:2010qc,Sen:2012dw,Sen:2012kpz,Sen:2012cj}.}
\begin{equation}
\label{eq:logbeta}
    \mathcal{S}_{\smalltext{Sch}}(\beta) \equiv -s \log \beta,
\end{equation}
where $s\geq0$ is a rational number which is found to only depend on supersymmetry considerations. In particular, $s=0$ for black holes which are BPS at extremality \cite{Heydeman:2020hhw}, but $s>0$ for any black hole which is not \cite{Iliesiu:2020qvm}. For instance, in purely bosonic theories one finds $s=3/2$. 

Altogether, combining the near-extremal results from \cref{eq:sbhs0,eq:slogc,eq:logbeta}, we obtain
\begin{equation}
\label{eq:genst}
    \mathcal{S}_{\smalltext{ext}}(\beta) = \underbrace{\mathcal{S}_{\smalltext{BH}}(\beta)}_\text{classical} + \underbrace{\mathcal{S}_{\smalltext{log}} + \mathcal{S}_{\smalltext{Sch}}(\beta)}_\text{$1$-loop} = S_0 + c_{\smalltext{log}} \log S_0 - s\log\beta + \frac{\#}{\beta},
\end{equation}
which gives the approximation $\mathcal{S}(\beta) \approx \mathcal{S}_{\smalltext{ext}}(\beta)$ at large $\beta$ of \cref{eq:slogZ} with $\mathcal{Z}$ capturing up to $1$-loop order in a large-$S_0$ expansion. In fact, as we will review in \cref{ssec:schW}, the result that gives rise to \cref{eq:logbeta} is $1$-loop exact, meaning that the $\beta$ dependence of this approximation should be reliable to all orders in perturbation theory at large $S_0$.
Non-perturbative effects of $O(e^{-S_0})$, about which the result in \cref{eq:genst} has nothing to say, will be addressed shortly.

Clearly, the corrections to the extremal entropy coming from \cref{eq:slogc} cannot possibly resolve the degeneracy puzzle raised above. Since these are temperature-independent, we will hereon absorb them in $S_0$ itself by a redefinition of its Bekenstein-Hawking value in \cref{eq:sbhs0} to $S_0 + \mathcal{S}_{\smalltext{log}}$. Hence, in the absence of symmetry principles, the only quantum correction which can possibly counter the large contribution from $S_0$ to the entropy near extremality must come from \cref{eq:logbeta}. In beautiful agreement with expectation, $\mathcal{S}_{\smalltext{Sch}}$ leaves unchanged the independently confirmed extremal entropy of BPS black holes ($s=0$), but provides the desired negative contribution to compete against $S_0$ for those which are not BPS at extremality ($s>0$).
We refer to near-extremal black holes which are BPS at extremality as near-BPS black holes.

For near-extremal black holes which are not near-BPS, remarkably not only does $\mathcal{S}_{\smalltext{Sch}}$ reduce the near-extremal entropy, but it also would make it negative at very low temperatures starting at $\beta \sim O(e^{S_0})$,\footnote{\,\label{fn:sutlelim}More precisely, this requires $\beta \sim O(e^{S_0/s})$, a distinction which will only be made when relevant.} a seemingly nonsensical result.
Different potential resolutions to this pathology have been suggested in the literature. Since for non-near-BPS black holes $\mathcal{Z}(\beta)\to0$ as $\beta\to\infty$, one option is that when $\beta\sim O(e^{S_0})$, some other object must dominate the partition function of the theory \cite{Turiaci:2023wrh,Rakic:2023vhv}.\footnote{\,It is straightforward to show that $S(\beta)\to-\infty$ as $\beta\to\infty$ indeed implies $\mathcal{Z}(\beta)\to0$ in the same limit. In the converse direction, the latter is just a necessary condition for the former. See \cref{sec:fact} for more details.} From the viewpoint of statistical physics, this is tantamount to saying that very near-extremal black holes which are not near-BPS do not exist. Alternatively, at such large $\beta\sim O(e^{S_0})$ one enters a regime where corrections which are non-perturbative in $S_0$ become important. These could potentially be able to tame the negativity induced by \cref{eq:logbeta} so as to restore a non-negative result for \cref{eq:slogZ}. Put differently, despite perturbative exactness, non-perturbative effects would constrain the range of validity of \cref{eq:genst} to $e^{S_0} \gg \beta \gg S_0$ \cite{Iliesiu:2020qvm,Iliesiu:2022onk,Rakic:2023vhv}, hence leaving room for non-perturbative effects of $O(e^{-S_0})$ to prevent a pathology in \cref{eq:slogZ}.\footnote{\,It has also been suggested that in fact one needs to account for doubly non-perturbative effects of $O(e^{-e^{S_0}})$ to restore the non-negativity of entropic quantities like this \cite{Johnson:2020mwi,Johnson:2021rsh}, a possibility which would seem reasonable from a random matrix formulation of the problem \cite{Johnson:2019eik,Johnson:2021zuo,Johnson:2022wsr}.}

However plausible, these ideas assume that the quantity $\mathcal{S}$ computed by \cref{eq:slogZ} in an effective theory of gravity ought to behave like a standard thermodynamic entropy and thus be non-negative. A central goal of this paper is to clarify that this assumption is incorrect: in effective quantum gravity $\mathcal{S}$ is in fact not the thermal entropy of any quantum system. Moreover, we show that $\mathcal{S}(\beta)$ must diverge negatively in the $\beta\to\infty$ limit for non-near-BPS black holes whenever the Euclidean gravitational path integral is used to compute $\mathcal{Z}(\beta)$. To identify what kind of quantity $\mathcal{S}$ is, we now turn to a brief review of the origin of the quantum effects that give rise to $\mathcal{S}_{\smalltext{Sch}}$ in \cref{eq:logbeta}.

\subsection{Non-perturbative Corrections}
\label{ssec:schW}

As mentioned above \cref{eq:logbeta}, the contribution $\mathcal{S}_{\smalltext{Sch}}$ to the near-extremal entropy comes from metric fluctuations of the black hole geometry which are nearly-zero modes near extremality and become exact zero modes in the extremal limit. These correspond to large diffeomorphisms supported on the throat of near-extremal black holes, a near-horizon region which is found to universally develop an AdS$_2$ factor with an enhanced $SL(2,\mathbb{R})$ symmetry in the extremal limit \cite{Kunduri:2007vf,Nayak:2018qej,Moitra:2018jqs}.
These strongly coupled gravitational dynamics in the throat are governed by a Schwarzian theory of reparameterization of the AdS$_2$ asymptotics or its supersymmetric generalizations \cite{Jensen:2016pah,Maldacena:2016upp,Engelsoy:2016xyb}, which can be argued to arise on general grounds from the emergent conformal symmetry at extremality and its breaking for $\beta<\infty$ \cite{Ghosh:2019rcj,Iliesiu:2020qvm,Heydeman:2020hhw,Iliesiu:2022onk,Boruch:2022tno,Pal:2023cgk}.
These Schwarzian theories turn out to be solvable quantum mechanical models with path integrals which are $1$-loop exact, and the computations of which have been carried in various contexts \cite{Bagrets:2016cdf,Cotler:2016fpe,Stanford:2017thb,Belokurov:2017eit,Mertens:2017mtv,Kitaev:2018wpr,Yang:2018gdb}.

Most often, a Schwarzian theory is arrived at by performing a dimensional reduction in the throat to the emergent AdS$_2$ factor.\footnote{\,
Even when the AdS$_2$ factor is warped in the full spacetime and this dimensional reductions may not seem as natural, a Schwarzian theory is still expected to govern the gravitational dynamics in the throat. For instance, for rotating black holes where the warping is caused by the breaking of spherical symmetry, this has been thoroughly studied in \cite{Maldacena:1997ih,Castro:2009jf,Castro:2018ffi,Moitra:2019bub}, and rigorously established for Kerr in \cite{Rakic:2023vhv}.}
The $2$-dimensional theory one obtains is JT gravity \cite{Jackiw:1984je,Teitelboim:1983ux} or a generalization thereof, where a dilaton field on the AdS$_2$ geometry captures fluctuations of the volume of the transverse dimensions. The non-trivial dynamics of this JT dilaton are well understood to reduce to an asymptotic boundary mode governed by a Schwarzian theory \cite{Almheiri:2014cka,Almheiri:2016fws}. As a result, JT gravity theories themselves are frequently taken as the starting point for the study of the low-energy dynamics of near-extremal black holes \cite{Almheiri:2014cka,Almheiri:2016fws,Kitaev:2018wpr,Yang:2018gdb,Sarosi:2017ykf,Nayak:2018qej,Moitra:2018jqs,Castro:2018ffi,Larsen:2018cts,Moitra:2019bub,Sachdev:2019bjn,Hong:2019tsx,Castro:2019crn,Charles:2019tiu,Mertens:2022irh}.

It is natural to treat the JT theory one obtains this way as a quantum gravitational theory in its own right. Formulated as a Euclidean gravitational path integral, such a theory involves integrating over all possible $2$-dimensional Riemannian manifolds with prescribed boundary conditions on AdS$_2$ boundaries.\footnote{\,The original Schwarzian theory would correspond to accounting for only the trivial disk topology.} The JT dilaton path integral imposes a constant negative curvature everywhere, thus making the gravity path integral run only over rigid hyperbolic spaces. These are Riemann surfaces which can be topologically classified, and whose moduli space within each topology class is finite-dimensional.
Hence the geometric part of the JT path integral is well-defined, and boils down to a sum over all topologies consistent with boundary conditions, weighted by the finite Weil–Petersson volumes of corresponding moduli spaces. In addition, each topology comes with a topological suppression in $O(e^{-S_0})$ by the Euler characteristic of the corresponding Riemann surface.
In the general settings of interest in JT gravity, one allows for spaces to have at least one asymptotic boundary. 
This way, on top of every geometric contribution, there also appear Schwarzian modes along every boundary governing the asymptotic dynamics of the JT dilaton. All such contributions are, once again, known exactly in $S_0$.
As worked out in \cite{Saad:2019lba}, the upshot is a JT gravitational path integral which is explicit at large $S_0$, 
perturbatively exact as a quantum expansion in $1/S_0$,
and non-perturbatively asymptotic as a topological expansion in powers of $e^{-S_0}$.

This remarkably powerful arena captures precisely the $O(e^{-S_0})$ effects that one needs to potentially probe non-perturbatively small temperatures with $\beta \sim O(e^{S_0})$.\footnote{\,\label{fn:fluct}Here we are referring to non-perturbative effects from topology change within the AdS$_2$ region of the original black hole spacetime near extremality. There may be non-perturbative contributions from the full-dimensional throat or even the black hole spacetime itself transitioning to higher topologies which we have no control over \cite{Iliesiu:2020qvm,Turiaci:2023wrh,Rakic:2023vhv}.} Since, the expansion at small $e^{-S_0}$ is an asymptotic series, the competition between the topological suppression and large $\beta \sim O(e^{S_0})$ has to be handled carefully. In particular, the asymptotic expansion parameter that becomes effective at large $\beta$ turns out to be the combination $\beta^s e^{-S_0}$, which ceases to be small when $\beta \sim O(e^{S_0/s})$. According to \cref{eq:genst}, this is actually also precisely the non-perturbatively scaling of $\beta$ for the onset of the negativity of $\mathcal{S}(\beta)$. A reliable analysis would thus seem impossible given that we begin to lose non-perturbative control as soon as we reach the regime of interest. In fact, the partial sums of the asymptotic series still provide tight enough bounds to show that the pathological negative behavior of $\mathcal{S}(\beta)$ persists non-perturbatively, as first shown by \cite{Engelhardt:2020qpv}. In other words, one reliably concludes that the quantity in \cref{eq:slogZ} with $\mathcal{Z}(\beta)$ computed by the JT gravitational path integral genuinely becomes negative at sufficiently low temperatures.
A central goal of this paper is to establish that this negativity occurs universally for non-near-BPS black holes and holds non-perturbatively in effective quantum gravity.
A key insight for this realization follows from the non-perturbative correspondence between JT gravity theories and random matrix ensembles \cite{Saad:2019lba}, which we make use of extensively.

\section{Proofs for Negative Entropies}
\label{sec:proof}

In this section we prove \cref{thm:fact}, as well as three additional requisite results.
The notation relevant to the results and proofs below is defined in \cref{sec:fact}.

We first establish the general region of convergence of $Z_g(z)\to0$ as $|z|\to0$ for a piecewise-continuous function $g$ with a well-defined Laplace transform $Z_g$ given by \cref{eq:lpgen}. This result is needed for understanding how a canonical partition function behaves at large $\beta$.

\begin{nlemma}
\label{lem:Ztozero}
    If $g$ is locally integrable and of exponential type such that $Z_g$ in \cref{eq:lpgen} exists, then $\displaystyle \lim_{\beta\to\infty} Z_g(\beta e^{i\alpha}) = 0$ for any $\alpha\in\left(-\frac{\pi}{2},\frac{\pi}{2}\right)$.
\end{nlemma}
\begin{proof}
    Assuming $Z_g(\beta e^{i\alpha})$ exists, it can be written
    $$
    Z_g(\beta e^{i\alpha}) = \lim_{\epsilon\to0^+} \int_{\epsilon}^\infty dE \, g(E) \, e^{-\beta e^{i\alpha} E}.
    $$
    By assumption of exponential type, $|g(E)|\leq O(e^{cE})$ for finite $c>0$ and sufficiently large $E$. Hence one can pick a finite $\beta^*>0$ satisfying $\beta^*\cos\alpha \geq c$ such that $Z_g(\beta e^{i\alpha})$ exists for all $\beta \cos\alpha$ with $\beta>\beta^*$. Additionally, by integrability of $g$, the function $f(E)\equiv |g(E)| e^{-\beta^* \cos\alpha \, E}$ is also integrable and dominates the integrand above, i.e., $|g(E)e^{-E e^{i\alpha}}| \leq f(E)$ for all $\beta\geq\beta^*$ and $E\geq\epsilon$. The dominated convergence theorem can thus be used to take the $\beta\to\infty$ limit inside the integral,
    \begin{equation}
    \label{eq:lim2}
        \lim_{\beta\to\infty} \lim_{\epsilon\to0^+} \int_{\epsilon}^\infty dE \, g(E) \, e^{-\beta e^{i\alpha} E} \leq \int_{0^+}^\infty dE \, g(E) \, \lim_{\beta\to\infty} e^{-\beta e^{i\alpha} E} = 0.
    \end{equation}
\end{proof}

Note that this result allows for $g(E)$ to diverge as $E\to0^+$. For instance, suppose $g(E) \supset E^p$. The Laplace transform of this term exists if and only if $p>-1$, and gives $p!/\beta^{p+1}$, which indeed tends to zero as $\beta\to\infty$ for all allowed $p>-1$. As another example, if $g(E)\supset -\log E$, the Laplace transform of this term gives $(\gamma + \log\beta)/\beta$, which again goes to zero in the limit.

The result that follows is needed to specifically address the behavior of the entropy at large $\beta$.

\begin{nlemma}
\label{lem:supperexp}
    Let $f(e^x) = \partial_x \log( -\log Z_g(e^x))$. If $\lim_{x\to\infty} f(e^x) = \lambda$, then for any small $\epsilon>0$, there exists $c_\pm,\beta^*> 0$ such that for all $\beta>\beta^*$, 
    $$
    e^{- c_- \beta^{{\lambda}-\epsilon}} > Z_g(\beta) > e^{- c_+ \beta^{{\lambda}+\epsilon}}.
    $$
\end{nlemma}
\begin{proof}
    If $\lim_{x\to\infty} f(e^x) = {\lambda}$, then for any $\epsilon>0$ there exists a large $0<x^*<\infty$ such that
    \begin{equation}
        {\lambda}-\epsilon < f(e^x) < {\lambda}+\epsilon, \quad \forall x>x^*.
    \end{equation}
    Under a semidefinite integral from $x^*$ to $x>x^*$, this gives
    \begin{equation}
        ({\lambda}-\epsilon)(x-x^*) < \log( \frac{\log Z_g(e^x)}{\log Z_g(e^{x^*})} ) < ({\lambda}+\epsilon)(x-x^*),
    \end{equation}
    Exponentiating twice and letting $\beta=e^x$, this leads to
    \begin{equation}
        \exp(\log Z_g(\beta^*) \; (\beta/\beta^*)^{{\lambda}-\epsilon}) > Z_g(\beta) > \exp(\log Z_g(\beta^*) \; (\beta/\beta^*)^{{\lambda}+\epsilon}),
    \end{equation}
    where we used that $\log Z_g(\beta)<0$ at large $\beta=\beta^*$, since $\lim_{\beta\to\infty} Z_g(\beta)=0$ by \cref{lem:Ztozero}. Letting $c_\pm \equiv -\frac{\log Z_g(\beta^*)}{{\beta^*}^{({\lambda}\pm\epsilon)}}$ and noting that $c_\pm>0$, we obtain the desired result.
\end{proof}

We now prove a mild generalization of the initial value theorem for Laplace transforms:
\begin{nlemma}
\label{nlem:inival}
    If $g$ is of exponential type and has Laplace transform $Z_g(\beta)$, then there is a constant $0<\nu\leq1$ such that $\lim_{E\to0^+} E^{1-\nu} g(E) = \hat{g}(0^+)<\infty$ exists and
    $$
        \lim_{\beta\to\infty} \beta^{\nu} Z_g(\beta) = \Gamma(\nu) \, \hat{g}(0^+).
    $$
    For $\nu=1$, this is known as the initial value theorem.
\end{nlemma}
\begin{proof}
    For $Z_g$ to exist, $g$ must be locally integrable, which requires $\lim_{E\to0^+}g(E)E=0$. Hence there must exist $0<\nu\leq1$ such that
    \begin{equation}
        g(E) = \frac{\hat{g}(E)}{E^{1-\nu}} \qquad\text{and}\quad \lim_{E\to0^+}\hat{g}(E)=g(0^+)<\infty.
    \end{equation}
    Using this and changing variables to $\xi\equiv\beta E$, the integral form of $Z_g(\beta)$ becomes
    \begin{equation}
        Z_g(\beta) = \beta^{-\nu} \int_0^\infty d\xi\, e^{-\xi} \; \frac{\hat{g}(\xi/\beta)}{\xi^{1-\nu}},
    \end{equation}
    Hence for $g$ of exponential type, using dominated convergence as in \cref{lem:Ztozero},
    \begin{equation}
        \lim_{\beta\to\infty} \beta^{\nu} Z(\beta) = \lim_{\beta\to\infty} \int_0^\infty d\xi\, e^{-\xi} \; \frac{\hat{g}(\xi/\beta)}{\xi^{1-\nu}} = \hat{g}(0^+) \int_0^\infty d\xi\; \frac{e^{-\xi}}{\xi^{1-\nu}} = \Gamma(\nu) \, \hat{g}(0^+).
    \end{equation}
\end{proof}

Finally, the main theorem and its proof go as follows:

\thermofact*

\begin{proof}
Applying \cref{lem:Ztozero} to $\expval{Z}=\hat{Z}+Z_g$ from \cref{eq:lpgen} and recalling $E_k>0$ gives
\begin{equation}
    \lim_{\beta\to\infty} \hat{Z}(\beta) = N_0,\qquad \lim_{\beta\to\infty} Z_g(\beta) = 0,
\end{equation}
and therefore
\begin{equation}
    \lim_{\beta\to\infty} Z(\beta) = N_0.
\end{equation}
Suppose first that $N_0>0$. Then starting from \cref{eq:annS} one easily obtains
\begin{equation}
    \lim_{\beta\to\infty} S_a(\beta) = \log N_0 - \frac{1}{N_0} \lim_{\beta\to\infty} \beta
    Z_g'(\beta).
\end{equation}
Note that $Z_g'(\beta)$ is equal to the Laplace transform of $-g(E) E$. Local integrability of $g(E)$ including $E=0$ requires $\lim_{E\to0^+} g(E) E = 0$. Hence by \cref{nlem:inival}, the Laplace transform of $-g(E) E$ satisfies $\lim_{\beta\to\infty} \beta Z_g'(\beta)=0$. As a result, for $N_0>0$,
\begin{equation}
    \lim_{\beta\to\infty} S_a(\beta) = \log N_0.
\end{equation}
This already shows that, for $N_0>0$ arbitrarily small, $S_a(\beta)$ can be made arbitrarily negative as $\beta\to\infty$. It thus seems reasonable to expect that $S_a(\beta)$ will diverge negatively as $\beta\to\infty$ for $N_0=0$, but taking the limit $N_0\to0$ above would be unjustified.

Suppose now that $N_0=0$, and rewrite \cref{eq:annS} as
\begin{equation}
\label{eq:sbdec}
    S_a(\beta) = \log Z(\beta) \left( 1 - f(\beta) \right), \qquad f(e^x) = \partial_x \log( -\log Z(e^x)).
\end{equation}
Then, if we define the limit
\begin{equation}
\label{eq:nlambda}
    \lambda \equiv \lim_{\beta\to\infty} f(\beta),
\end{equation}
a sufficient condition for $S(\beta)$ to diverge negatively as $\beta\to\infty$ is $\lambda<1$.
We now show that indeed only $\lambda<1$ is possible.
Using \cref{lem:supperexp}, the limiting behavior of $f$ in \cref{eq:nlambda} would imply there exist finite $c_\pm,\beta^*>0$ such that
\begin{equation}
\label{eq:sandw}
    e^{- c_- \beta^{{\lambda}-\epsilon}} > Z_g(\beta) > e^{- c_+ \beta^{{\lambda}+\epsilon}},
\end{equation}
for any $\beta>\beta^*$ and arbitrarily small $\epsilon>0$.
For a function $g$, \cref{lem:Ztozero} implies
\begin{equation}
    \lim_{\beta\to\infty} Z_g(\beta e^{i\alpha}) = 0,
\end{equation}
for any $\alpha\in\left(-\frac{\pi}{2},\frac{\pi}{2}\right)$, which on the asymptotic form for $Z$ above this imposes $\lambda \leq 1$.\footnote{\,Intuitively, as an integral over exponentials $e^{-\beta E}$ against a non-negative measure $g(E)$, the Laplace transform cannot possibly decrease faster than exponential at large $\beta$.}
Hence it only remains to show that $\lambda=1$ cannot happen.

Because $N_0=0$, the distribution $\Delta$ in \cref{eq:rdc} has no support near $E=0$, so the canonical shift that defines $\rho$ in \cref{eq:ltdef} implies $\inf\supp g = 0$.
Recall $\lim_{E\to0^+} g(E) \equiv g(0^+) \in [0,\infty]$, with divergence allowed so long as $g$ remains integrable. Suppose $g(0^+)>0$.
By piecewise continuity, there is a sufficiently small $\epsilon>0$ such that $(0,\epsilon)\subset\supp g$ and $g_0 \equiv \min g(0,\epsilon) > 0$. Hence,
\begin{equation}
    Z_g(\beta) \geq \int_{0^+}^\epsilon dE \, g(E)\, e^{-\beta E} \geq \frac{g_0}{\beta} \left( 1-e^{-\beta  \epsilon} \right).
\end{equation}
This lower bound on $Z_g(\beta)$ is only consistent with the upper bound in \cref{eq:sandw} for $\lambda=0$.
Finally, for the sake of contradiction, suppose $\lambda=1$. By \cref{lem:supperexp}, for this to hold in the $\beta\to\infty$ limit, it must be the case that $\lim_{\beta\to\infty} e^{c\beta} Z_g(\beta)$  is finite for some $c>0$. However, this would require $g$ to have no support on $(0,c)$, for otherwise the limit would diverge. This would contradict the construction, which fixes $\inf\supp g = 0$.

Altogether, we have shown that $\lim_{\beta\to\infty} S_a(\beta)=\log N_0<0$ for $N_0>0$, and that for $N_0=0$ the limit diverges negatively. If $N_0\geq0$, then $S_a(\beta)\geq0$ for all $\beta$ by continuity and monotonicity of the Laplace transform. That for $N_0<1$ including $N_0=0$ there exists some finite $\beta^*>0$ such that $S_a(\beta)<0$ for all $\beta>\beta^*$ follows for the same reason.

\end{proof}

For the $N_0=0$ case, note that one could consider writing\footnote{\,
These two terms give the standard thermodynamic relation $S(\beta) = \beta (\expval{E}-F(\beta))$. Hence, thermodynamically the negativity of $S(\beta)$ basically means that at some temperature there is more free energy than energy itself. This would obviously be non-sense if $\rho$ were the spectral density of a single physical system.}
\begin{equation}
    S_a(\beta) = \log \expval{Z(\beta)} - \frac{\beta \expval{Z'(\beta)}}{\expval{Z(\beta)}}.
\end{equation}
As $\beta\to\infty$, the first term clearly diverges negatively.
However, $Z'(\beta)<0$ for any finite $\beta$, so in principle the second term could cancel the divergence of the first.
For this second term we know $\lim_{\beta\to\infty} Z(\beta)=0$ and also $\lim_{\beta\to\infty} \beta Z'(\beta)=0$. In particular, the latter follows from $\expval{Z}=\hat{Z}+Z_g$ by noting that all terms in $\hat{Z}'(\beta)$ decrease exponentially in $\beta$, and using \cref{nlem:inival} for $Z_g'(\beta)$ as the Laplace transform of $-Eg(E)$.
Then one would hope to apply L'H\^opital, but this turns out to be useless: both numerator and denominator continue to give zeroes at arbitrary derivative order.
Not only is the behavior of this term in the limit not obvious, but it can also be inconclusive. Namely, while in many cases one can show this ratio is finite in the limit, there are examples for which it is divergent and could indeed compete with the first term. These correspond to cases in which the parameter in \cref{eq:nlambda} satisfies $0<\lambda<1$, which occurs when $\rho$ and all of its derivatives vanish at zero (cf. $1/\rho(E)$ having an essential singularity at $E=0$). For instance, $\rho(E) \sim e^{-1/E^{s}}$ gives $\lambda = 1/(1+s^{-1})$ for $s\geq0$; functions with non-trivial derivatives at zero give $\lambda=0$.

\section{Supersymmetric Random Matrix Theory}
\label{sec:susyrmt}

This section discusses the relevant details for the construction of supersymmetric Hamiltonians that justifies our general use of \cref{eq:hcool} throughout this paper. Since the implementation of higher supersymmetry in random matrices becomes increasingly intricate, let us begin with the smallest amount of supersymmetry.

\subsection{\texorpdfstring{$\mathcal{N}=1$}{N=1}}

Random matrix theory with $\mathcal{N}=1$ supersymmetry is discussed in e.g. \cite{Fu:2016vas,Li:2017hdt,Kanazawa:2017dpd,Sun:2019yqp,Stanford:2019vob}. Its simple structure involves a single self-adjoint supercharge $Q$ and a Hamiltonian given by $H=Q^2$. The canonical construction of $Q$ out of the random matrix $M$ takes the form\footnote{\,\label{fn:n1susy}Alternatively, one can write the real supercharge $Q = \widetilde{Q} + \widetilde{Q}^\dagger$ in terms of the associated complex supercharges
$$
    \widetilde{Q} = 
    \begin{pmatrix}
         0 & M \\
         0 & 0
    \end{pmatrix},\qquad   
    \widetilde{Q}^\dagger = 
    \begin{pmatrix}
         0 & 0 \\
         M^\dagger & 0
    \end{pmatrix},
$$
which are nilpotent, $\widetilde{Q}^2 = {\widetilde{Q}^\dagger}{}^2 =0 $, and give the Hamiltonian in the anti-commutator form $H = \frac{1}{2} \{\widetilde{Q},\widetilde{Q}^\dagger\}$.}
\begin{equation}
\label{eq:superq}
    Q =
    \begin{pmatrix}
         0 & M \\
         M^\dagger & 0
    \end{pmatrix},
\end{equation}
in terms of which the $\mathcal{N}=1$ random Hamiltonian $H=Q^2$ is identically given by \cref{eq:hcool}.
The symmetry group $G(N+\bar{\upnu})\times G(N)$ naturally realizes the structure of a $\mathbb{Z}_2$-graded algebra, where the sectors transforming under $G(N+\bar{\upnu})$ and $G(N)$ may respectively be identified as bosonic and fermionic. The grading is implemented in terms of the usual fermion number operator $\mathsf{F}$ by
\begin{equation}
    (-1)^{\mathsf{F}} \equiv
    \begin{pmatrix}
         \mathds{1}_{N+\bar{\upnu}} & 0 \\
         0 & -\mathds{1}_{\bar{\upnu}}
    \end{pmatrix},
\end{equation}
which clearly satisfies the desired relations $[H,(-1)^{\mathsf{F}}]=0$ and $\{Q,(-1)^{\mathsf{F}}\}=0$. Because there are $N+\bar{\upnu}$ bosonic and $N$ fermionic states, the supersymmetric index is $\bar{\upnu}$. In addition, because $(-1)^{\mathsf{F}}$ commutes with $H$, we have the Witten index 
\begin{equation}
    \Tr ~(-1)^{\mathsf{F}} e^{-\beta H} = \bar{\upnu},
\end{equation}
at any inverse temperature $\beta$. Hence $\bar{\upnu}=0$ is associated to the breaking of all supersymmetry, whereas $\bar{\upnu}>0$ corresponds to having $\bar{\upnu}$ states that preserve supersymmetry. That these BPS states precisely capture the degenerate ground state of $H$ is easily seen as follows.

If $M$ is a square AZ matrix of size $N$, then $H$ is a square matrix of size $2N$ whose only $N$ distinct eigenvalues are those of $M$ squared. A global $2$-fold degeneracy is in this case expected given that $M^\dagger M$ and $M M^\dagger$ have identical eigenvalues. Besides these, the resulting Hamiltonian $H$ will generically exhibit no other degeneracies and will, in particular, be non-singular.
If instead $M$ is an $(N+\bar{\upnu})\times N$ rectangular BPS matrix, then $H$ still is a square matrix, but now of size $2N+\bar{\upnu}$. The Hamiltonian $H$ is then singular with at least $\bar{\upnu}>0$ zero eigenvalues. This means that generically every $H$ drawn from this ensemble deterministically has $\bar{\upnu}$ degenerate ground states all with the same zero energy, corresponding to the BPS states mentioned above.

\subsection{\texorpdfstring{$\mathcal{N}=2$}{N=2}}

Although $\mathcal{N}=1$ random matrices can clearly accommodate BPS states, recall that the $\mathcal{N}=1$ super-Schwarzian spectrum from \cref{eq:susynon} has none. It is thus of interest to also understand the $\mathcal{N}=2$ matrix construction for the supermultiplet spectra of the $\mathcal{N}=2$ super-Schwarzian theories in \cref{eq:susy2}, which do exhibit BPS states. The presentation here will be concise, and we refer the reader to \cite{Turiaci:2023jfa} for more details on the relevant structure of $\mathcal{N}=2$ supersymmetry.

The $\mathcal{N}=2$ superalgebra has a Hermitian conjugate pair of nilpotent generators, the supercharges $Q$ and $Q^\dagger$, and a Hamiltonian $H= \{Q,Q^\dagger\}$ they commute with (cf. \cref{fn:n1susy} for $\mathcal{N}=1$). The algebra also has an R-symmetry given by a $U(1)$ outer automorphism group.
There is no unique $\mathcal{N}=2$ theory, but a two-parameter family thereof.
A general $\mathcal{N}=2$ theory can be specified by an odd integer $\hat{q}>0$ and, to account for a possible anomaly, by a constant $\delta\in[0,1)$.
The R-charges of $Q$ and $Q^\dagger$ are respectively $\pm\hat{q}+\delta$,\footnote{\,Exchanging $Q \leftrightarrow Q^\dagger$ one can always make $Q$ have positive R-charge, which is why one can assume $\hat{q}>0$.} and states have R-charge $k\in \mathbb{Z}+\delta$.
The Hilbert space is a direct sum $\mathcal{H}=\bigoplus_k\mathcal{H}_k$, where $\mathcal{H}_k$ consists of all states $\psi_k$ with R-charge $k$. Correspondingly, the supercharge $Q = \sum_k Q_k$ with the restriction $Q_k : \mathcal{H}_k \to \mathcal{H}_{k+\hat{q}}$, and similarly for $Q^\dagger$ with $Q_{k}^\dagger : \mathcal{H}_{k+\hat{q}} \to \mathcal{H}_{k}$.
The supersymmetry condition $Q^2={Q^\dagger}^2=0$ then implies $Q_{k+\hat{q}} Q_k=Q_{k}^\dagger Q_{k+\hat{q}}^\dagger=0$ for all $k$.

The algebra decomposes into two types of irreducible supermultiplets. Any state $\psi_k$ satisfying $Q\psi_k=Q^\dagger\psi_k=0$ forms a singlet by itself; these are BPS states. Any pair of states $(\psi_k,\psi_{k+\hat{q}})$ obeying $Q\psi_k = \lambda \psi_{k+\hat{q}}$ and $Q^\dagger \psi_{k+\hat{q}} = \lambda^* \psi_{k}$ with $\lambda,\lambda^*\neq0$ forms a doublet; these are non-BPS states.
All BPS states have zero energy since $H\psi_k=0$, whereas non-BPS states come in degenerate pairs with energy $H\psi_k = H\psi_{k+\hat{q}} = \lambda^*\lambda > 0$. Because doublets always involve consecutive R-charge values $(k,k+\hat{q})$, a convenient parameter turns out to be their average R-charge $q_k\equiv k + \frac{\hat{q}}{2}$. In terms of $q_k$, a $(k,k+\hat{q})$ doublet involves R-charges $q_k \mp \frac{\hat{q}}{2}$. All irreducible supermultiplets contain disjoint sets of states, so we can decompose $\mathcal{H}_k$ into
\begin{equation}
    \mathcal{H}_k = \mathcal{H}_k^- \oplus \mathcal{H}_k^0 \oplus \mathcal{H}_k^+,
\end{equation}
where $\mathcal{H}_k^0$ consists of all BPS states in $\mathcal{H}_k$, $\mathcal{H}_k^-$ of all non-BPS states in $\mathcal{H}_k$ forming $(k-\hat{q},k)$ doublets with states in $\mathcal{H}_{k-\hat{q}}$, and $\mathcal{H}_k^+$ of all non-BPS states in $\mathcal{H}_k$ forming $(k,k+\hat{q})$ doublets with states in $\mathcal{H}_{k+\hat{q}}$. By definition, a state is in $\mathcal{H}_k^+$ if and only if it has a doublet pair in $\mathcal{H}_{k+\hat{q}}^-$, so clearly $N_k^{\pm}\equiv\dim \mathcal{H}_k^\pm$ satisfy $N_k^{+}=N_{k+\hat{q}}^{-}$. 
This leads to the following reorganization of the Hilbert space:
\begin{equation}
    \mathcal{H} = \bigoplus_k \, \mathcal{H}_k^0 \oplus \mathcal{H}^\leftrightarrow_{q_k}, \qquad \mathcal{H}^\leftrightarrow_{q_k} \equiv \mathcal{H}_k^+ \oplus \mathcal{H}_{k+\hat{q}}^-.
\end{equation}
Here we have combined all non-BPS states forming $(k,k+\hat{q})$ doublets into $\mathcal{H}^\leftrightarrow_{q_k}$, which itself forms a (reducible) supermultiplet that we will refer to as a $q_k$-multiplet. Since $\mathcal{H}^\leftrightarrow_{q_k}$ contains an even number of states, we let $2N^\leftrightarrow_{q} \equiv \dim \mathcal{H}^\leftrightarrow_{q}$ and note that $N^\leftrightarrow_{q_k} = N_k^{+}=N_{k+\hat{q}}^{-}$.
Dropping the $k$ from $q_k$ we may sometimes use $q$ as a standalone parameter. The definition of $q_k$ makes $q\in\mathbb{Z}+\delta-\frac{1}{2}$, since $\hat{q}$ is always odd. 
Each of the Hilbert subspaces $\mathcal{H}_k^0$ and $\mathcal{H}_{q}^{\leftrightarrow}$ defines a supermultiplet which turns out to be statistically independent of the rest. In particular, there is a fixed number of zero-energy BPS states in each $\mathcal{H}_k^0$, and there is an independent probability distribution governing the spectrum of energies of non-BPS states in each $\mathcal{H}_{q}^{\leftrightarrow}$.
Interestingly though, the spectral statistics of a $q$-multiplet do happen to be sensitive to the total number of BPS and certain non-BPS states in $\mathcal{H}^0_{q \mp \hat{q}/2}$.\footnote{\,\label{fn:bpsindep}That the spectra of non-BPS states in adjacent spaces $\mathcal{H}_{q}^{\leftrightarrow}$ and $\mathcal{H}_{q+\hat{q}}^{\leftrightarrow}$ may depend on shared numerical parameters does not correlate them statistically. Each spectrum still follows an independent probability distribution, regardless of whether these distributions happen to have parameters in common (cf. iid random variables).} Because these supermultiplets have statistically independent spectra, they are the right target for a matrix ensemble description.
Indeed, what \cref{eq:susy2} describes is the spectral density of a $q$-multiplet, including a potential contribution from BPS states in $\mathcal{H}^0_{q\pm\hat{q}/2}$.\footnote{\,Note however that generically at least one of $\mathcal{H}^0_{q\pm\hat{q}/2}$ is expected to be empty, for otherwise BPS states from each could be perturbed into non-BPS states forming doublets in $\mathcal{H}_{q}^{\leftrightarrow}$ \cite{Turiaci:2023jfa}.}

To understand these spectra consider first expanding the Hamiltonian as
\begin{equation}
\label{eq:hqk}
    H = \sum_k Q_k Q_k^\dagger + Q_k^\dagger Q_k.
\end{equation}
We want to isolate the operators that act on and thus determine the energies of each $q_k$-multiplet. Since states in $\mathcal{H}_{q_k}$ only have R-charge $k$ or $k+\hat{q}$, let $\psi_k$ and $\psi_{k+\hat{q}}$ respectively denote any two of them. Those $\psi_k\in\mathcal{H}_k$ can be acted upon by $Q_k$ and $Q_{k-\hat{q}}^\dagger$, whereas those $\psi_{k+\hat{q}}\in\mathcal{H}_{k+\hat{q}}$ can be acted upon by $Q_{k+\hat{q}}$ and $Q_{k}^\dagger$. The BPS states are annihilated by any supercharge operator that acts on them, whereas the non-BPS states all come in doublets satisfying $Q_k\psi_k = \lambda\psi_{k+\hat{q}}$ and $Q_k^\dagger \psi_{k+\hat{q}} = \lambda^* \psi_{k}$. But by the supersymmetry conditions $Q_{k+\hat{q}} Q_k = Q_{k-\hat{q}}^\dagger Q_k^\dagger = 0$ the latter are also annihilated as $Q_{k+\hat{q}} \psi_{k+\hat{q}} = Q_{k-\hat{q}}^\dagger \psi_k = 0$. Hence in fact the only supercharge operators that act non-trivially on $\mathcal{H}_{q_k}$ are $Q_k$ and $Q_{k}^\dagger$. This means one can single out from \cref{eq:hqk} the term
\begin{equation}
\label{eq:hkplet}
    H_{q_k} \equiv  Q_k Q_k^\dagger + Q_k^\dagger Q_k = 
    \begin{pmatrix}
         (QQ^\dagger)_{q+\hat{q}/2} & 0 \\
         0 & (Q^\dagger Q)_{q-\hat{q}/2}
    \end{pmatrix}.
\end{equation}
as it solely determines the spectrum of the $q_k$-multiplet. In the second expression we have simply made use of the fact that $H_{q}$ acts just on the direct sum $\mathcal{H}_{q-\hat{q}/2}\oplus \mathcal{H}_{q+\hat{q}/2}$ to make \cref{eq:hkplet} suggestively similar to \cref{eq:hcool}.

Since $H_{q_k}$ only depends on $Q_k$ and its Hermitian conjugate, one has to understand the structure of $Q_k$ operators.
Recall each of these is a linear map $Q_k : \mathcal{H}_k \to \mathcal{H}_{k+\hat{q}}$, so letting $N_k \equiv \dim \mathcal{H}_k$ one can represent $Q_k$ as an $N_{k+\hat{q}} \times N_k$ matrix.
Comparing \cref{eq:hkplet} to the $\mathcal{N}=1$ matrix Hamiltonian from \cref{eq:hcool}, one would like to identify $Q_k$ with an $N_{k+\hat{q}} \times N_k$ random matrix $M$ in an AZ or BPS ensemble. Unfortunately this is not immediately possible given the supersymmetry constraints $Q_{k}Q_{k-\hat{q}}=0$ that these must satisfy. Nonetheless, by analyzing the interplay between the Jacobian measure of a matrix integral over $Q_k$ and a functional constraint imposing $Q_{k}Q_{k-\hat{q}}=0$, \cite{Turiaci:2023jfa} showed that $Q_{k}$ continues to be describable by an independent AZ or BPS random matrix on the orthogonal complement of the image of $Q_{k-\hat{q}}$. However, their focus on short sequences of R-charge values suggested that the effective random sub-matrix of $Q_k$ would only be rectangular when $\mathcal{H}_k\oplus\mathcal{H}_{k+\hat{q}}$ contained BPS states, and would be a square matrix otherwise. As we discuss in \cref{sec:gaps}, this would explain the formation of spectral gaps on $q$-multiplets when $\mathcal{H}_{q\pm\hat{q}/2}$ contains BPS states, but would not explain why the onset of the spectrum keeps shifting to higher energies when there no longer are any BPS states nearby. In what follows we depart from the analysis of \cite{Turiaci:2023jfa} to address this issue.

The strategy will be to analyze an ordered sequence of R-charge values, providing a random matrix representation for each $Q_k$ that consistently implements the supersymmetry conditions $Q^2=0$. As we will see, the random matrix description that arises for the spectrum of an arbitrary $q$-multiplet becomes manifestly intrinsic to the structure of the Hilbert spaces $\mathcal{H}_{q\pm\hat{q}/2}$ relevant to the supermultiplet, irrespective of the order followed in the construction.

Suppose a sequence of R-charge values begins at $k$, i.e., $\mathcal{H}_k\neq\varnothing$ but $Q_{k-\hat{q}} = 0$. In general $\mathcal{H}_k$ may contain BPS states in $\mathcal{H}_k^0$ and non-BPS states in $\mathcal{H}_k^+$, but clearly $\mathcal{H}_k^- =\varnothing$, so $N_k^-=0$.
The first non-trivial supercharge operator $Q_k : \mathcal{H}_k \to \mathcal{H}_{k+\hat{q}}$ may be represented by an $N_{k+\hat{q}} \times N_k$ matrix. Since $Q_{k-\hat{q}}$ is trivial, there are actually no constraints on $Q_k$ from $Q_k Q_{k-\hat{q}} = 0$. Hence $Q_k$ should be identified with a random AZ or BPS matrix of size $N_{k+\hat{q}} \times N_k$. Interestingly, the cases $N_{k+\hat{q}} > N_k$ and $N_{k+\hat{q}} \leq N_k$ turn out to be qualitatively distinct.

Suppose first that $N_{k+\hat{q}} \leq N_k$. Then generically $\rank Q_k = N_{k+\hat{q}}$ and $\ker Q_k = N_k - N_{k+\hat{q}}$. Because $Q_{k-\hat{q}}^\dagger=0$, any state in $\mathcal{H}_k$ annihilated by $Q_k$ is automatically a BPS state, giving
\begin{equation}
    N_k^0 = N_k - N_{k+\hat{q}},
\end{equation}
which could also be just zero.
Since $N_k^-=0$, the remaining states in $\mathcal{H}_k$ must be forming doublets with states in $\mathcal{H}_{k+\hat{q}}$, so $N_{k}^+ = N_{k+\hat{q}}$. But by the doublet relation between $\mathcal{H}_k^+$ and $\mathcal{H}_{k+\hat{q}}^-$ we know $N_{k}^+ = N_{k+\hat{q}}^- = N_{q_k}^{\leftrightarrow}$, which in particular implies $N_{k+\hat{q}} = N_{q_k}^{\leftrightarrow}$. In other words, $\mathcal{H}_{k+\hat{q}} = \mathcal{H}_{q_k}^\leftrightarrow$ solely consists of non-BPS states in $(k,k+\hat{q})$ doublets. As a result, $\mathcal{H}_{k+\hat{q}}$ contains neither BPS states nor non-BPS states that could possibly continue the R-charge sequence beyond $k+\hat{q}$. Having terminated with $H = \mathcal{H}_{k}\oplus \mathcal{H}_{k+\hat{q}}$, we simply have supercharge $Q=Q_k$ and Hamiltonian $H=H_{q_k}$ as in \cref{eq:hkplet}. This simple case reproduces the $\mathcal{N}=1$ setting with a random AZ or BPS matrix $M$ of dimensions $(N+\bar{\upnu})\times N$ having $N=N_{q_k}^{\leftrightarrow}$, $\bar{\upnu} = N_k^0$, and random supercharge $Q=M^\dagger$.

Suppose now that $N_{k+\hat{q}} > N_k$. Then generically $\rank Q_k = N_{k}$ and the kernel of $Q_k$ in $\mathcal{H}_k$ is trivial. Hence there are no BPS states in $\mathcal{H}_k$ and since both $N_k^-=N_k^0=0$ we have $N_k = N_{q_k}^\leftrightarrow$, i.e., every state in $\mathcal{H}_k$ is in a non-BPS doublet with some state in $\mathcal{H}_{k+\hat{q}}$. While the content of $\mathcal{H}_{k+\hat{q}}$ remains to be determined, let us emphasize that at this point $Q_k$ already is rectangular even though there are no BPS states in $\mathcal{H}_k$. The random $(N+\bar{\upnu})\times N$ matrix $M$ here has $N=N_{q_k}^{\leftrightarrow}$, $\bar{\upnu} = N_{k+\hat{q}} - N_{q_k}^{\leftrightarrow}$, and is identified with $Q_k = M$. The corresponding $q_k$-multiplet spectrum is described by the Hamiltonian $H_{q_k}$ given by the right-most expression in \cref{eq:hkplet} with $Q=M$. Crucially, even though the matrix $M M^\dagger$ is singular with $\bar{\upnu}$ degenerate zero eigenvalues, these now have nothing to do with BPS states. In fact, this $\bar{\upnu}$-dimensional kernel of $(M M^\dagger)_{k+\hat{q}}$ simply corresponds to all those states in $\mathcal{H}_{k+\hat{q}}$ which are not in $\mathcal{H}_{q_k}^\leftrightarrow$. Namely, this term corresponds to $Q_k Q_k^\dagger$, which annihilates any state in $\mathcal{H}_{k+\hat{q}}$ that is not in a $(k,k+\hat{q})$ doublet, whether or not it is BPS. We thus learn that the statistics of this purely non-BPS $q_k$-multiplet do actually require a rectangular random matrix with $\bar{\upnu}>0$ from a BPS ensemble for their description. In addition, the spectrum of non-BPS states will behave just like that of a supermultiplet with $\bar{\upnu}$ BPS states with zero energy, even though in this case there actually are no such states at zero.

Continuing the sequence to $\mathcal{H}_{k+\hat{q}}$, we know $\dim H_{k+\hat{q}}^- = N_{q_k}^{\leftrightarrow}$, and that the remaining $N_{k+\hat{q}}-N_k$ states in $\mathcal{H}_{k+\hat{q}}$ must be either BPS or non-BPS in doublets with states in $\mathcal{H}_{k+2\hat{q}}$. To address these, we must now study $Q_{k+\hat{q}} : \mathcal{H}_{k+\hat{q}} \to \mathcal{H}_{k+2\hat{q}}$. Although $Q_{k+\hat{q}}$ can be represented by an $N_{k+2\hat{q}} \times N_{k+\hat{q}}$ matrix as usual, we now have to take into account the supersymmetry constraint $Q_{k+\hat{q}} Q_k = 0$. Because $\rank Q_k = N_{q_k}^{\leftrightarrow}$, this projects out a subspace of codimension $N_{q_k}^{\leftrightarrow}$ in $\mathcal{H}_{k+\hat{q}}$ from the domain of $Q_{k+\hat{q}}$. In other words, $Q_{k+\hat{q}}$ effectively becomes representable by a sub-matrix of dimensions $N_{k+2\hat{q}} \times (N_{k+\hat{q}} - N_{q_k}^{\leftrightarrow})$. This reduced matrix is actually now unconstrained and thus identifiable with a random AZ or BPS matrix. The situation here essentially reduces to our starting point at the beginning of the sequence on $\mathcal{H}_k$. More explicitly, note that the subspace of $\mathcal{H}_{k+\hat{q}}$ that $Q_{k+\hat{q}} Q_k = 0$ removes corresponds precisely to the non-BPS states in the $q_k$-multiplet $\mathcal{H}_{q_k}^{\leftrightarrow}$. Hence the reduced random matrix for $Q_{k+\hat{q}}$ is representing a map acting only on $\mathcal{H}_{k+\hat{q}}^0\oplus\mathcal{H}_{k+\hat{q}}^+\subset \mathcal{H}_{k+\hat{q}}$. This is just like the situation we encountered when starting the sequence at R-charge $k$, where obviously $\mathcal{H}_k^-=\varnothing$ and we had $\mathcal{H}_k = \mathcal{H}_k^0 \oplus \mathcal{H}_k^+$.
Hence at this point one can simply iterate the logic we followed thus far, but starting from $\mathcal{H}_{k+\hat{q}}$ and an operator $Q_{k+\hat{q}}$ restricted to an effective domain of codimension $N_{q_k}^{\leftrightarrow}$.

Above we chose to analyze a sequence of R-charges starting from some smallest value $k$, iteratively imposing supersymmetry constraints of the form $Q_{k+\hat{q}} Q_k=0$. However, this order was arbitrary and we could have equally started from some largest value $k+\hat{q}$, and iteratively imposed supersymmetry constraints of the form $ Q_{k-\hat{q}}^\dagger Q_k^\dagger = 0$. In fact, the pattern that arises for an arbitrary $q_k$-multiplet is clear, and can be easily described solely in terms of the number of BPS states and structure of non-BPS doublets in $\mathcal{H}_k\oplus \mathcal{H}_{k+\hat{q}}$. Explicitly, the spectrum of the $q_k$-multiplet is captured by $Q_{k} : \mathcal{H}_{k} \to \mathcal{H}_{k+\hat{q}}$, which before imposing any supersymmetry conditions $Q^2=0$ is represented by an $N_{k+\hat{q}}\times N_k$ matrix with $N_k=N_k^- + N_k^0 + N_k^+$ and $N_{k+\hat{q}}=N_{k+\hat{q}}^- + N_{k+\hat{q}}^0 + N_{k+\hat{q}}^+$, where $N_k^+=N_{k+\hat{q}}^-\equiv N_{q_k}^\leftrightarrow $ is the number of non-BPS forming $(k,k+\hat{q})$ doublets, and at most one of $N_k^0$ or $N_{k+\hat{q}}^0$ is non-zero and counts the number of BPS states of the corresponding R-charge. The supersymmetry conditions $Q_{k+\hat{q}}Q_k = Q_{k}Q_{k-\hat{q}} = 0$ restrict the rank of $Q_k$ to $N_{q_k}^\leftrightarrow$, and thus its resulting AZ or BPS random matrix representation has $N=N_{q_k}^\leftrightarrow$. As for the $\bar{\upnu}$ parameter, if there are BPS states in $\mathcal{H}_k\oplus\mathcal{H}_{k+\hat{q}}$ then $\bar{\upnu} = \max\{N_k^0,N_{k+\hat{q}}^0\}$, whereas if there are no BPS states $N_k^0=N_{k+\hat{q}}=0$ and then $\bar{\upnu}=\min\{N_k^-,N_{k+\hat{q}}^+\}=\min\{N_{q_k-\hat{q}}^\leftrightarrow,N_{q_k+\hat{q}}^\leftrightarrow\}$.\footnote{\,The transition between these two cases is consistent. Namely, as we found, if $N_{k+\hat{q}}^0>0$ then necessarily $N_{k+\hat{q}}^+=0$ because the R-charge sequence is upper bounded by $k+\hat{q}$, so in the limit $N_{k+\hat{q}}^0\to0$ the BPS and non-BPS definitions of $\bar{\upnu}$ agree. Similarly, $N_{k}^0>0$ implies $N_k^-=0$ and thus in the limit $N_{k}^0\to0$ there is agreement at $\bar{\upnu}=0$. This is simply because the transition between the two cases corresponds to the rectangular aspect ratio of $Q_k$ crossing $1$.} In the BPS case, if $N_k^0>0$ the R-charge sequence terminates at the smallest R-charge $k$, and if $N_{k+\hat{q}}^0>0$ it terminates at the largest R-charge $k+\hat{q}$. In the non-BPS case the sequence need not terminate, unless it turns out $N_{k-\hat{q}}=0$ or $N_{k+2\hat{q}}=0$. In either case, we see that the $q_k$-multiplet spectrum is indeed described by \cref{eq:hcool} with the random matrix $M$ drawn from a BPS ensemble with $\bar{\upnu}>0$, or possibly an AZ ensemble if $\bar{\upnu}=0$. 

Let us note in passing that if the number of states in $q$-multiplets $2N_{q}^\leftrightarrow$ e.g. monotonically increases with the R-charge, then we generally find that $\bar{\upnu}_q = N_{q-\hat{q}}^\leftrightarrow$ for the $q$-multiplet. 
Hence if $N_{q}^\leftrightarrow$ grows say quadratically in $q$, we would similarly expect a monotonic growth of $\bar{\upnu}_q \propto q^2$.
Even if there are no BPS states involved, this intuitively means that the spectrum of the non-BPS states in the $q$-multiplet with R-charge $q-\hat{q}/2$ feels the absence of those states with the same R-charge which actually fell into the $(q-\hat{q})$-multiplet.

\subsection{\texorpdfstring{$\mathcal{N}=4$}{N=4}}

For $\mathcal{N}=4$ supersymmetry, the algebra contains four Hermitian supercharges $\{Q_I\}_{I=1}^4$ which transform as vectors of $SO(4)$ and give the Hamiltonian as
\begin{equation}
\label{eq:n4ham}
    \{Q_I,Q_J\}=2 H \delta_{IJ}.
\end{equation}
In this case the R-symmetry is $SU(2)$ and the Hilbert space can be organized into supermultiplets labelled by a half-integer spin $J$, which is what the spectral density in \cref{eq:rq4} describes. Once again, multiplets at different values of $J$ are expected to be statistically independent and describable by separate ensembles. Although the specific random matrix construction would now be technically more involved (see \cite{Turiaci:2023jfa}), the conclusion seems to be similar to that for the $\mathcal{N}=1,2$ cases: the spectrum of different supermultiplets would be captured by \cref{eq:hcool} with $M$ a random matrix. In this case, however, $M$ should always be drawn from a BPS ensemble where we expect $\bar{\upnu}>0$ to be related to the varying number of non-BPS states associated to supermultiplets of different spin.

\section{Random Matrix Tools}

This section reviews some standard tools for the analysis of random matrix integrals.

\subsection{The Resolvent}

The resolvent formalism is particularly suited to the study of spectral properties of operators, including but not limited to matrices. Given an operator $U$, the resolvent operator is
\begin{equation}
    \sR(z) \equiv (z\mathds{1}-U)^{-1}.
\end{equation}
Let $\rho$ denote the spectral density of $U$. The support $\supp\rho$ defines the spectrum of $U$. The resolvent $\sR$ is only well-defined on $\mathbb{C}\smallsetminus \supp\rho$, an open set called the resolvent set. Taking the trace of $\sR$ leads to\footnote{\,In the particular case in which the spectrum of $U$ is a discrete set, $\rho(u)=\sum_{i} \delta(u-u_i)$ and
the resolvent reduces to a sum over eigenvalues, $$
    R(u) = \sum_i \frac{1}{u-u_i}.
$$Note that if the spectrum is a finite set, $\rho$ integrates to the total number of eigenvalues.}
\begin{equation}
\label{eq:restr}
    R(z) \equiv \Tr \sR(z) = \int_{\mathbb{R}} du \, \frac{\rho(u)}{z-u}, \qquad z\in\mathbb{C}\smallsetminus \supp\rho.
\end{equation}
Henceforth, $R$ will be referred to simply as the resolvent. Assume the spectrum of $U$ is real, and define $\rho$ on all $\mathbb{R}$ by letting $\rho(u)=0$ for any $u\in\mathbb{R}\smallsetminus \supp\rho$. The integral representation in \cref{eq:restr} may be recognized as the Stieltjes transform of a non-negative measure of density $\rho$.\footnote{\,In solving matrix integrals, oftentimes resolvent methods involving $\sR$ are not needed, and it suffices to make use of the analyticity properties of the Stieltjes transform of the spectral density (see e.g. the derivation of the loop equations \cite{Migdal:1983qrz,Eynard:2004mh,Stanford:2019vob}). Nonetheless, following conventions, we will continue to refer to $R$ as the resolvent.}
Importantly, as defined away from $\supp\rho$, note that $R$ is an analytic function.
For a general point $z$ in the resolvent set, \cref{eq:restr} can be evaluated as a contour integral closing around the pole at $z$.\footnote{\,\label{ft:formalR}The contribution from the contour closing off the real axis may not vanish at infinity if $\rho$ is not sufficiently sparse or if it is in fact not a normalizable density (cf. the DSL). In such cases the definition of $R(z)$ is only formal.} Near the real axis, letting $z=x+i\epsilon$, separating real and imaginary parts of $R$, and taking the limit $\epsilon\to0$, one obtains a representation of the resolvent in terms of the Sokhotski–Plemelj theorem\footnote{\,Even when the definition of $R$ is only formal (see \cref{ft:formalR}), its imaginary part can still be evaluated if the limit $\epsilon\to0$ is taken before sending the contour to infinity. The real part of the resolvent may diverge either way. Its imaginary part diverges for any finite $\epsilon$, but is well-defined if the $\epsilon\to0$ limit is taken first.}
\begin{equation}
\label{eq:Rplem}
    \lim_{\epsilon\to0^+} R(x\pm i\epsilon) = \fint_\mathbb{R} du \, \frac{\rho(u)}{x-u} \mp i \pi \rho(x),
\end{equation}
where $\fint$ denotes a principal value integral. The inverse to the Stieltjes transform immediately follows from this expression. Explicitly, the spectral density $\rho$ can be extracted from the imaginary part of \cref{eq:Rplem} or, equivalently, from the discontinuity of $R$ across the real axis,
\begin{equation}
\label{eq:resima2}
    \rho( x ) = \pm \frac{1}{\pi} \lim_{\epsilon\to0^+} \, \Im \, R(x \mp i\epsilon) = \frac{i}{2\pi} \lim_{\epsilon\to0^+} \, ( R(x+i\epsilon) - R(x-i\epsilon) ).
\end{equation}
It is in this way that the resolvent of an operator precisely captures its spectral density.

In the context of random matrix theory, the operators of interest are functions of the random matrix.
Oftentimes though, one is just interested in the spectral properties of the matrices in the ensemble, and thus studies the resolvent for the random matrix itself. However, more generally, $R$ stands for the resolvent of any matrix operator $U$, and $\rho$ is the spectral density of that operator. 
Since $U$ may be any non-trivial functional of the random matrix, $\rho$ may generally be very different from the spectral density of the random matrix.
For example, the construction of Hamiltonians for AZ ensembles in \cref{ssec:hamran} makes the matrix operator of interest quadratic in the actual random matrices, and the leading spectrum of these turns out to exhibit a hard edge which would not be present otherwise (see \cref{sssec:zoom,ssec:ortho,sec:linazens}).

Given that the resolvent of an operator encodes its spectral density and can be inverted to obtain it, both can be used to extract statistical properties of the ensemble. In particular, correlation functions of spectral densities and resolvents are related by
\begin{equation}
\label{eq:RRrhorho}
    \langle R(z_1) \cdots R(z_n) \rangle = \int_{\mathbb{R}} \left[ \prod_{i=1}^n \frac{du_i}{z_i - u_i} \right]  \langle \rho(u_1)\cdots  \rho(u_n) \rangle, \qquad \forall z_i\in\mathbb{C}\smallsetminus \supp\rho.
\end{equation}
These resolvent correlation functions are particularly useful for their role in the loop equations that allow for solving matrix integrals perturbatively in a $1/N$ expansion \cite{Stanford:2019vob}.\footnote{\,In analogy with functional integrals, the loop equations can be understood as Schwinger-Dyson equations for resolvent correlation functions $\langle R(z_1) \cdots R(z_n) \rangle$ under rigid shifts of the integrated eigenvalues. Namely, since the eigenvalues being integrated over are dummy variables, this correlator is left invariant under the change of variables $\lambda_{k} \to \lambda_{k} + \epsilon$ with constant $\epsilon$, for any eigenvalue $\lambda_k$. For an infinitesimal shift, demanding invariance to linear order in $\epsilon$ leads to the loop equations (see \cite{Stanford:2019vob} for details).} While we will not be using this machinery here, it will still be very important for us to study leading order results at large $N$. Expanding both sides of \cref{eq:RRrhorho} and taking the large-$N$ limit for $n=1$ gives
\begin{equation}
\label{eq:leadR0}
    R_*(z) = \int_{\mathbb{R}}  du \, \frac{\hat{\rho}_*(u)}{z - u}, \qquad z\in\mathbb{C}\smallsetminus\supp\hat{\rho}_*,
\end{equation}
where \cref{eq:rhoslim} is the leading average spectral density, unit-normalized by absorbing factors of $N$ on both sides.
These large-$N$ objects correspond precisely to the continuum saddle-point results studied in \cref{ssec:contl}. In particular, by \cref{eq:expcond}, the spectral density $\hat{\rho}_*$ for the Hamiltonians of interest for WD (energies linear in $x$) and AZ/BPS (energies quadratic in $x$) ensembles obeys the equilibrium condition
\begin{equation}
\label{eq:alleom}
    \fint_{\mathbb{R}} dy \frac{\hat{\rho}_*(y)}{x - y} - \frac{V_\upnu'(x)}{\upbeta} = 0, \qquad \forall x\in\supp\hat{\rho}_*.
\end{equation}
Using this, for the leading spectral density the Sokhotski–Plemelj representation of the resolvent in \cref{eq:Rplem} can be written
\begin{equation}
\label{eq:allRplemV}
    \lim_{\epsilon\to0^+} R_*(x\pm i\epsilon) = \frac{V_\upnu'(x)}{\upbeta} \mp i \pi \hat{\rho}_*( x ), \qquad x\in\supp\hat{\rho}_*.
\end{equation}
In turn, this provides a discontinuity relation for the resolvent in terms of the matrix potential,
\begin{equation}
    V_\upnu'(x) = \frac{\upbeta}{2} \lim_{\epsilon\to0^+} \left( R_*(x+ i\epsilon) + R_*(x- i\epsilon) \right) \qquad x\in\supp\hat{\rho}_*,
\end{equation}
These important relations allow one to formulate the study of the spectrum of matrix integrals as a Riemann-Hilbert problem, which we turn to in \cref{sec:gensol}.

\subsection{The Spectral Curve}
\label{ssec:speccur}

The analytic continuation of the left-hand side of \cref{eq:alleom} off the support of the spectral density $\hat{\rho}_*$ defines a particularly important object known as the spectral curve,
\begin{equation}
\label{eq:ydef}
    y(z) \equiv R_*(z) - \frac{V_\upnu'(z)}{\upbeta}, \qquad z\in\mathbb{C}\smallsetminus\supp\hat{\rho}_*.
\end{equation}
Recalling that \cref{eq:alleom} comes from \cref{eq:deffV0}, $y$ is related to the leading effective potential by
\begin{equation}
\label{eq:vdevy}
    \widehat{V}_*'(z) = -\upbeta y(z), \qquad z\in\mathbb{C}\smallsetminus\supp\hat{\rho}_*.
\end{equation}
Integrating gives the leading effective potential off the support of $\hat{\rho}_*$ as
\begin{equation}
\label{eq:vyd}
    \widehat{V}_*(z) = -\upbeta \int^{z} dw \, y(w), \qquad z \in \mathbb{C} \smallsetminus \supp\hat{\rho}_*,
\end{equation}
up to a constant which may be fixed by making $\widehat{V}_*(z)$ vanish as $z\to\supp\hat{\rho}_*$ (cf. \cref{eq:deffV0}).

The spectral curve is a useful object due to its analyticity properties. 
Since the matrix potential $V$ is typically analytic, it does not contribute any poles to $y$. In WD ensembles $\upnu=0$ so $V_\upnu$ is also analytic, and $y$ just inherits the branch cut of $R_*$ along the real support of $\hat{\rho}_*$. For the AZ/BPS ensembles we are actually interested in the spectrum of squared eigenvalues, which because the integral measure $dx=\frac{d(x^2)}{2x}$ induces a $1/x$ pole on $y(x)$ at the origin (cf. \cref{eq:murho2}). In addition, if $\upnu=\nu N \sim O(N)$ as in the BPS ensembles, a logarithmic divergence $\nu \log x$ survives the large-$N$ limit in $V_\upnu$, giving a total $1/x^2$ pole to $y(x)$. This does not happen for AZ ensembles since $\upnu\sim O(1)$ and this pole is ignored at infinite $N$.

Generally $V_\upnu$ does not have any branch cuts, so applying \cref{eq:allRplemV} to \cref{eq:ydef} we have
\begin{equation}
\label{eq:scrho}
    \lim_{\epsilon\to0^{\pm}} y(x + i\epsilon) = \mp i \pi \hat{\rho}_*(x), \qquad x\in\supp\hat{\rho}_*,
\end{equation}
i.e., the discontinuity of $y$ across the support of $\hat{\rho}_*$ precisely captures $\hat{\rho}_*$ (cf. \cref{eq:resima2}).  Because \cref{eq:scrho} corresponds to a square-root branch cut,\footnote{\,If $\supp\hat{\rho}_*$ consists of multiple connected components (cuts), $y$ will have as many branch cuts.} it follows from the above that for WD ensembles $y(x)^2$ defines an analytic function everywhere. For AZ ensembles we can still obtain an analytic function by considering the combination $y(x)^2 x$. As for BPS ensembles, there actually remains a $1/x$ pole at $x=0$ for $y(x)^2 x$. One could obviously consider instead $y(x)^2x^2$, but this turns out to not be necessary. The logarithmic divergence of $V_\upnu$ as $x\to0^+$ means the support of $\hat{\rho}_*$ cannot possibly extend all the way to $x=0$. Since this BPS pole is off the the support of $\hat{\rho}_*$, for BPS ensembles we will still be interested in the analyticity properties of $y(x)^2 x$.

In general, we see that $y(x)^2$ is an analytic function except for possible poles at $x=0$, and thus the spectral curve itself, $y(z)$, defines a double cover of the complex plane branched over the support of $\hat{\rho}_*$.
An important corollary of this is that the value $y$ attains off the support of $\hat{\rho}_*$ is given by the analytic continuation of $\hat{\rho}_*$ itself (cf. \cref{eq:rhotoy}). Correspondingly, by \cref{eq:vyd}, the effective potential for $\hat{\rho}_*$ is given by an integral of the analytic continuation of $\hat{\rho}_*$ itself.

Let us end this section by quoting the spectral curve for each of our ensembles of interest with a Gaussian potential. For WD ensembles, \cref{eq:wdscal} gives
\begin{equation}
\label{eq:wdy}
    y(z) \reprel{GWD}{=}  \frac{1}{\upbeta} \sqrt{z^2 - 2\upbeta}, \qquad z\in \mathbb{C} \smallsetminus [-\sqrt{2\upbeta},\sqrt{2\upbeta}],
\end{equation}
which indeed is analytic everywhere it is defined. For AZ ensembles, \cref{eq:singaz} gives
\begin{equation}
\label{eq:azy}
    y(z) \reprel{GAZ}{=} \frac{1}{2\upbeta} \sqrt{\frac{z - 4\upbeta}{z}}, \qquad z\in \mathbb{C} \smallsetminus [0,4\upbeta].
\end{equation}
Since $z$ is now referring to squared eigenvalues, this reproduces the expected $1/x$ pole at $x=0$. Finally, for BPS ensembles, \cref{eq:gaznu} gives
\begin{equation}
\label{eq:gaznuy}
    y(z) \reprel{\!GBPS\!}{=} \frac{\sqrt{(a_+ - z)(a_- - z)}}{2\upbeta z}, \qquad z\in\mathbb{C}\smallsetminus[a_-,a_+],
\end{equation}
with the cut endpoints given by \cref{eq:endpoints}. Once again, since $z$ refers to squared eigenvalues, this objects exhibits the expected $1/x^2$ pole at $x=0$. However, as pointed out above, this pole indeed lies off the $[a_-,a_+]$ support of $\hat{\rho}_*$.

\subsection{General Equilibrium Measure}
\label{sec:gensol}

This section shows how to solve \cref{eq:expcond} to obtain an explicit expression for the equilibrium measure $\hat{\rho}_*$ that minimizes the action $\widehat{I}$ in \cref{eq:contIm} for an arbitrary matrix potential $V$.

According to \cite{Muskhelishvili1977SingularIE,Pastur1996}, \cref{eq:alleom} is a singular integral equation which has a unique solution under fairly general assumptions which have been implicit throughout.\footnote{\,The assumptions for uniqueness are: the potential $V$ is real, bounded-below, and grows such that $V(x)>\upbeta\log|x|$ for $|x|\geq L$ for some $L<\infty$; $\hat{\rho}_*$ is a non-negative, unit-normalized density; $-\int dx\,dy\,\hat{\rho}_*(x)\hat{\rho}_*(y) \log|x-y|<\infty$; $\int dy\,\hat{\rho}_*(y) \log|x-y| - V(x)<\infty$. These conditions are either implicitly true by construction (e.g. the first one is required for the matrix integral to exist), or needed for the variational problem of minimizing $\widehat{I}$ to be well defined.}
Furthermore, the unique equilibrium measure $\hat{\rho}_*$ that attains the infimum of $\widehat{I}$ can be shown to be absolutely continuous and supported on a finite union of intervals,
\begin{equation}
    \label{eq:support}
    \Sigma\equiv \supp\hat{\rho}_* = \bigcup_{j=1}^q [a_j,b_j],    
\end{equation}
where $a_j,b_j\in\mathbb{R}$ are ordered such that $a_j<b_j<a_{j+1}<b_{j+1}$ for all $j=1,\dots q-1$\footnote{\,For more details, see Proposition 4.1.2 of \cite{harnad-2011}, Theorem 1.38 of \cite{DEIFT1998388}, or Theorem 11.2.1 of \cite{pastur2011eigenvalue}.}. By \cref{eq:scrho}, we know that across each of these intervals, the spectral curve $y$ has a sign-change discontinuity, which must take the form of a square-root branch cut. For the WD ensembles, this same branch cut structure can be reproduced by the square root of the following function:
\begin{equation}
\label{eq:sigmawd}
    \eta(z) \reprel{WD}{\equiv} \prod_{j=1}^q (z-a_j)(z-b_j).
\end{equation}
For the AZ ensembles, we know $y(x)^2$ additionally has a $1/x$ pole at the origin (cf. the discussion below \cref{eq:scrho}), so for these we instead define
\begin{equation}
\label{eq:sigmaaz}
    \eta(z) \reprel{AZ}{\equiv} \frac{1}{z^2} \prod_{j=1}^q (z-a_j)(z-b_j).
\end{equation}
In writing this expression, we are using the fact that these ensembles typically have $a_1=0$. Hence one needs a factor of $1/z^2$ in order for $\eta$ to develop the right $1/z$ pole. For BPS ensembles the potential diverges at $x=0$, thus effectively repelling eigenvalues away from the origin towards positive values. This effect forces the endpoint closest to $x=0$ to obey $a_1>0$ strictly.
As a result, to cancel the non-analyticities of $y$ one can still just use \cref{eq:sigmaaz}: a nonzero $a_1$ here implies the desired $1/z^2$ pole for $\eta$. As a result for BPS ensembles \cref{{eq:sigmaaz}} actually exhibits a $1/z^2$ pole, precisely matching the behavior of $y(z)$ as well. Hence \cref{eq:sigmaaz} successfully reproduces the pole structure of the spectral curve of both AZ and BPS ensembles.

These $\eta$ functions can thus be used to cancel out any non-analyticity of $y$. Defining the square root of $\eta$ with its branch cuts along $\Sigma$, the combination $y(z)/\sqrt{\eta(z)}$ no longer has branch cuts or poles anywhere, and is in particular analytic everywhere. The upshot is that for a contour $\mathcal{C}$ around $\Sigma$, the integral
\begin{equation}
\label{eq:analys}
    \frac{1}{2\pi i}\oint_\mathcal{C} \frac{dw}{w-z} \frac{y(w)}{\sqrt{\eta(w)}} = 0,
\end{equation}
for any point $z\in\mathbb{C}$ not enclosed by $\mathcal{C}$. By analyticity of the integrand, the contour $\mathcal{C}$ can be deformed to surround $\Sigma$ arbitrarily tightly, which means \cref{eq:analys} in fact holds for any $z\in\mathbb{C}\smallsetminus\Sigma$.
This result can be applied to \cref{eq:ydef} in order to eliminate the dependence on the spectral curve, and obtain a direct relation between the resolvent and the matrix potential. To do so, consider using Cauchy's integral formula to write
\begin{equation}
    \frac{R_*(z)}{\sqrt{\eta(z)}} = \frac{1}{2\pi i} \oint_{\gamma_z} \frac{dw}{w-z} \frac{R_*(w)}{\sqrt{\eta(w)}}, \qquad z\in\mathbb{C}\smallsetminus\Sigma,
\end{equation}
for $\gamma_z$ a counterclockwise contour around a small neighborhood of $z$. Since $R_*(z)$ clearly goes as $1/z$ at infinity according to \cref{eq:leadR0}, the integrand above is guaranteed to go to zero faster than $1/z$ at infinity.
Together with the analyticity of this expression away from $\Sigma$, this means we can deform $\gamma_z$ past infinity without picking up any contributions. Retracting the contour back from infinity we obtain a contour that goes around $\Sigma$ with clockwise orientation. This can be matched onto the contour $\mathcal{C}$ used in \cref{eq:analys}, leading to
\begin{equation}
    \frac{R_*(z)}{\sqrt{\eta(z)}} = - \frac{1}{2\pi i} \oint_{\mathcal{C}} \frac{dw}{w-z} \frac{R_*(w)}{\sqrt{\eta(w)}}, \qquad z\in\mathbb{C}\smallsetminus\Sigma,
\end{equation}
where we are sticking to $\mathcal{C}$ running counterclockwise. Substituting \cref{eq:ydef} for $R_*$ on the right-hand side above and using \cref{eq:analys}, we arrive at\footnote{\,This is referred to as a Tricomi relation in \cite{Eynard:2015aea}, and as a dispersion relation in \cite{Stanford:2019vob}.}
\begin{equation}
\label{eq:R0V}
    R_*(z) = - \frac{1}{\upbeta} \sqrt{\eta(z)} \, h(z), \quad h(z) \equiv \frac{1}{2\pi i } \oint_{\mathcal{C}} \frac{dw}{w-z} \frac{V_\upnu'(w)}{\sqrt{\eta(w)}}, \qquad z\in\mathbb{C}\smallsetminus\Sigma.
\end{equation}
Using \cref{eq:ydef} allows to express the spectral curve explicitly in terms of the matrix potential,
\begin{equation}
\label{eq:yVwi}
    y(z) = -\frac{V_\upnu'(z)}{\upbeta} - \frac{1}{\upbeta} \sqrt{\eta(z)} \, h(z), \qquad x \in \mathbb{C} \smallsetminus \Sigma.
\end{equation}
By \cref{eq:scrho}, to obtain $\hat{\rho}_*$ it just remains to understand the limit in which $z$ approaches $\Sigma$. The non-trivial aspect is the evaluation of the contour integral for $h$ in \cref{eq:R0V}.
Firstly, flattening the contour around $\Sigma$, one easily arrives at
\begin{equation}
\label{eq:hreal}
     h(z) = \frac{1}{\pi} \int_{\Sigma} \frac{dy}{y-z} \frac{V_\upnu'(y)}{\sqrt[+]{\eta(y)}}, \qquad z\in\mathbb{C}\smallsetminus\Sigma,
\end{equation}
where it has been convenient to introduce the following branch cut prescription:
\begin{equation}
\label{eq:bcutapp}
    \sqrt[+]{\eta(x)} \equiv \lim_{\epsilon\to0^+} \Im \sqrt{\eta(x+i\epsilon)}.
\end{equation}
Taking $z$ towards $\Sigma$, we obtain a Sokhotski–Plemelj representation of $h$ (cf. \cref{eq:Rplem},
\begin{equation}
\label{eq:hplem}
    \lim_{\epsilon\to0^+} h(x\pm i\epsilon) = \frac{1}{\pi} \fint_{\Sigma} dy \, \frac{1}{y-x} \frac{V_\upnu'(y)}{\sqrt[+]{\eta(y)}} \pm i \frac{V_\upnu'(x)}{\sqrt[+]{\eta(x)}}, \qquad x\in\Sigma.
\end{equation}
Equipped with this result, we can finally obtain an explicit expression for the equilibrium measure. Taking the limit of $z$ towards $\Sigma$ in \cref{eq:yVwi} and using \cref{eq:scrho},
\begin{equation}
\label{eq:rho0final}
    \hat{\rho}_*(x) = \frac{\sqrt[+]{\eta(x)}}{\upbeta \, \pi^2} \fint_\Sigma \frac{dy}{y-x} \frac{V_\upnu'(y)}{\sqrt[+]{\eta(y)}}, \qquad x\in \Sigma,
\end{equation}
where recall the support $\Sigma=\supp\hat{\rho}_*$ takes the form of \cref{eq:support}.\footnote{\,A common alternative way of rewriting this final result utilizes the identity
$$
    \fint_\Sigma \frac{dy}{\sqrt[+]{\eta(y)}}  \frac{1}{y-x}  = 0, \qquad x\in\Sigma,
$$
which can be used to turn \cref{eq:rho0final} into \cite{pastur2011eigenvalue}
$$
    \hat{\rho}_*(x) = \frac{1}{\upbeta \, \pi} \sqrt[+]{\eta(x)} \, Q(x), \quad Q(x) \equiv \frac{1}{\pi} \fint_\Sigma  \frac{dy}{\sqrt[+]{\eta(y)}} \frac{V'(x)-V'(y)}{x-y}, \qquad x\in \Sigma.
$$
}
For instance, one can easily verify the expressions for WD, AZ, and BPS ensembles with Gaussian potentials respectively in \cref{eq:wdscal,eq:singaz,eq:gaznu} by applying \cref{eq:rho0final} to the single-cut case $q=1$ using the appropriate function from \cref{eq:sigmawd,eq:sigmaaz}.
The determination of the location of the $a_j,b_j$ endpoints of the support of $\hat{\rho}_*$ from \cref{eq:support}, as well as the correct number of cuts $q$, however, has to be worked out separately. 

Essentially, what \cref{eq:rho0final} provides is a solution to the Euler-Lagrange equation for the matrix integral action, which is guaranteed to give a local minimum. The solution to the extremization problem that is unique is the infimum of this action, which requires finding the appropriate values of $q$ and the $a_j,b_j$ endpoints for a global minimum \cite{harnad-2011}.
For a general matrix potential $V$, one may wonder how to even determine the number of intervals $q$ the measure $\hat{\rho}_*$ will be supported on, and their specific endpoints. In general $q$ is not known, and one has to try different values.\footnote{\,A simple case is when the potential is a convex function, for which it is easy to show that $\Sigma$ consists of a single finite interval and thus $q=1$ (see e.g. Theorem 11.2.3 of \cite{pastur2011eigenvalue}).} Given a choice of $q$, one attempts to find the endpoints of the intervals as follows. Expanding \cref{eq:R0V} at large $|z|$, and using the fact that the resolvent $R(z)$ goes as $1/z$ at large $|z|$, one obtains the following relations,
\begin{equation}
\label{eq:endpcons}
    \fint_{\Sigma} dx \, \frac{x^l \, V'(x)}{\sqrt[+]{\eta(x)}} = \upbeta\pi \, \delta_{l\tilde{q}}, \qquad l=0,\dots,\tilde{q},
\end{equation}
where $\tilde{q}=q$ for the WD ensembles and $\tilde{q}=q-1$ for the AZ/BPS ones. This difference between ensembles comes from the fact that at large $|z|$ \cref{eq:sigmawd} gives $\eta(z) = z^q + O(|z|^{q-1})$ for the former, whereas \cref{eq:sigmaaz} gives $\eta(z) = z^{q-1} + O(|z|^{q-2})$ for the latter. Additionally, for the AZ ensembles we know that $a_1=0$ from non-negativity, and for the BPS ensembles $a_1>0$ implies \cref{eq:endpcons} also vanishes for $l=-1$ due to the enhanced pole. Altogether, this means that in all cases we have a system of $q+1$ equations for the $2q$ endpoints of the intervals that make up $\Sigma$. 

Another $q-1$ equations are needed to fully fix these. Consider the effective potential from \cref{eq:vyd} along the real axis. By the equations of motion we know $\widehat{V}_*$ must vanish at the $a_j,b_j$ endpoints of the support intervals. In its indefinite integral expression from \cref{eq:vdevy} this means $\widehat{V}_*$ must attain the same constant value on all of them. This gives us the desired additional $q-1$ equations as
\begin{equation}
\label{eq:endp2}
    0 = \widehat{V}_*(b_j) - \widehat{V}_*(a_{j+1}) = \upbeta \int_{b_j}^{a_{j+1}} dx \, y(x), \qquad j=1,\dots,q-1.
\end{equation}
Altogether we obtained $2q$ conditions for the $2q$ variables determining the endpoints of $\supp\hat{\rho}_*$. These can be shown to be linearly independent (cf. remark 11.2.6 from \cite{pastur2011eigenvalue}), thus completing the characterization of the unique equilibrium measure $\hat{\rho}_*$ for a general matrix integral.

\subsection{Orthogonal Polynomials}
\label{ssec:ortho}

A standard treatment in random matrix theory proceeds by factoring the matrix potential term into the eigenvalue determinant and then performing elementary row operations so as to build a suitably useful sequence of orthogonal polynomials. The upshot of such manipulations is a repacking of the joint PDF of eigenvalues into determinants,
\begin{equation}
\label{eq:kernel}
    p_n(x_1,\dots,x_n) = \frac{1}{N!} \det \{ K_N(x_i,x_j)\}_{i,j=1}^N.
\end{equation}
Here $K_N$ is a symmetric integral kernel given by
\begin{equation}
\label{eq:cdkernel}
    K_N(x,y) = \sum_{k=0}^{N-1} \varphi_k(x)\varphi_k(y),    
\end{equation}
with $\{\varphi_k\}_{k=0}^{N-1}$ a family of polynomials orthonormalizing the matrix integral,
\begin{equation}
    \int dx \, \varphi_j(x) \varphi_k(x) = \delta_{jk}.
\end{equation}
This kernel has the following important properties:
\begin{equation}
\label{eq:speck}
    \int du \, K_N(x,u) \, K_N(u,y) = K_N(x,y), \qquad K_N(x,x) = \rho(x).
\end{equation}
The Christoffel-Darboux formula reduces the series in \cref{eq:cdkernel} to
\begin{equation}
    K_N(x,y) = c_N \frac{\varphi_N(x)\varphi_{N-1}(y)-\varphi_N(y)\varphi_{N-1}(x)}{x-y},
\end{equation}
where the $N$-dependent constant can be extracted from
\begin{equation}
\label{eq:cdcn}
    c_N \equiv \lim_{x\to\infty} \frac{x\, \varphi_{N-1}(x)}{\varphi_{N}(x)}.
\end{equation}

For the GUE, the matrix integral involves the measure in \cref{eq:genjac} with $\upbeta=2$ and $\upnu=0$, and the matrix potential is $V(x)=\frac{1}{2}x^2$. The $\{\varphi_k\}$ polynomials in this case are built so as to orthonormalize the sequence $\{x^k\, e^{-\frac{1}{2} x^2} \}$ over $\mathbb{R}$. This is accomplished by 
\begin{equation}
    \varphi_k(x) = \frac{1}{(2 \pi)^{1/4} \sqrt{k!}} \, h_k(x)\, e^{-\frac{x^2}{4}}, \qquad h_k(x) \equiv (-1)^k e^{\frac{x^2}{2}} \frac{d^k}{dx^k} e^{-\frac{x^2}{2}},
\end{equation}
where $h_k$ are standard Hermite polynomials. For these the constant in \cref{eq:cdcn} is $c_N=\sqrt{N}$, and the spectral density can be readily computed using \cref{eq:speck}. The Wigner semicircle follows from the following scaling limit:
\begin{equation}
\label{eq:classgue}
    \hat{\rho}_*(x) \equiv \lim_{N\to\infty} \frac{1}{\sqrt{N}} \rho\left(\sqrt{N} x\right) = \frac{1}{2\pi} \sqrt{4-x^2}, \qquad -2<x<2.
\end{equation}
In this paper, this is actually the canonical scaling that our conventional factor of $N$ in front of the matrix potential in \cref{eq:matrixZ} implements.\footnote{\,Correspondingly, note how in this section we have chosen the $\{\varphi_k\}$ polynomials to orthonormalize the sequence $\{x^k\, e^{-\frac{1}{2} x^2} \}$, not the sequence $\{x^k\, e^{-\frac{N}{2} x^2} \}$.} The Airy spectral density is obtained by the following scaling limit which zooms on the lower edge of the spectrum \cite{Tracy:1992kc}:
\begin{equation}
    \hat{\rho}^{\infty}(x) \equiv \lim_{N\to\infty} \frac{1}{N^{1/6}} \rho\left( -2\sqrt{N} + \frac{x}{N^{1/6}}\right) = \text{Ai}'(-x)^2 + x \, \text{Ai}(-x)^2.
\end{equation}
The exact spectral density of the Airy model is defined to be
\begin{equation}
    \hat{\rho}^{\smalltext{Airy}}(x) \equiv e^{-S_0/3} \hat{\rho}^{\infty}(e^{2S_0/3}x),
\end{equation}
where $e^{S_0}$ captures the eigenvalue density. Its leading spectral density at large $S_0$ is
\begin{equation}
    \hat{\rho}_*^{\smalltext{Airy}}(x) \equiv \lim_{S_0\to\infty} \hat{\rho}^{\smalltext{Airy}}(x) = \frac{\sqrt{x}}{\pi}, \qquad x>0,
\end{equation}
thereby reducing to the behavior of \cref{eq:classgue} near the $x=-2$ edge of the spectrum. Because $\hat{\rho}_*^{\smalltext{Airy}}$ goes to zero at its edge, the Airy spectrum is said to exhibit a soft edge.

A similar exercise can be done for the AZ/BPS ensembles. In particular, consider the measure in \cref{eq:genjac} applied to squared eigenvalues for $\upbeta=2$ and any $\upnu$. With a Gaussian potential the squared eigenvalues experience a linear potential of the form $V(x)=\frac{1}{2}x$. For this integral, the $\{\varphi_k\}$ polynomials are built so as to orthonormalize the sequence $\{x^k\, x^\upnu e^{-\frac{1}{2} x} \}$ over $\mathbb{R}^+$. This is accomplished by 
\begin{equation}
    \varphi_k(x) = \sqrt{\frac{k!}{2^{\nu +1} (k+\nu )!}} \, L_k^{\nu }(x/2) \, x^{\upnu /2} \, e^{-\frac{x}{4}}, \qquad
    L_k^{\nu }(x) \equiv \frac{x^{-\upnu} e^{x}}{k!} \frac{d^k}{dx^k} \left(x^{k+\upnu} e^{-x}\right),
\end{equation}
where $L_k^\upnu$ are generalized Laguerre polynomials. For this reason, these ensembles are often referred to as Laguerre unitary ensembles. These polynomials give $c_N = -\sqrt{2N(N+\upnu)}$ for the constant in \cref{eq:cdcn}. A singular Mar\v{c}enko-Pastur law is obtained in the following scaling limit:
\begin{equation}
\label{eq:azscal}
    \hat{\rho}_*(x) = \lim_{N\to\infty} \rho(Nx) = \frac{1}{4 \pi } \sqrt{\frac{8-x}{x}}, \qquad 0<x<8.
\end{equation}
The divergent behavior of this distribution as $x\to0$ is commonly referred as a hard edge, in contrast with the soft edge of \cref{eq:classgue}. Zooming onto this hard edge gives rise to Bessel spectral densities \cite{Tracy:1993xj},
\begin{equation}
     \hat{\rho}_{\upnu}^{\infty}(x) \equiv \lim_{N\to\infty} \frac{1}{N} \rho\left( \frac{x}{2N} \right) = \frac{1}{2} \left(J_{\upnu }\left(\sqrt{x}\right){}^2-J_{\upnu +1}\left(\sqrt{x}\right) J_{\upnu -1}\left(\sqrt{x}\right)\right).
\end{equation}
The exact spectral density of the Bessel model is defined to be
\begin{equation}
    \hat{\rho}_{\upnu}^{\smalltext{Bessel}}(x) \equiv e^{S_0} \hat{\rho}_{\upnu}^{\infty}(e^{2S_0} x),
\end{equation}
whose behavior at large $S_0$ yields the leading spectral density
\begin{equation}
    \hat{\rho}_*^{\smalltext{Bessel}}(x) = \lim_{S_0\to\infty} \hat{\rho}_{\upnu}^{\smalltext{Bessel}}(x) = \frac{1}{\pi\sqrt{x}}, \qquad x>0.
\end{equation}
An alternative standard scaling limit of the Laguerre ensembles involves making the parameter $\upnu$ extensive in $N$ by letting $\upnu\equiv \nu N$ and keeping $\nu \sim O(1)$. Doing so, the analogous limit to \cref{eq:azscal} leads in this case to a non-singular Mar\v{c}enko-Pastur law,
\begin{equation}
    {\hat{\rho}_\nu}{}_*(x) = \lim_{N\to\infty} \rho_{\nu}(Nx) = \frac{\sqrt{(a_+-x)(x-a_-)}}{4 \pi  x}, \qquad a_-<x<a_+,
\end{equation}
where the endpoints of the leading spectrum are given by
\begin{equation}
    a_\pm \equiv 2 \left(1-\sqrt{1+\nu}\right)^2.
\end{equation}
The previous case in \cref{eq:azscal} is recovered in the limit $\nu\to0$ in which $a_-\to0$ and the pole at $x=0$ becomes again accessible. However, for finite $\nu>0$, the distribution ${\hat{\rho}_\nu}{}_*$ has now a characteristic gap $(0,a_-)$ with no support. Additionally, $x=a_-$ now becomes a soft edge of the same $\sqrt{x}$ form as in the Airy model.

\subsection{WD/AZ Ensemble Relation}
\label{sec:linazens}

In \cref{ssec:ensemmess} we replaced the natural AZ eigenvalues $\lambda_i$ by their squares $\omega_i\equiv\lambda_i^2$ at the level of the measure. This allowed us to obtain a general measure in \cref{eq:genjac} whose form captures all WD and AZ ensembles together. However, we had to keep in mind that for the WD ensembles the eigenvalues $\lambda_i$ could extend over all $\mathbb{R}$, whereas for AZ ensembles the new $\omega_i$ variables had to be restricted to being non-negative. Additionally, at the level of the potential this meant that e.g. a quadratic term in $\lambda_i$ would just give a linear term in $\omega_i$. Here we make this process more transparent by working out the AZ matrix integral in its original eigenvalues $\lambda_i$ which extend over all $\mathbb{R}$ and are subject to the original AZ measure in \cref{eq:AZmeasure}.

Firstly, let us quote that for the WD ensembles the equilibrium condition is simply
\begin{equation}
\label{eq:wdexp}
    \fint\displaylimits_{\mathbb{R}} dy \, \frac{\hat{\rho}(y)}{x-y} - \frac{V'(x)}{\upbeta} \reprel{WD}{=} 0, \qquad x\in\supp\hat{\rho}_*,
\end{equation}
where we just set $\upnu=0$ in \cref{eq:deffV0,eq:dveffr}, which followed from extremizing the effective potential in \cref{eq:effwd}. On the other hand, using \cref{eq:AZmeasure}, the continuum effective potential for the AZ ensemble reads
\begin{equation}
\label{eq:azveff}
    \widehat{V}[\hat{\rho};x] \reprel{AZ}{=} V_\upalpha(x) - \int_{\mathbb{R}} dy \, \hat{\rho}(y) \log \, |x^2-y^2|^\upbeta,
\end{equation}
where $V_\upalpha$ uses the notation from \cref{eq:redef}, and thus \cref{eq:deffV0} gives the equilibrium condition
\begin{equation}
\label{eq:azeom}
    \fint_{\mathbb{R}} dy \frac{\hat{\rho}_*(y)}{x^2 - y^2} - \frac{V_\upalpha'(x)}{2x\upbeta} \reprel{AZ}{=} 0, \qquad \forall x\in\supp\hat{\rho}_*.
\end{equation}
This is more closely related to \cref{eq:wdexp} than it may seem at first glance. Using the identity
\begin{equation}
    \frac{1}{x^2-y^2} = \frac{1}{2x} \left( \frac{1}{x-y} + \frac{1}{x+y} \right)
\end{equation}
one easily rearranges \cref{eq:azeom} into
\begin{equation}
\label{eq:azeom2}
    \fint_{\mathbb{R}} dy \left( \frac{\hat{\rho}_*(y)+\hat{\rho}_*(-y)}{x-y} \right)
    - \frac{V_\upalpha'(x)}{\upbeta} \reprel{AZ}{=} 0, \qquad \forall x\in\supp\hat{\rho}_*.
\end{equation}
Comparing this to \cref{eq:wdexp}, the following result is immediate: if $\hat{\rho}^{\smalltext{WD}\scalebox{0.5}{$(\upbeta)$}}_*$ is the leading spectral density for a WD ensemble with parameter $\upbeta$ and matrix potential $V_\upalpha$, then the leading spectral density for the AZ ensemble with the same potential and parameter $\upbeta$ must obey\footnote{\,The factors of $2$ come from multiplying \cref{eq:azeom2} through by $1/2$ for a consistent normalization of all densities.}
\begin{equation}
\label{eq:azwdrel}
    \hat{\rho}^{\smalltext{AZ}\scalebox{0.5}{$(\upbeta)$}}_*(x) + \hat{\rho}^{\smalltext{AZ}\scalebox{0.5}{$(\upbeta)$}}_*(-x) = 2\hat{\rho}^{\smalltext{WD}\scalebox{0.5}{$(2\upbeta)$}}_*(x).
\end{equation}
For instance, if $\upalpha\sim O(1)$ in the large-$N$ limit, then $V_\upalpha\to V$ and e.g. for a Gaussian potential $V(x)=\frac{1}{2}x^2$ the solution to \cref{eq:azeom} is
\begin{equation}
\label{eq:azrho0Gauss}
    \hat{\rho}_*(x) \reprel{GAZ}{=}  \frac{1}{2\pi \upbeta} \sqrt{4\upbeta - x^2}, \qquad - 2\sqrt{\upbeta} < x < 2\sqrt{\upbeta},
\end{equation}
which recalling \cref{eq:wdscal} verifies our finding in \cref{eq:azwdrel}. In this sense, we thus see that for general matrix potentials, WD and AZ/BPS ensembles are in fact rather similar at leading order, particularly if $\upalpha/N\to0$ in the large-$N$ limit.
What actually makes them quite different in the context of random Hamiltonians, whether or not $\upalpha$ is extensive in $N$, is the fact that for AZ ensembles we are interested in the spectrum of squared eigenvalues (see \cref{ssec:hamran}).

Denote the leading spectral density of the Hamiltonian $H$ for an AZ ensemble by $\hat{\mu}_*$, and that of the random AZ matrix by $\hat{\rho}_*$, as above. By the construction of $H$ in \cref{ssec:hamran}, the eigenvalues of the former $\chi$ are related to those of the latter $x$ simply by $\chi=x^2$.
Clearly $\hat{\mu}_*(\chi) = 0$ for any $\chi<0$. To get an expression for $\mu$ on the non-negative real axis, we demand that for any even function $f$, the expectation values computed with respect to either density agree. In particular, this must be true for a constant function,
\begin{equation}
\int_{-\infty}^\infty dx \, \hat{\rho}_*(x) = \int_0^{\infty} \frac{d(x^2)}{2x} \, (\hat{\rho}_*(x)+\hat{\rho}_*(-x)),
\end{equation}
from which one reads off\footnote{\,If the matrix potential $V$ happens to be an even function, then by \cref{eq:azeom} it follows that $\hat{\rho}_*(-x)=\hat{\rho}_*(x)$, which simplifies this to $\hat{\mu}_*(x^2) = \frac{\hat{\rho}_*(x)}{x}$.}
\begin{equation}
\label{eq:murho2}
    \hat{\mu}_*(x^2) = \frac{\hat{\rho}_*(x)+\hat{\rho}_*(-x)}{2x}, \qquad x>0.
\end{equation}
Hence, on general grounds, the spectral density $\hat{\mu}_*(\chi)$ of the square of an operator will exhibit a $1/\sqrt{\chi}$ singularity near $\chi=0^+$, and a discontinuity at the origin since $\hat{\mu}_*(\chi)=0$ for $\chi<0$. For e.g. the Gaussian case in \cref{eq:azrho0Gauss}, the result in \cref{eq:murho2} yields (cf. \cref{eq:singaz})
\begin{equation}
\label{eq:sqmu2}
    \hat{\mu}_*(\chi) \reprel{GAZ}{=} \frac{1}{2\pi \upbeta} \sqrt{\frac{4\upbeta - \chi}{\chi}}, \qquad 0<\chi < 4\upbeta.
\end{equation}
The procedure above shows rather explicitly how the leading spectral densities of WD and AZ/BPS ensembles are directly related in terms of their original eigenvalue variables. This is particularly relevant for all the classical AZ ensembles where $\alpha$ is a fixed constant and therefore the potential $V_\upalpha$ appearing in \cref{eq:azeom} at large $N$ reduces to the same potential $V$ that would appear in the WD equilibrium condition. Additionally, we see that for general AZ ensembles, the hard edge that arises when $\upalpha/N\to0$ at large $N$ is essentially a consequence of squaring eigenvalues, rather than an intrinsic feature of the AZ ensembles. In other words, the original eigenvalues of an AZ matrix integral exhibit no hard edges, just like a WD one.

\addcontentsline{toc}{section}{References}
\bibliographystyle{JHEP}
\bibliography{references.bib}

\end{document}